\documentclass[10pt]{iopart}

\usepackage[bottom=1in,left=1in,right=1in,top=1in]{geometry}
\usepackage{graphicx}       
    \graphicspath{ {Figs/} }
    
\usepackage{graphicx,xcolor,mathbbol,amsthm,amsfonts,amssymb,bm,braket,footmisc,multirow,latexsym,times,dsfont}

\usepackage[colorlinks]{hyperref}
\hypersetup{
 colorlinks=true,
 citecolor=blue,
 linkcolor=blue,
 urlcolor=cyan
 }
 \usepackage{outlines}
\usepackage[utf8]{inputenc}
\usepackage{bbm}
\usepackage{bm}
\usepackage[caption = false]{subfig}

\color[rgb]{0.,0.,0.} 
\usepackage{enumitem}  
\usepackage{mathdots}
\usepackage{tikz}
\usetikzlibrary{mindmap,backgrounds}
\usepackage[normalem]{ulem}
\usepackage{float,soul}
\usepackage{hyperref}
\usepackage{tikz}
\usetikzlibrary{mindmap}
\usetikzlibrary{trees}
\usepackage{comment}
\expandafter\let\csname equation*\endcsname\relax

\expandafter\let\csname endequation*\endcsname\relax

\usepackage{amsmath}
\usepackage{tablefootnote}
\usepackage{pifont}
\usepackage{amsthm}
\usepackage{yhmath}
\usepackage{mathrsfs}
\usetikzlibrary{decorations.pathreplacing}

\tikzstyle{transition}=[circle,draw=black!50,fill=black!20,thick,
inner sep=0pt,minimum size=4mm]

\usepackage{braket}
\usepackage{multirow}
\usepackage{cancel}
\usepackage{mathtools}

\usepackage{hyperref}
\renewcommand{\P}{\Phi}

\renewcommand{\inf}{\infty}

\newcommand{\mb}{\mathbf}

\newcommand{\norm}[1]{
\lVert#1
\rVert}

\newtheorem{proposition}{Proposition}
\newtheorem{theorem}{Theorem}
\newtheorem*{theorem*}{Theorem}
\newtheorem{lemma}{Lemma}

\begin{document}

\title[Criticality in quadratic bosonic Hamiltonians]{Quantum criticality beyond thermodynamic stability}

\author{Mariam Ughrelidze$^1$, Vincent P. Flynn$^{2,1}$, Emilio Cobanera$^{3,1}$,
and Lorenza Viola$^{1,*}$}

\address{$^1$Department of Physics and Astronomy, Dartmouth College, 6127 Wilder Laboratory, Hanover, NH 03755, USA}
\address{$^2$Department of Physics, Boston College, 140 Commonwealth Avenue, Chestnut Hill, MA 02467, USA}
\address{$^3$Department of Physics, SUNY Polytechnic Institute, 100 Seymour Avenue, Utica, NY 13502, USA}
\medskip

\address{$*$ Corresponding author {\em Lorenza.Viola@Dartmouth.edu}}

\begin{abstract}
For a many-body system in equilibrium, described by a thermodynamically stable Hamiltonian, quantum criticality is associated with structural changes of the many-body ground state. However, there exist physically relevant models, notably, certain quadratic bosonic Hamiltonians (QBHs), which fail to have a ground state. From a dynamical system-theory standpoint,  QBHs can be either dynamically stable or unstable. We show that criticality is a meaningful concept for the entire class of QBHs that are {\em dynamically stable} or at the boundary of instability -- {\em regardless} of thermodynamic stability -- and that the key state for such QBHs is a naturally and unambiguously defined {\em quasiparticle vacuum}. This state is Gaussian, and coincides with the ground state if the QBH is thermodynamically stable. We identify a relevant spectral gap, the {\em Krein gap}, associated to the minimal spectral separation between creation and annihilation operators, and show that the quasiparticle vacuum is unique when the Krein gap is positive. In addition, we prove that, for dynamically stable QBHs with finite-range couplings, correlations are {\em exponentially bounded} unless the Krein gap closes -- which is associated with one of two kinds of spectral degeneracies: an exceptional point (physically, a free-particle-like collective mode) or a Krein collision (often a bosonic zero mode). As a consequence, long-range correlations can ensue in the quasiparticle vacuum. We conclude that (i) the Krein gap takes the role of the spectral gap for dynamically stable QBHs, and (ii) the boundary of dynamical stability and criticality (associated to exceptional points) or  multicriticality (associated to Krein collisions) are one and the same. Alongside long-range correlations, we find that bosonic critical behavior beyond thermodynamic stability is witnessed by the scaling of the entanglement entropy and other indicators of equilibrium criticality from information geometry. Hence, our framework opens the door to investigating {\em all} dynamically stable QBHs through the lens of critical phenomena, including thermodynamically unstable models from photonics, opto-mechanics, cavity-QED, and magnonics. 
\end{abstract}


\maketitle


\section{Introduction}

Understanding quantum critical phenomena is a cornerstone of contemporary condensed-matter physics, with implications ranging from fundamental advances in discovering and classifying new phases of matter to practical leverage for next-generation quantum technologies \cite{Sachdev2023}. For closed systems in equilibrium, quantum criticality is traditionally associated with a structural change of the many-body ground state (GS) in the thermodynamic limit, which takes place as one or more non-thermal control parameters in the Hamiltonian are tuned at special, critical values. The underlying feature that makes a critical GS so remarkable is that characteristic length scales associated with pairwise correlation functions of local observables diverge at these critical points, and scale invariance emerges in spite of all microscopic couplings being short-ranged. Fueled by  progress in the control capabilities over several experimental platforms, significant effort has recently been made in extending the study of criticality to open (Markovian) quantum systems, leading to the identification of non-analytic behavior and diverging correlations in both mixed and pure (non-equilibrium) steady states \cite{NoiseDrivenCriticality,BarthelCriticality,Sieberer}. 

The extension we advocate for in this work still focuses on closed quantum-dynamics, but zooms in onto one of the archetypal and most widely used classes of many-body quantum systems: namely, systems of ``independent'' (``free'') bosons, described by {\em quadratic bosonic Hamiltonians} (QBHs) \cite{Derezinski,Decon},    {which commonly arise as a mean-field or linearized description of an interacting system}. {Encountered in modeling phenomena across disparate fields -- including, phonon structure \cite{Blaizot}, parametric amplification and topological photonics \cite{Ozawa,WangClerk}, magnon amplification and lasing \cite{Antimagnons}, cosmological perturbations \cite{Kamal}, to name a few -- QBHs are also extensively considered in the context of continuous-variable quantum information \cite{CVQI,GQI}. Several rigorous results have been established on the scaling of the two-point correlation functions and, by extension, the {\em entanglement entropy} (EE) in the GS, thanks to the fact that, for these Gaussian models, a covariance-matrix formalism can be leveraged. In {\em harmonic lattices}, with couplings only in the position degrees of freedom, the closing of the many-body energy gap is associated to the onset of long-range correlations \cite{Cramer2006}. Consistently, the EE has been shown to be inversely proportional to the many-body energy gap and to scale as the boundary of the system in non-critical regimes \cite{Audenart2002,EisertEntropyEntanglement,AreaLawBosonic,Cramer2006}. This type of area law scaling is of interest from the perspective of efficient simulation of quantum states \cite{EisertReview}. Beyond harmonic lattices, the magnitude of the many-body energy gap was found to no longer be an accurate indicator of proximity to criticality. Instead, the role of the energy gap was found to be assumed by the energy gap of an auxiliary, ``symmetrized" Hamiltonian \cite{Wolf2006}, with the opening of the latter gap guaranteeing exponentially decaying correlations. 

While the focus on the GS of QBHs is, in a way, natural, we argue it is also too restrictive, for at least two main reasons. First, many systems of current theoretical and experimental interest are well-described by effective QBHs that lack a GS.       {Such QBHs are called \textit{thermodynamically unstable}. Thermodynamic stability, i.e., the existence of a finite lower (or upper) bound on the energy, and hence of a GS, is to be sharply distinguished from \textit{dynamical stability}, which instead demands that the expectation of arbitrary observables remain bounded at all times, for any initial state. For a QBH, the latter holds if and only if it can be brought to diagonal form by a Bogoliubov transformation, with real quasiparticle energies. Crucially, these two stability notions stability are independent of each other. Spin-wave excitations, or magnons, associated with nonequilibrium configurations of magnetic systems provide a notable example of a thermodynamically unstable system. These quasiparticles can carry negative energy despite the overall dynamical stability of the system \cite{Antimagnons, SSF, DMM}, a scenario that has now
been experimentally observed \cite{MagnonExpt}}. In the context of cavity QED, the linearized Hamiltonian obtained from performing the Holstein-Primakoff transformation on the Hamiltonian for an ensemble of atoms coupled to a cavity mode may lack a lower bound on its energy \cite{Wiersig}. 
   {Likewise, for Bose-Einstein condensates, the linearized Gross-Pitaevskii equation distinguishes \emph{Landau} or \emph{energetic} (thermodynamic, in our language) from dynamical instability. A condensate Bloch wave can be energetically unstable, while remaining dynamically stable, a separation that has been resolved experimentally \cite{BECNakamura,BECSarlo,BECExpt}.} Yet another example is the so-called bosonic Kitaev chain \cite{ClerkBKC}: despite being dynamically stable for any finite-size realization, it is       {always} thermodynamically unstable. Besides providing an early example of a topological directional amplifier, this QBH displays other interesting features (including distinctive EE scaling \cite{EntanglementPTClerk,DynamicalMetastability} and hints of non-trivial topology \cite{Decon,Bosoranas,EmilioQBH}), and has been experimentally realized in an optomechanical setting \cite{BKCOptomechanical}. Second, a QBH can be thermodynamically unstable and yet still display dynamical stability phases separated by a critical phase boundary. This is because a dynamically stable QBH away from a stability phase boundary is associated to a unique {\em quasiparticle vacuum} (QPV), the state annihilated by all the (quasiparticle) normal modes of the Hamiltonian. 

In this work, we propose to extend the range of applicability of the theory of critical phenomena to the entire class of dynamically stable QBHs (see Fig.\,\ref{fig: StabilityDiagram}), by showing how -- {\em irrespective of  thermodynamic stability} -- the QPV furnishes the relevant state that can become critical (and even multicritical) at a dynamical stability phase transition. While, for thermodynamically stable QBHs, the QPV coincides with the GS, the class of dynamically stable QBHs is larger than the class of thermodynamically stable ones (which it almost completely contains); descriptions of many experimental bosonic systems naturally yield dynamically stable QBHs with a QPV, but not necessarily a GS. The resulting notion of ``generalized criticality" is somewhat intermediate between traditional equilibrium and non-equilibrium statistical mechanics. On the one hand, it is strictly outside of the scope of standard equilibrium statistical mechanics because there is no GS (or Gibbs state). 
On the other hand, the kind of criticality we describe is not typical of a non-equilibrium continuous phase transition either. It is a property of a distinguished many-body state that is generically in the middle of the many-body energy spectrum (unless the QBH happens to be thermodynamically stable) but is {\em pure}, and can always be chosen to be zero-mean and Gaussian. It is also a property that can be characterized in detail solely in terms of {\em equal-time} correlations in the thermodynamic limit. 

\begin{figure}[t]
\centering
\hspace*{25mm}\includegraphics[width=0.6\columnwidth]{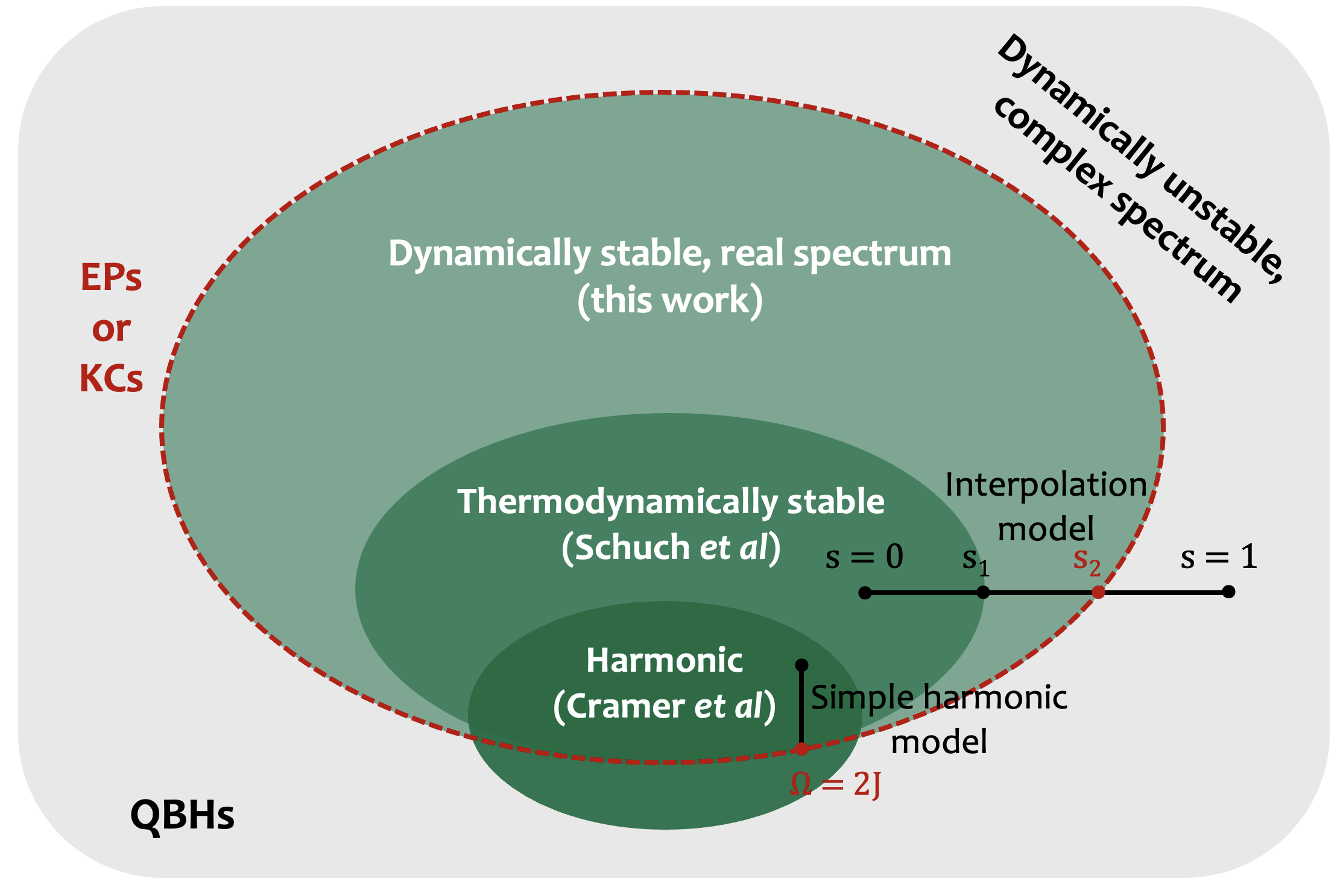}
\caption{Pictorial rendering of the landscape of QBHs. Our theory of criticality applies to the entire class of dynamically stable QBHs. The boundary between dynamically stable and unstable QBHs (red dashed line) boasts a mixture of both possible types of bosonic spectral singularities: dynamically unstable QBHs on the boundary host EPs, while the stable ones exhibit KCs. All thermodynamically stable QBHs are dynamically stable or at the boundary of stability; they are the subject of Ref.\,\cite{Wolf2006}. The special class of harmonic QBHs is the subject of Ref.\,\cite{Cramer2006}.  The simple harmonic chain is a textbook example of several physical concepts and fits our theory well. The interpolation model serves as the main illustrative example of our general theory. As a function of its parameter $s$, the model leaves the class of thermodynamically stable QBHs, since its many-body energy gap closes at $s_1$, with a finite correlation length throughout. As we show, for a QBH, criticality is associated to the transition between dynamical stability regimes; the correlation length of the interpolation model diverges at $s_2$, where an EP develops. In all cases, the relevant state is the QPV.} 
\label{fig: StabilityDiagram}
\end{figure}

Before getting to a more technical discussion, it is worth highlighting some crucial conceptual advances that derive from investigating the class of dynamically stable QBHs -- and, with that, QPVs -- as a whole. While a GS is a special kind of QPV (the only kind that has been subjected to systematic investigation so far) it is not {\em a priori} clear how to fit the body of knowledge about GS criticality inside a more general and comprehensive theory. In this paper, we show that one can take a significant step in building such a theory, by introducing a new kind of many-body gap for dynamically stable QBHs that we call the {\em Krein gap}. Loosely speaking, the Krein gap is the minimal spectral separation between creation and annihilation operators, and can only vanish if the parameters are tuned at one of the two possible spectral singularities that a          {QBH can host by virtue of its {\em effective non-Hermiticity}: although such a Hamiltonian is Hermitian and, as such, generates unitary dynamics at the many-body level, the dynamical matrix that governs the equations of motion is generically non-Hermitian (in fact, non-normal) \cite{Blaizot, ClerkBKC, Decon}. As proved in \cite{Decon}, loss of dynamical stability can then occur through either an {\em exceptional point} (EP) or a {\em Krein collision} (KC)}. Building on this new notion of a spectral gap, we will show that many results about GSs             {of QBHs} can be recovered as corollaries to general results for the broader class of QPVs.

Viewed from a different angle, a major limitation of thermodynamically stable QBHs is their strict inability to support topologically-mandated boundary zero modes, isolated or as part of a surface band. This result is one of the no-go theorems proved in Ref.\,\cite{Squaring}. This is, in turn, consistent with another no-go theorem, also proved therein: any two thermodynamically stable (in particular, gapped) QBHs in the same symmetry class (whatever the many-body classifying symmetries may be) can be deformed into each other adiabatically, without closing the many-body energy gap or leaving the symmetry class. These results suggest themselves as two sides of the same coin: thermodynamically stable QBHs do not support symmetry-protected topological quantum phases. However, it is well documented by now that dynamically but not thermodynamically stable QBHs can display topologically-mandated boundary modes \cite{Peano2015, EmilioQBH}. Dropping the condition of thermodynamic stability makes the spectral flow of dynamically stable QBHs somewhat more similar to that of fermions. This alone could make the case for the additional physical interest of the general class of QPVs over its special subclass, the GSs.

               {To further situate our work relative to existing studies of criticality in ``non-Hermitian'' systems, we reiterate that, in our context, non-Hermiticity enters {\em only} at the level of the auxiliary, pseudo-Hermitian dynamical matrix. Thus, distinctively non-Hermitian features such as EPs are not a mechanism added to an otherwise conventional transition, but rather, one of the two available critical boundaries, forced by the structural constraints of QBHs. This is in contrast with settings where an explicitly non-Hermitian generator arises from the Liouvillian of a driven-dissipative system or from a $\mathcal{PT}$-symmetric Hamiltonian with balanced gain and loss, and where arranging for an EP to coincide with a critical point has been linked to strongly amplified critical fluctuations, modified critical exponents, anomalously enhanced correlations, and altered universality classes~\cite{LeeEnhanced,LittlewoodCritFluct,UedaPTcritPoint,NonunitaryCrit}.}

With the above motivations and guiding principles in mind, in this paper we investigate finite-range, {\em translationally invariant} QBHs and develop a theory of their critical behavior in the QPV that is not tied to the property of thermodynamic stability, but includes it as a special case. Our main results may be summarized as follows.

\begin{itemize}
\item We introduce a new notion of a many-body spectral gap for dynamically stable QBHs, the \textit{Krein gap}, $\Delta_\text{Krein}$. This leverages the fact that the dynamical matrix of a dynamically stable QBH is \textit{pseudo-Hermitian} and has eigenvectors of $\pm 1$ \textit{Krein signature}, with respect to an indefinite metric that arises from bosonic commutation relations \cite{GohbergIndefiniteLA,Decon}. We show that $\Delta_\text{Krein}$ quantifies the minimal distance between quasiparticle energy bands of opposite Krein signature or, equivalently, the distance to the nearest dynamical stability phase boundary in parameter space. Thus, it can only be zero at an EP, where eigenvectors of the dynamical matrix merge, or at a KC, where the number of independent eigenvectors does not change but at least one eigenvalue has eigenvectors of opposite Krein signature. Crucially -- for arbitrary spatial dimension and number of internal degrees of freedom -- a {\em unique, translationally invariant QPV always exists} if $\Delta_\text{Krein}$ is positive. 

\item  We prove that if $\Delta_\text{Krein}>0$, the two-point (equal-time) spatial correlation functions in the QPV are exponentially bounded. It follows that the QPV can become critical and support long-range correlations {\em only} at locations in parameter space where $\Delta_\text{Krein} =0$. In this sense, the Krein gap is the correct generalization -- beyond thermodynamic stability -- of the usual many-body energy gap \cite{Cramer2006} or the ``symmetrized gap'' introduced in \cite{Wolf2006} for general (not necessarily harmonic) thermodynamically stable QBHs. The emerging physical picture is that QPV criticality will generally take place at dynamical, not thermodynamic, stability phase boundaries: thermodynamic stability phase boundaries need not be associated to QPV (in fact, GS in this case) criticality, unless they coincide with dynamical stability phase boundaries. While this is surprising in light of typical models of condensed-matter (e.g., phonon) systems, one can recognize, with hindsight, that these models do lose dynamical stability together with thermodynamic stability. We investigate in detail a compelling example where these two transitions do not occur concomitantly and, indeed, the correlations in the GS remain short-ranged at the point where the system loses thermodynamic, but dynamical stability is retained. When, for a different parameter value, the Krein gap closes and the model eventually loses dynamical stability, its QPV becomes critical.

\item Consistently with the above picture, we show that thermodynamic stability may be tuned out of a QBH without affecting dynamical stability. In essence, this amounts to adjusting the Hamiltonian parameter that asymmetrically couples position and momentum quadratures. Moreover, we show that the covariance matrix in the QPV is independent of this parameter. Hence, correlation functions computed in the QPV are blind to transitions out of thermodynamic stability, unless they happen to coincide with a transition out of dynamical stability as well.This implies that criticality in the QPV can be logically disconnected from GS criticality in free boson systems: formally, given a thermodynamically and dynamically stable QBH $H$, there exists a dynamically stable, but thermodynamically unstable $H'$, with $|\text{GS}_{H}\rangle=|\text{QPV}_{H'}\rangle$, such that $H$ and $H'$ can be continuously connected without closing $\Delta_\text{Krein}$. 
  
\item Through analyzing illustrative examples, we find that, generically, an EP displays characteristics of a standard quantum critical point --  as it is parametrically approached, the correlation length diverges and, at the EP, correlation functions tend to be long-range. KCs, found at the intersection of loci of EPs, share instead points of contact with multicritical points, with their critical properties being highly dependent on the path of approach in parameter space and their behavior harder to pin down. 

\item We show that the bipartite {\em EE in the QPV obeys an area law} and, as a Krein-gap closing is approached, it scales inversely with $\Delta_\text{Krein}$ -- consistently extending existing results for thermodynamically stable QBHs \cite{Cramer2006, Wolf2006}. The multicritical character of a KC is directly reflected in a path-dependent behavior of the EE as a Krein-gap closing associated to a KC is approached. These conclusions are further corroborated by analyzing our representative models through the lenses of information-geometric tools used for identifying quantum phase transitions \cite{GuFid} and appropriately modified to account for the pseudo-Hermitian metric inherent to QBHs. Specifically, by analyzing the behavior of the quantum fidelity and the {\em quantum metric tensor} (QMT) \cite{ZanardiInfoGeom}, we find that the points in parameter space we have identified as (multi-)critical through diverging correlation functions consistently coincide with singularities of the QMT.
\end{itemize}

The remaining content is organized as follows. In Sec.\,\ref{sec:background} we introduce the class of QBHs, as well as the pertinent notions of stability and a description of the key state of interest -- the QPV. Existing results on GS criticality in thermodynamically stable QBHs and the behavior of information-theoretic indicators are also revisited and discussed in a self-contained manner. Sec.\,\ref{sec: newtake} presents our core theoretical results, by developing a detailed description for our new paradigm of criticality in QBHs, and deriving an analytic form for the correlation functions in the QPV. In Sec.\,\ref{sec:models}, we exemplify the theory derived in the previous section through key illustrative models. In Sec.\,\ref{sec:otherindicators}, these models are re-analyzed from the point of view of the EE scaling and the behavior of the QMT. Sec.\,\ref{sec:conclusion} concludes by summarizing the key findings and  discussing their significance, along with outstanding problems and avenues for future exploration. Several appendices provide full technical detail on derivations not included in the main text and complete proofs of our main theorems.

\section{Background}
\label{sec:background}

\subsection{Translationally invariant quadratic bosonic Hamiltonians}

We focus on translationally invariant systems defined on a $D$-dimensional lattice, with bi-infinite boundary conditions (BCs) in every spatial direction, so that $\bm{j} \in \mathbb{Z}^D$. When it is useful to work in finite size with periodic BCs, we instead take $\bm{j} \in \{1, \ldots, N\}^D$. More general lattices could be considered along the lines of 
Ref.~\cite{Cramer2006}, but we do not pursue this here. To each site we assign $d$ internal bosonic degrees of freedom (DOFs) (equivalently, we assume there are $d$ bosons per unit cell) and denote the corresponding creation and annihilation operators by $a_{\bm{j},\ell}^\dag$ and $a_{\bm{j},\ell}$, respectively, with $\ell=1,\ldots,d$. These satisfy the canonical commutation relations (CCRs) $[a_{\bm{j},\ell},a_{\bm{j'},\ell'}] = 0$ and $[a_{\bm{j},\ell},a_{\bm{j'},\ell'}^\dag] = \delta_{\bm{j}\bm{j'}}\delta_{\ell \ell'}$, where we leave the identity operators on the underlying bosonic multimode Fock space implicit. The most general quadratic bosonic Hamiltonian (QBH), with a finite coupling range $R \geq 1$, takes the form \cite{Blaizot}:
\begin{align}
\label{eq: QBH}
H = \sum_{\bm{j},\bm{r},\ell,\ell'} a_{\bm{j},\ell}^\dag (\mathbf{K}_{\bm{r}})_{\ell \ell'} a_{\bm{j}+\bm{r},\ell'} + \frac{1}{2}\big[ a_{\bm{j},\ell}^\dag(\bm{\Delta}_{\bm{r}})_{\ell\ell'} a_{\bm{j}+\bm{r},\ell'}^\dag + \text{H.c.}\big], 
\end{align}
where the $d\times d$ \textit{hopping} matrices $\mathbf{K}_{\bm{r}}$ encode the U(1)-symmetric particle-conserving couplings between sites separated by $\bm{r} \in {\mathbb Z}^D$, while the $d\times d$ \textit{pairing} or \textit{squeezing} matrices $\bm{\Delta}_{\bm{r}}$ encode the particle-nonconserving couplings. The finite-range condition implies that $\mathbf{K}_{\bm{r}} = \bm{\Delta}_{\bm{r}} = 0$ for $|\bm{r}|>R$. 
The hopping matrices satisfy $\mathbf{K}_{\bm{r}}^\dag = \mathbf{K}_{-\bm{r}}$ and the pairing matrices satisfy $\bm{\Delta}_{\bm{r}}^T = \bm{\Delta}_{-\bm{r}}$. We note that the condition on $\mathbf{K}_{\bm{r}}$ ensures that $H=H^\dag$, while the condition on $\bm{\Delta}_{\bm{r}}$ removes a redundancy in the description inherited from the CCRs in the way of a constant shift \cite{Decon, PostBosoranas}.


Translational invariance provides a convenient representation in momentum space, in terms of the canonical Fourier bosonic modes given by
\begin{align*}
    b_{\ell}(\bm{k}) \equiv \sum_{\bm{j}}e^{-i\bm{k}\cdot\bm{j}} a_{\bm{j},\ell},\qquad a_{\bm{j},\ell} = \frac{1}{V}\int_\text{BZ}d^D\bm{k} \,e^{i\bm{k}\cdot\bm{j}} b_\ell(\bm{k}), 
\end{align*}
where $\bm{k}$ ranges over the first Brillouin zone (BZ), of volume $V$ (in the representative case of a cubic lattice with unit lattice constant, the BZ is given by $[-\pi, \pi]^D$, yielding $V = (2\pi)^D$). The CCRs take the form:
\begin{align*}
    [b_\ell(\bm{k}),b_{\ell'}^\dag(\bm{k'})] = V  \delta(\bm{k}-\bm{k'})\delta_{\ell \ell'},\qquad   [b_{\ell}(\bm{k}),b_{\ell'}(\bm{k'})] = 0.
\end{align*}
In terms of these modes, the Hamiltonian takes a block-diagonal form
\begin{align}
\label{eq: Hblockdiag}
    H = \frac{1}{V}\int_\text{BZ} d^D\bm{k}\, H(\bm{k}),\quad H(\bm{k}) = \sum_{\ell,\ell'} b_\ell^\dag(\bm{k})\mathbf{K}_{\ell \ell'}(\bm{k})b_{\ell'}(\bm{k}) + \frac{1}{2}\big[ b_{\ell}^\dag(\bm{k})\bm{\Delta}_{\ell \ell'}(\bm{k}) b_{\ell'}^\dag (-\bm{k}) + \text{H.c.}\big],
    \end{align}
where $\mathbf{K}(\bm{k})$ and $\bm{\Delta}(\bm{k})$ are the Fourier transforms of the hopping and pairing matrices, namely, 
\begin{align}
\label{eq: KkDeltak}
    \mathbf{K}(\bm{k}) \equiv \sum_{\bm{r}} e^{i\bm{k}\cdot \bm{r}} \mathbf{K}_{\bm{r}},\qquad  \bm{\Delta}(\bm{k}) \equiv\sum_{\bm{r}} e^{i\bm{k}\cdot \bm{r}} \bm{\Delta}_{\bm{r}}.
\end{align}
Owing to the properties the hopping and pairing matrices obey in real space, these $\bm{k}$-dependent matrix-valued functions satisfy $\mathbf{K}^\dag(\bm{k}) = \mathbf{K}(\bm{k})$ and $\bm{\Delta}^T(\bm{k}) = \bm{\Delta}(-\bm{k})$, respectively. The finite-range condition ensures that the sums in Eq.\,\eqref{eq: KkDeltak} truncate, implying in turn that the matrix elements of $\mathbf{K}(\bm{k})$ and $\bm{\Delta}(\bm{k})$ are trigonometric polynomials. Consequently, these matrix-valued functions are analytic in $\bm{k}$.

\subsubsection{Diagonalization of translationally invariant QBHs.}
\label{sub:Diag}

A QBH can be diagonalized, if a set of bosonic quasiparticle modes exist, in terms of which $H$ takes a diagonal form \cite{Blaizot}. That is, we seek a new, complete set of bosonic operators $\beta_n(\bm{k}) = \sum_{i} [u_{ni} b_i(\bm{k}) + v_{ni}b_{i}^\dag(-\bm{k})]$ and real energies $\omega_n(\bm{k})$ (in units where $\hbar=1$) such that
\begin{align}
\label{eq: betaomega}
    [H,\beta_n(\bm{k})] = -\omega_n(\bm{k}) \beta_n(\bm{k}), \quad n=1,\ldots,d . 
\end{align}
If such a set of quasiparticles and energies exists, then it follows that 
\begin{align}
\label{eq: Hk}
    H \equiv \frac{1}{V}\int_\text{BZ} d^D\bm{k} \sum_{n}\omega_n(\bm{k}) \beta_n^\dag(\bm{k}) \beta_n(\bm{k}) + E_{\text{qpv}} ,
\end{align}
where $E_{\text{qpv}}$ is a (potentially infinite) ``zero-point'' energy contribution that arises from the spatially infinite lattice and is associated to a Bogoliubov {\em quasiparticle vacuum state} (QPV) $\ket{\tilde{0}}$, defined by $\beta_n(\bm{k})\ket{\tilde{0}} = 0$, for all $n$ and $\bm{k}$ (see Sec.\,\ref{sec: qvacback} for more details). Eigenstates of $H$ are then built from $\ket{\tilde{0}}$, by creating quasiparticle excitations.

To construct the \textit{Bogoliubov transformation} that maps $b(\bm{k})\mapsto \beta(\bm{k})$, let bosonic Nambu arrays in real and momentum space be defined as 
\begin{align}
    \phi_{\bm{j}} \equiv  [a_{\bm {j},1},\ldots,a_{\bm{j},d},a_{\bm{j},1}^\dag,\ldots,a_{\bm {j},d}^\dag]^T = \frac{1}{V}\int_{\text{BZ}} d^D\bm{k} \, e^{i\bm{k}\cdot\bm{j}} \phi(\bm{k}), 
    \label{mNambu}
\end{align}
with $ \phi(\bm{k}) \equiv [b_1(\bm{k}),\ldots,b_d(\bm{k}),b_1(-\bm{k})^\dag,\ldots,b_{d}^\dag(-\bm{k})]^T$. It then follows that under real-space lattice translations $\phi_{\bm{j}}\mapsto \phi_{\bm{j}+\bm{r}}$, the momentum-space Nambu array transforms naturally as $\phi(\bm{k}) \mapsto e^{i\bm{k}\cdot\bm{r}}\phi(\bm{k})$.
The momentum-space Nambu array in Eq.\,\eqref{mNambu} further satisfies
\begin{align}
\label{eq: phiCCR}
    \phi(\bm{k})^\dag = [\bm{\tau}_1\phi(-\bm{k})]^T,\quad [\phi_{\ell}(\bm{k}),\phi_{\ell'}^\dag(\bm{k'})] = V\delta(\bm{k}-\bm{k'})(\bm{\tau}_3)_{\ell\ell'}, \quad\bm{\tau}_3 \equiv \begin{pmatrix}
        \bm{1}_d & 0 \\ 0 & -\bm{1}_d
    \end{pmatrix}.
\end{align}
We similarly define the Nambu-Pauli matrices $\bm{\tau}_{1,2} \equiv  \bm{\sigma}_{1,2}\otimes \bm{1}_d$. One can then verify that \cite{Blaizot, PostBosoranas, DynamicalMetastability}
\begin{align}
\label{eq: Hphicomm}
    [H,\phi(\bm{k})] \equiv -\mathbf{g}(\bm{k})\phi(\bm{k}),\quad \mathbf{g}(\bm{k}) = \begin{pmatrix}
        \mathbf{K}(\bm{k}) & \bm{\Delta}(\bm{k})
        \\
        -\bm{\Delta}^*(-\bm{k}) & - \mathbf{K}^*(-\bm{k}) 
    \end{pmatrix},
\end{align}
where the $2d \times 2d$ matrix-valued function $\mathbf{g}(\bm{k})$, which we call the \textit{Bloch dynamical matrix} for reasons that will become clear momentarily, can be shown to be \textit{pseudo-Hermitian} with respect to the indefinite metric $\bm{\tau}_3$: that is, $\mathbf{g}^\dag(\bm{k}) = \bm{\tau}_3 \mathbf{g}(\bm{k})\bm{\tau}_3$. In fact, $\mathbf{g}(\bm{k})$ fails to be Hermitian whenever the pairing matrices are non-vanishing. Equivalently, $\mathbf{g}(\bm{k})$ is Hermitian if and only if $H$ has U(1) symmetry, and thus conserves particle number \cite{Colpa,Blaizot,Decon}\footnote{Since the key property is in fact {\em normality}, one can show that $[\mathbf{g} (\bm{k}), \mathbf{g}^\dag(\bm{k})]\equiv 0$ if and only if $\mb K(\bm{k})\mathbf{\Delta} (\bm{k})+\mathbf{\Delta}(\bm{k}) \mb K^*(-\bm{k})=0$.  }.
As an additional constraint, the dynamical matrix always obeys a ``charge-conjugation'' symmetry, $\mathbf{g}^*(\bm{k}) = -\bm{\tau}_1 \mathbf{g}(-\bm{k}) \bm{\tau}_1$. 

In light of Eqs.\,\eqref{eq: betaomega} and \eqref{eq: Hphicomm}, finding the desired Bogoliubov transformation amounts to finding a transformation $\phi(\bm{k}) \mapsto \psi(\bm{k}) \equiv \mathbf{L}^{-1}(\bm{k})\phi(\bm{k})$, such that (i) $\mathbf{L}(\bm{k})$ diagonalizes $\mathbf{g}(\bm{k})$; and (ii) $\psi(\bm{k})$ satisfies Eq.\,\eqref{eq: phiCCR}, ensuring that the elements obey bosonic statistics. Mathematically, (ii) requires that $\mathbf{L}^*(\bm{k}) = \bm{\tau}_1\mathbf{L}(-\bm{k})\bm{\tau}_1$ and $\mathbf{L}^\dag(\bm{k})\bm{\tau}_3\mathbf{L}(\bm{k}) = \bm{\tau}_3$, which ensures that $\psi(\bm{k})$ takes the form
\begin{align*}
\psi(\bm{k}) = [\beta_1(\bm{k}),\ldots,\beta_d(\bm{k}),\beta_1(-\bm{k})^\dag,\ldots,\beta_{d}^\dag(-\bm{k})]^T,
\end{align*}
with $\{\beta_n(\bm{k})\}$ satisfying the CCRs. Meanwhile, condition (i) requires $[H,\psi(\bm{k})] = -\bm{\Omega}(\bm{k})\psi(\bm{k})$,  with
\begin{align}
\label{eq:gkdiag}
    \bm{\Omega}(\bm{k}) \equiv \text{diag}\big[\omega_1(\bm{k}),\ldots,\omega_d(\bm{k}),-\omega_1(-\bm{k}),\ldots,-\omega_d(-\bm{k})\big] = \mathbf{L}^{-1}(\bm{k})\mathbf{g}(\bm{k})\mathbf{L}(\bm{k}).
\end{align}
By definition, the columns of $\mathbf{L}(\bm{k})$ are the eigenvectors of $\mathbf{g}(\bm{k})$, which, by symmetry constraints, must satisfy
\begin{align}
\label{eq: Geig}
    \mathbf{g}(\bm{k})\vec{\beta}_{n,\pm}(\bm{k}) = \pm \omega_n(\pm \bm{k}) \vec{\beta}_{n,\pm}(\bm{k}),\quad \vec{\beta}_{n,s}^\dag(\bm{k})\bm{\tau}_3 \vec{\beta}_{n',s'}(\bm{k}) = \delta_{nn'}(\bm{\tau}_3)_{ss'},\quad \vec{\beta}_{n,\mp}(\bm{k}) = \bm{\tau}_1 \vec{\beta}_{n,\pm}^*(-\bm{k}),
\end{align}
so that $\mathbf{L}(\bm{k})$ is a proper Bogoliubov transformation. Explicitly, the quasiparticles in Eq.\,\eqref{eq: betaomega} are then given by $ \beta_n(\bm{k}) = \vec{\beta}_{n,+}^{\,\dag}(\bm{k})\bm{\tau}_3 \phi(\bm{k}).$ We will henceforth refer to the eigenvalue $\omega_n(\bm{k})$ as the {\em $n$-th quasiparticle energy band}. It should be noted that the choice of $\omega_n(\bm{k})$ and $-\omega_n(-\bm{k})$ in Eq.\,\eqref{eq:gkdiag} is {\em not} arbitrary. Perhaps counterintuitively, this choice does not depend on the sign of the quasiparticle energy. Rather, as we will discuss in Sec.\,\ref{sec: DSKS}, it is dictated by the \textit{Krein signature} of the corresponding eigenvector, which arises from the pseudo-Hermiticity of $\mb g(\bm{k})$. Importantly, since $\mathbf{g}(\bm{k})$ is generically non-normal, there is no guarantee that the desired transformation $\mb L(\bm k)$ exists. In the most extreme case, $\mathbf{g}(\bm{k})$ need not be diagonalizable at all ({see Sec.\,\ref{sec: DSKS}}). Further to that, even if $\mathbf{g}(\bm{k})$ is diagonalizable, the desired Bogoliubov transformation is guaranteed to exist if and only if all the eigenvalues of $\mathbf{g}(\bm{k})$ are real \cite{Colpa,Blaizot,Decon}.

\subsubsection{Dynamical stability and Krein theory.}
\label{sec: DSKS}

Matrices that are diagonalizable and possess an entirely real spectrum are especially interesting from the perspective of dynamical systems theory. Specifically, a set of coupled linear time-invariant (LTI) differential equations, say, $\partial_t\vec{v} = -i\mathbf{M}\vec{v},$ is known to generate bounded evolution for all initial conditions at all times if $\mathbf{M}$ is diagonalizable with an entirely real spectrum. This connects directly with the diagonalization of QBHs by noting that Eq.\,\eqref{eq: Hphicomm} provides the Heisenberg equations of motion for the Nambu array, $\partial_t \phi(\bm{k};t) = i[H,\phi(\bm{k};t)] = -i \mathbf{g}(\bm{k})\phi(\bm{k};t)$. Since this dynamical equation is of the aforementioned form, it is stable whenever $\mb{g}(\bm{k})$ is diagonalizable with only real eigenvalues. In fact, due to the symmetry conditions that a bosonic dynamical matrix $ \mathbf{g}(\bm{k})$ must obey, reality of the spectrum is sufficient and also necessary for $\phi(\bm{k};t)$ to remain bounded for all times. Moreover, since the dynamics generated by $H$ is unitary, this implies any observable will also remain bounded for all times, irrespective of the initial state. We call systems with this property \textit{dynamically stable}. This follows from the fact that the Heisenberg dynamics of {\em any} operator built from linear combinations of products of creation and annihilation operators is entirely determined by that of the fundamental operators. For example, $(a^\dag a)(t) = a^\dag(t) a(t)$.  Thus, if each mode $b_j(\bm{k};t)$ undergoes strictly bounded evolution, then so will any arbitrary observable. This establishes a one-to-one correspondence between the class of QBHs that are dynamically stable and those that may be diagonalized via a Bogoliubov transformation. 

Within the mathematical literature, LTI dynamical systems generated by pseudo-Hermitian matrices have been long studied \cite{YakuKrein}. In particular, a comprehensive characterization of the stability phase boundaries of such systems has been obtained \cite{GohbergIndefiniteLA}. Consider a general pseudo-Hermitian dynamical system, of the form
\begin{align}
    \partial_t \vec{v} = -i \mathbf{M}\vec{v},\quad \mathbf{M}^\dag = \bm{\eta} \mathbf{M}\bm{\eta}^{-1} ,
\end{align}
where $\bm{\eta}$ is an invertible Hermitian matrix. By definition, $\mathbf{M}$ is pseudo-Hermitian with respect to the indefinite metric $\bm{\eta}$. Suppose now that $\mathbf{M}$ depends on a set of parameters $p$ in a manifold $\mathcal{M}$ and that, within some submanifold $\mathcal{S}\subset \mathcal{M}$, the system is dynamically stable. One may then ask, what is the nature of the system along the stability phase boundaries $\partial \mathcal{S}$? This question can be answered by introducing the \textit{Krein signature}: given a vector $\vec{u}$, the Krein signature of $\vec{u}$ with respect to $\bm{\eta}$ is defined as the sign of $\vec{u}^{\,\dag} \bm{\eta}\vec{u}$, with the convention $\text{sgn}(0)=0$. A subspace is said to be $\bm{\eta}$-definite if every vector within the subspace has the same Krein signature. Otherwise, we say it is $\bm{\eta}$-indefinite. 

One may show that the eigenspaces associated to eigenvalues of $\mathbf{M}$ that are (i) non-real or (ii) engender nontrivial Jordan chains -- resulting in exponential or polynomial instabilities over time, respectively -- must be $\bm{\eta}$-indefinite. In case (ii), $\mathbf{M}$ is nondiagonalizable. This suggests that $\bm{\eta}$-definiteness may be closely tied to dynamical stability. This connection is made precise in a key result of Krein stability theory, known as the {\em Krein-Gel'fand-Lidskii theorem} (see e.g., Chapter III of Ref.\,\cite{YakuKrein}). Specifically, if $\lambda$ denote a {\em real} eigenvalue of $\mathbf{M}$ with a corresponding eigenspace $\mathcal{E}_{\lambda}$, this theorem establishes that: 

\smallskip

{\bf (a)} {\em If $\mathcal{E}_{\lambda}$ is $\bm{\eta}$-definite, all Jordan chains associated to $\lambda$ are of length one and, in addition, there exists $\epsilon,\delta>0$ such that if $\mathbf{M}'$ is $\bm{\eta}$-pseudo-Hermitian and $\norm{\mathbf{M}'-\mathbf{M}}<\delta$, then all the eigenvalues $\lambda'$ of $\mathbf{M}'$ such that $|\lambda'-\lambda|<\epsilon$ are also real and correspond to Jordan chains of length one.}

{\bf (b)} {\em If $\mathcal{E}_{\lambda}$ is $\bm{\eta}$-indefinite and all the Jordan chains associated to $\lambda$ are of length one, then for {\em every} $\epsilon>0$ there exists a $\bm{\eta}$-pseudo-Hermitian matrix $\mathbf{M}'$ which obeys $\norm{\mathbf{M}-\mathbf{M}'}<\epsilon$ and possesses {\em non-real} eigenvalues in an open neighborhood of $\lambda$.}

\smallskip

This theorem reveals that the stability phase boundaries of pseudo-Hermitian systems corresponds {\em precisely} to those parameters $p$ for which $\mathbf{M}(p)$ has {\em one or more indefinite eigenspaces}. As stated earlier, if $\mathbf{M}(p)$ is nondiagonalizable, it must have at least one indefinite eigenspace. In this case, the point $p=p_\text{EP}$ is called an \textit{exceptional point} (EP). Remarkably, however, $\mathbf{M}(p)$ may be fully diagonalizable yet still posses an indefinite eigenspace. In this case, we say that the eigenvalues corresponding to this eigenspace supports a \textit{Krein collision} (KC). Equivalently, there is a KC at $\lambda$ if there are at least two corresponding eigenvectors $\vec{u}_\pm$ with opposite Krein signatures. We may then compactly characterize dynamical stability phase boundaries as those points in the parameter space for which $\mathbf{M}(p)$ has either an EP or a KC. 

These results are directly applicable to translationally invariant QBHs. If we append to $\mathbf{g}$ a dependence on some set of parameters, $\mathbf{g}(\bm{k};p)$, we identify the dynamical stability phase boundaries of $H$ as those points $p$ where $\mathbf{g}(\bm{k},p)$ supports an EP or KC. In the former case, $H$ is dynamically unstable. In the latter case, the system is stable but, by virtue of point {\bf (b)} above, it may be destabilized by some arbitrarily small perturbation. As a consequence of Eq.\,\eqref{eq: Geig}, one may show that both EPs and KCs can only occur if there are bands $n$ and $m$ such that $\omega_n(\bm{k}) = -\omega_m(\bm{k})$. In particular, we see that a bosonic zero mode, i.e., a mode corresponding to some band $n^*$ and momentum $\bm{k}^*$ such that $\omega_{n^*}(\bm{k}^*)=0$, is intrinsically unstable to certain arbitrarily small perturbations \cite{Squaring, Decon}.

\subsubsection{Thermodynamic stability.}

Another notion of stability that is keenly relevant for QBHs is thermodynamic stability. We say a Hamiltonian $H$ is {\em thermodynamically stable} if there exists either a finite lower or upper bound on its expectations. It is possible to show (e.g., by considering expectations $\braket{\alpha|H|\alpha}$, with $\ket{\alpha}$ a multimode coherent state) that any QBH whose Bloch dynamical matrix possesses a non-real eigenvalue is thermodynamically unstable \cite{Decon}. However, dynamical stability is neither necessary nor sufficient for thermodynamic stability. For example, consider a free particle QBH, $H \propto p^2$, where for later reference we introduce the {\em quadrature operators}
\begin{equation}
p\equiv \frac{i(a^{\dag}-a)}{\sqrt{2}}, \quad  x\equiv \frac{a^{\dag}+a}{\sqrt{2}},\quad [x, p] = i  .
\label{quad}
\end{equation}
The Hamiltonian is dynamically unstable, as evident from considering the evolution of $x$, yet manifestly bounded from below. Assessing whether a given dynamically stable QBH is also thermodynamically stable can be accomplished by examining Eq.\,\eqref{eq: Hk}. The eigenvalues of $H$ are built from integer linear combinations of the quasiparticle energies $\omega_n(\bm{k})$. Thus, a dynamically stable $H$ is bounded from below (above) if and only if $\omega_n(\bm{k})\geq 0$ ($\omega_n(\bm{k}) \leq 0$) for all $n$ and $\bm{k}$. In the first, non-negative case, it follows that the QPV is the GS of $H$. In Appendix \ref{App: BlochMat}, we illustrate how the conditions for both thermodynamical and dynamical stability can be expressed directly in terms of properties of $\mathbf{g}(\bm{k})$ for the relevant case of single internal degree of freedom (DOF, $d=1$).

Necessary and sufficient conditions for thermodynamic stability have been established for general QBHs, beyond translationally invariant ones we focus on. Specifically, a QBH $H$ is thermodynamically stable if and only if the matrix $\mathbf{H}=\bm{\tau}_3\mathbf{G}$ is positive- (or negative-) semidefinite \cite{Derezinski} (Theorem 2.4). The possibility of having {\em both} positive and negative quasiparticle energies allows QBHs to be dynamically stable, but thermodynamically unstable \cite{Decon}.

\subsubsection{The quasiparticle vacuum, its energy and correlations.}
\label{sec: qvacback}

In Sec.\,\ref{sub:Diag}, we noted that the many-body eigenstates of a dynamically stable QBH may be constructed in the standard way by exciting quasiparticles on top of the QPV, i.e., the state $\ket{\tilde{0}}$ satisfying $\beta_n(\bm{k})\ket{\tilde{0}} = 0$ for all $n$ and $\bm{k}$. Let us characterize the associated energy contribution, as appearing in Eq.\,\eqref{eq: Hk}. Starting from Eq.\,\eqref{eq: Hblockdiag}, note that we may symmetrize the hopping term, which allows us to rewrite the Hamiltonian as:
\begin{align*}
    H = \frac{1}{V}\int_\text{BZ} d^D\bm{k}\, \frac{1}{2}\bigg[ \phi^\dag(\bm{k}) \underbrace{\begin{pmatrix}
\mathbf{K}(\bm{k}) & \bm{\Delta}(\bm{k})
\\
\bm{\Delta}^*(-\bm{k}) & \mathbf{K}^*(-\bm{k})
    \end{pmatrix}}_{\mb h(\bm k)\equiv \bm{\tau}_3\mb{g}(\bm k)}\phi(\bm{k}) - \tr[\mathbf{K}(\bm{k})]                  {V}\delta(\bm{0}) \bigg].
\end{align*}
The Hermiticity of $\mb h(\bm k)$ ensures the Hermiticity of $H$ while, in view of the above results, the condition $\mb h(\bm k)\geq0$ guarantees thermodynamic stability. Meanwhile, as discussed in Sec.\,\ref{sec: DSKS}, if $\mathbf{g}(\bm{k})$ has an entirely real spectrum for all $\bm{k}$, we may find a Bogoliubov transformation $\mathbf{L}(\bm{k})$ that diagonalizes $\mathbf{g}(\bm{k})$ according to $\mathbf{g}(\bm{k}) = \mathbf{L}(\bm{k}) \bm{\Omega}(\bm{k}) \mathbf{L}^{-1}(\bm{k})$, with $\bm{\Omega}(\bm{k})$ defined in Eq.\,\eqref{eq:gkdiag}.  Therefore, we have 
\begin{align*}
    H =  \frac{1}{V}\int_\text{BZ} d^D\bm{k} \,\frac{1}{2}\left[\phi^\dag(\bm{k}) \bm{\tau}_3\mathbf{g}(\bm{k})\phi(\bm{k}) - \tr[\mathbf{K}(\bm{k})]                  {V}\delta(\bm{0})\right]
    =  \frac{1}{V}\int_\text{BZ} d^D\bm{k} \, \frac{1}{2}\left[\psi^\dag(\bm{k}) \bm{\tau}_3 \bm{\Omega}(\bm{k}) \psi(\bm{k}) - \tr[\mathbf{K}(\bm{k})]                  {V}\delta(\bm{0})\right].
\end{align*}
Expanding out the first term in the second equality and simplifying, we arrive at the final expression
\begin{align*}
    H = \frac{1}{V}\int_\text{BZ} d^D\bm{k}\left[\sum_{n=1}^d \omega_n(\bm{k}) \beta_n^\dag(\bm{k}) \beta_n(\bm{k}) + \frac{1}{2} \left( \sum_{n=1}^d \omega_n(\bm{k}) - \tr[\mathbf{K}(\bm{k})] \right)                  {V}\delta(\bm{0})\right].
\end{align*}
Thus, the energy of the QPV, $E_{\text{qpv}}$, may be read off as 
\begin{align}
\label{eq:QPVenergydensity}
  H\ket{\tilde{0}}=E_{\text{qpv}}\ket{\tilde{0}},\quad
  E_{\text{qpv}}= \underbrace{\frac{1}{V}\int_\text{BZ} d^D\bm{k}\,
  \frac{1}{2}\left(\,\sum_{n=1}^d \omega_n(\bm{k}) - \tr[\mathbf{K}(\bm{k})]\right)}_{                  {\varepsilon_{\text{qpv}}}}                  {V}\delta(\bm{0}),
\end{align}
                  {with $\varepsilon_{\text{qpv}}$ being the QPV energy density.}
It is worth noting that, in the number-conserving case, $\tr[\mathbf{K}(\bm{k})] =\sum_{n=1}^d \omega_n(\bm{k})$. Therefore, this term cancels out. This is not surprising because, in this case, the quasiparticles are linear combinations of $b_i(\bm{k})$ (they do not include $b_i^\dag(-\bm{k})$), preserving the normal ordering of the operators in the real-space formulation of the Hamiltonian in Eq.\,\eqref{eq: QBH}. Normal ordering ensures that the bare Fock vacuum is annihilated by $H$ and does not have a divergent energy. 

Quasiparticle (or non-interacting) vacuum states may be most generally defined in a {\em Hamiltonian-independent} way as the states that can be reached from a Fock vacuum $\ket{0}$ by means of a (homogeneous) Gaussian unitary transformation. Gaussian unitaries $U$ are the ones that implement (linear) Bogoliubov transformations \cite{GQI}. For translationally invariant systems, we have seen in Sec.\,\ref{sub:Diag} that the Gaussian unitaries which preserve translational symmetry act as 
\begin{align}
\label{eq:GUT}
    U\phi(\bm{k})U^\dag = \mathbf{L}^{-1}(\bm{k})\phi(\bm{k}),\quad \mathbf{L}^*(\bm{k}) = \bm{\tau}_1\mathbf{L}(-\bm{k})\bm{\tau}_1,\quad \mathbf{L}^\dag(\bm{k})\bm{\tau}_3\mathbf{L}(\bm{k}) = \bm{\tau}_3,
\end{align}
where the inverse is included as a matter of convention. Writing $\psi(\bm{k}) = \mathbf{L}^{-1}(\bm{k})\phi(\bm{k})$ as before, the QPV state corresponding to the unitary $U$ is given simply by $\ket{\tilde{0}} = U\ket{0}$. This follows from:
\begin{align}
    \beta_n(\bm{k})\ket{\tilde{0}} = \beta_n(\bm{k})U \ket{0}  = (U b_n(\bm{k})U^\dag )U\ket{0} =U b_n(\bm{k})\ket{0} = 0.
    \label{eq:Uvac}
\end{align}
Furthermore, since $U$ is translationally invariant (i.e., it leaves different $\bm{k}$-sectors uncoupled), the state $\ket{\tilde{0}}$ is also translationally invariant. Because such states are reached by means of a Gaussian unitary transformation, they are necessarily Gaussian states themselves. A quantum state $\rho$ on a lattice is \textit{Gaussian} (or quasi-free) if it is uniquely determined by its bosonic mean vector and bosonic  covariance matrix (CM), respectively defined as \cite{Geza2002, Wolf2006, GQI}:
\begin{align}
\vec{m}^\rho_{\bm{j}} \equiv \braket{\phi_{\bm{j}}}_\rho  ,\quad 
(\mathbf{C}^\rho_{\bm{j},\bm{j'}})_{\ell\ell'} \equiv \braket{\{\phi_{\bm{j},\ell} - \langle \phi_{\bm{j},\ell}}, \phi_{\bm{j'},\ell'}^\dag - \braket{\phi_{\bm{j'},\ell'}^\dag} \} \rangle_\rho  , 
\label{CCM}
\end{align}
where $\{A,B\} \equiv AB+BA$ denotes the anticommutator and $\phi_{\bm{r}}$ is defined in Eq.\,\eqref{mNambu}. Since we have further restricted to homogeneous Gaussian unitary transformations, it can be shown that the corresponding mean vectors vanish, leaving only the CM to uniquely specify the state. For translationally invariant states, the CM can only depend on the separation $\bm{r} = \bm{j'}-\bm{j}$. As such, we may refocus our attention from the connected correlation function in Eq.\,\eqref{CCM} to $\mathbf{C}_{\bm{r}}^\rho \equiv \mathbf{C}_{\bm{j},\bm{j}+\bm{r}}^\rho$, which encodes all real-space correlations between the modes localized at any two sites separated by $\bm{r}$. It is also natural to define the momentum-space CM as
\begin{align*}
    \mathbf{C}^\rho(\bm{k}) \equiv \sum_{\bm{r}} e^{-i\bm{k}\cdot\bm{r}}\mathbf{C}^\rho_{\bm{r}}.
\end{align*}

In the context of QBHs, the Gaussian unitary that maps the Fourier modes $b_n(\bm{k})$ to the quasiparticles $\beta_n(\bm{k})$ is precisely the Bogoliubov transformation that diagonalizes $H$. Therefore: (i) dynamically stable QBHs are precisely those that may be diagonalized by a Gaussian unitary; and (ii) their QPVs are obtained by acting on the Fock vacuum with this Gaussian unitary. Thus, with $\rho^\text{qpv}=\ket{\tilde{0}}\bra{\tilde{0}}$, property (ii) specifically provides us with a way to compute the QPV CM in momentum space, $\mb C^\text{qpv}(\bm{k})$, and from there the real-space CM, $\mb C_{\bm{r}}^\text{qpv}$, in terms of the modal matrix $\mathbf{L}(\bm{k})$ of $\mathbf{g}(\bm{k})$. While we refer to Appendix \ref{App: CMfromModal} for the details, the final result takes a very simple form:
\begin{align}
\label{eq: Ck}
    \mathbf{C}^\text{qpv}(\bm{k}) = \mathbf{L}(\bm{k})\mathbf{L}^\dag(\bm{k}),\qquad \mathbf{C}^\text{qpv}_{\bm{r}} = \frac{1}{V}\int_{\text{BZ}} d^D\bm{k} \, e^{i\bm{k}\cdot \bm{r}} \mathbf{L}(\bm{k})\mathbf{L}^\dag(\bm{k}).
\end{align}

\subsection{Ground-state correlation functions and criticality}

As mentioned, critical quantum states are characterized by the long-range behavior of their correlation functions. For QBHs, the most relevant states and their dynamics are Gaussian, making the framework of continuous-variable (CV) \textit{Gaussian quantum information} a natural and widely adopted language \cite{CVQI, GQI}. Because many of the existing results are formulated within this formalism \cite{Cramer2006,Wolf2006,AreaLawBosonic}, and this will play a central role in our analysis of QPV criticality and entanglement, we now recast the preceding discussion in terms of the Gaussian quadrature representation.

\subsubsection{The quadrature formalism.} The basic concepts are most easily introduced by considering the simple setting of a single DOF and spatial dimension, $d=D=1$, and by working under periodic BCs on a lattice of size $N$. The generalization to multiple internals degrees of freedom and higher spatial dimension is straightforward in principle, at the cost of added notation \cite{Cramer2006}. Instead of bosonic Nambu arrays, we introduce quadrature arrays by associating a quadrature operator to each site: 
\begin{align}
\label{eq: DimLessQuads}
R \equiv [x_1,\ldots,x_N,p_1,\ldots,p_N]^T \equiv \{R_j \}, \qquad {j=1,\ldots, 2N.}
\end{align} 
The CCRs are then given by:
\begin{align*}
    [R_j,R_{j'}]=i\,\boldsymbol{\Sigma}_{ j j'}
     , \quad \boldsymbol{\Sigma}\equiv \begin{pmatrix}
    \mbox{\hspace*{2mm}0}&\mb{1}\\
    -\mb{1}&0
\end{pmatrix},
\end{align*}
where again we leave the Fock-space identity operator implicit and $\boldsymbol{\Sigma}$ denotes the \textit{symplectic form}. If $\mathbf{H}_{xx}$ ($\mathbf{H}_{pp}$) is an $N\times N$ real symmetric matrix that encodes the couplings between position (momentum) degrees of freedom, and $\mathbf{H}_{xp}$ an $N\times N$ real matrix that encodes the couplings between position and momentum degrees of freedom, we can then write a general QBH in the form 
\begin{align}
\label{eq:quadHam}
H= {\frac{1}{2}\sum
\limits_{ j j'}^{} R_{j}\mb H_{ j j'}R_{j'}- f(\mb H), }
\qquad \mb H\equiv\begin{pmatrix}
    \mb H_{xx}&\mb H_{xp}\\
    \mb H_{xp}^T&\mb H_{pp}
\end{pmatrix}.
\end{align}
Denoting $\text{Re}(\mb X)\equiv ({\mb X+\mb X^{*}})/2$ and $\text{Im}(\mb X)\equiv ({\mb X-\mb X^{*}})/2i$ the element-wise real and imaginary parts, respectively, the quadrature couplings in Eq.\,\eqref{eq:quadHam} can be related to the bosonic ones from Eq.\,\eqref{eq: QBH} as follows: 
$$\mb H_{xx}=\text{Re}(\mb K+\mb \Delta),\quad \mb H_{pp}=\text{Re}(\mb K-\mb \Delta), \quad \mb H_{xp}=\text{Im}(\mb \Delta-\mb K), \quad \mb H_{xp}^T=\text{Im}(\mb \Delta+\mb K), \quad f(\mb H)=\frac{1}{2} \tr(\mb K).$$
The important subclass of \textit{harmonic} lattices, as investigated in \cite{Wolf2006, Cramer2006}, corresponds to QBHs where the position and momentum degrees of freedom are {\em not} coupled, $\mb H_{xp}=0$ and $\mb H=\mb H_{xx}\oplus \mb H_{pp}$. 

Given a Gaussian quantum state $\rho$, the CM $\mb \Gamma$ in the quadrature basis is naturally defined as
\begin{equation}
\mathbf{\Gamma}_{j j'}^\rho \equiv \langle \{ R_{j} - \braket{R_{j}},R_{j'} - \braket{R_{j'}}\} \rangle_\rho.
\label{GammaCM}
\end{equation}
Any $\mb \Gamma^\rho$ corresponding to a physical state must be a positive-definite, real, symmetric matrix \cite{SolvingQuasifree}. Furthermore, the Heisenberg uncertainty relation obeyed by the quadratures imposes an additional constraint \cite{GaussianCov,AdessoEE}:
\begin{align}
\label{state condition}
    \mb \Gamma^\rho +i \, \boldsymbol{\Sigma}\geq 0, \quad \forall \rho\geq 0.
\end{align}
By the Williamson theorem, we know that any symmetric, positive $2N\times 2N$ matrix can be decomposed as:
\begin{align}
\label{CM decomp}
    \mb \Gamma^\rho = \mb S^T \text{diag}\big[\nu_1,\ldots,\nu_N,\nu_1,\ldots,\nu_N \big]\mb S, \quad \nu_j>0, \;\forall j, 
\end{align} 
where $\mb S$ is a \textit{symplectic transformation}, namely, one that preserves CCRs, $\mb S^T \boldsymbol{\Sigma} \mb S=\boldsymbol{\Sigma}$, and $\nu_j$ are the \textit{symplectic eigenvalues} of $\mb\Gamma^\rho$, that is, the (strictly) positive eigenvalues of $i\boldsymbol{\Sigma}\mb\Gamma^\rho$, as $\sigma(i\boldsymbol{\Sigma}\mb\Gamma^\rho)=\{\pm \nu_j\}$. Finally, imposing Eq.\,\eqref{state condition}, it follows that a valid CM must satisfy $\nu_j\geq 1$. 

While the above applies to an arbitrary Gaussian state, additional constraints arise if $\rho$ is pure. In particular, the purity condition, $\tr(\rho^2)=1$, at the level of the CM amount to $\det{\mb \Gamma^\rho}=1$ \cite{AdessoPurity,AdessoEE}. This fixes $\nu_j=1$ for all $j$, and saturates the minimum uncertainty in Eq.\,\eqref{state condition}. The condition for purity of a Gaussian state may be compactly written in terms of the CM as \cite{KrugerThesis}: 
\begin{align}
\label{CovMatPurity}
    (\boldsymbol{\Sigma}\mb \Gamma^\rho)^2=-\mb 1_{2N}.
\end{align}
Note that $\mb C^{\rho}$ in Eq.\,\eqref{CCM} and $\mb \Gamma^{\rho}$ in Eq.\,\eqref{GammaCM} are related to each other via the linear transformation that takes the bosonic Nambu array to the quadrature one, i.e., $(a_j,a_j^\dag) \mapsto (x_j,p_j)$.

\subsubsection{Existing results on long-range correlations.} Several results have been rigorously established for relevant classes of {\em thermodynamically stable} QBHs with finite-range interactions, relating the decay of GS position or momentum correlation functions to the behavior of the many-body spectral gap. To the best of our knowledge, while lattices in $D>1$ dimensions have been considered, existing analyses focus on systems with a single DOF (single-band), $d=1$.
Since the GS of a QBH is zero-mean Gaussian (hence, uniquely determined by its CM) GS criticality properties have been extensively studied through the CM formalism we described above. Let $\mb \Gamma^\text{gs} \equiv \begin{pmatrix}
        \mb \Gamma_{xx}&\mb \Gamma_{xp}\\
      \mb \Gamma_{xp}^T&  \mb \Gamma_{pp}
    \end{pmatrix}$, where $\mb \Gamma_{xx}$ contains all the GS correlations between $x$ quadratures, and the other blocks are analogously defined. The most pertinent results for GS criticality may be summarized as follows: 

\begin{outline}
\1 For a single-band, finite-range harmonic lattice, described by a QBH of the form $\mb H=\mb H_{xx}\oplus \mb 1_N$, the GS CM and the many-body energy gap are given by:
     \begin{equation}
     \label{eq:HarmonicCorrelationsPIdentity}
\mb \Gamma^\text{gs}= \mb H_{xx}^{-1/2}\oplus \mb H_{xx}^{1/2}, \quad \Delta^\text{gs}=\lambda_{\text{min}}(\mb H_{xx})^{1/2}. 
\end{equation}
When $\Delta^\text{gs}>0$, $\mb H$ is \textit{non-critical}: the GS correlations decay exponentially as
\begin{equation}
|\langle x_{\bm j}x_{\bm j'}\rangle_\text{gs}|,\,|\langle p_{\bm j}p_{\bm j'}\rangle_\text{gs}|\leq K e^{-|\bm j-\bm j'|/\xi}, \quad K>0 ,
\label{eq:HarmonicDecay}
\end{equation}
with $\xi$ being the correlation length and the constant $K$ being generally different for $x$ and $p$ correlations \cite{Cramer2006}. As the gap closes, $\Delta^\text{gs}\rightarrow 0$, both $K$ and $\xi$ diverge. When $\Delta^\text{gs}=0$, the system is \textit{critical}: it hosts long-range correlations, that cannot be bounded by a decaying exponential. Similar conclusions may be drawn for harmonic lattices that are not translationally invariant, by appropriately generalizing the expressions for $\mb \Gamma^\text{gs}$ and $\Delta^\text{gs}$\cite{Cramer2006}.
  
\1 For a single-band, finite-range, translationally invariant QBH, long-range GS correlations have been shown to arise only if the many-body energy gap of an auxiliary, \textit{symmetrized Hamiltonian} $\mathcal{S}(\mb H)$ vanishes \cite{Wolf2006}. Specifically, the GS CM may be expressed in the form
\begin{align*}
 \mb \Gamma^\text{gs}&\equiv (\boldsymbol{ \mathcal{E}}^{-1}\oplus \boldsymbol{ \mathcal{E}}^{-1})\,\boldsymbol{\Sigma}\mathcal{S}(\mb H)\boldsymbol{\Sigma}^T, \quad  
\end{align*}
where $\boldsymbol{ \mathcal{E}}$, the symmetrized Hamiltonian $\mathcal{S}(\mb H)$ and the corresponding  {\em symmetrized gap} are defined by
 \begin{align}
\label{eq:WolfCovMat}
\boldsymbol{ \mathcal{E}} \equiv \sqrt{\mb H_{xx}\mb H_{pp}-\tfrac{1}{4}(\mb H_{xp}+\mb H_{xp}^T)^2}&, \quad 
\mathcal{S}(\mb H)\equiv \begin{pmatrix}
    \mb H_{xx}&\frac{\mb H_{xp}+\mb H_{xp}^T}{2}\\
   \frac{\mb H_{xp}+\mb H_{xp}^T}{2}&\mb H_{pp}
\end{pmatrix}, \quad \Delta^\text{gs}_{\text{sym}}\equiv \lambda_{\text{min}}(\boldsymbol{ \mathcal{E}}).
\end{align} 
All the GS correlations are bounded by a decaying exponential and the system is non-critical when $\Delta^\text{gs}_{\text{sym}}>0$.  This is significant because for non-symmetric QBHs ($\mb H_{xp}\neq \mb H_{xp}^T$), $\Delta^\text{gs}$, the gap of the physical Hamiltonian, can generically be zero without $\Delta^\text{gs}_{\text{sym}}=0$. This hints at the fact that even for thermodynamically stable QBHs, the vanishing of the Hamiltonian's energy gap fails to characterize critical phase boundaries in general.

\end{outline}

\subsection{Information-theoretic approaches to criticality}
\label{sub: info-th}

Investigations of quantum phase transitions and critical phenomena have been naturally enriched and complemented by tools from quantum information theory. In what follows, we will discuss two of them. The first focuses on the entanglement structure of many-body ground states, where quantities such as entanglement entropy and spectra diagnose correlations and universal scaling behavior near criticality \cite{Osborne,RolandoGE,Srivastava2024}. The second emphasizes the differential-geometric structure of quantum states and their parent Hamiltonians, using the {\em quantum fidelity} and the {\em quantum geometric tensor} (QGT) to probe state-distinguishability and sensitivity to parameter changes near critical singularities in parameter space \cite{Caves,GuFid,ZanardiInfoGeom}.

\subsubsection{Entanglement entropy and area law.} The EE quantifies the amount of bipartite entanglement between a distinguished region and the rest of the system. Formally, let $B$ denote a distinguished subregion of the defining lattice, with $A\cup B$ corresponding to the full system. Even if the GS $\rho$ of the system is pure,  hence its von Neumann entropy, $S(\rho)\equiv-\text{Tr} (\rho \log \rho)$, is zero, a non-vanishing entropy generically arises for the state describing subsystem $B$ alone, obtained via the partial trace $\rho_{B}\equiv \text{Tr}_{A}(\rho)$. The EE of subsystem $B$ is then defined as $S(\rho_{B})\equiv S_B \equiv -\text{Tr}\left(\rho_{B} \log \rho_{B}\right)$. Notably, for Gaussian states, $S_B$ can be computed from the CM alone \cite{AdessoEE}: 
\begin{align*}
    S_B&\nonumber=\sum\limits_{i=1}^N S(\nu_i), \qquad S(\nu)=\frac{\nu+1}{2}\log\frac{\nu+1}{2}-\frac{\nu-1}{2}\log\frac{\nu-1}{2},
\end{align*}
where $\nu_i$ are the  \textit{symplectic eigenvalues} of $[\bm{\Gamma}]_B$, namely, the positive eigenvalues of $[i \boldsymbol{\Sigma} \mb \Gamma]_B$, with $[\bullet ]_B$ denoting a projection of the matrix to a submatrix that only has support on $B$. For a pure state $\rho$, the scaling of $S(\rho_B)$ is upper-bounded by the \textit{logarithmic negativity} \cite{Vidal2002,PlenioLogNeg}, $E_N(\rho_B)\equiv \log \norm{\rho^{T_B}}_1$, where $\norm{X}_1=\text{tr}\left((X^{\dag}X)^{1/2}\right)$ is the trace norm, and $T_B$ denotes partial transposition with respect to subsystem $B$ \cite{Audenart2002,Vidal2002,AdessoContVar2004}. The logarithmic negativity of a pure Gaussian state can equivalently be computed from the corresponding partially transposed CM \footnote{Note that, on the level of the CM, partial transposition is achieved by time-reversing the quadratures that belong to the distinguished subregion, which amounts to mapping $(x_A,x_B,p_A,p_B)\mapsto (x_A,x_B,p_A,-p_B)$. As an example, for a symmetrically bisected chain with a CM $\mb \Gamma$, the transposed CM takes the form $\mb \Gamma^{T_B}=\boldsymbol{\Theta}_B\mb \Gamma \boldsymbol{\Theta}_B$, with $\boldsymbol{\Theta_B}=\mathbf{1}_N\oplus \text{diag}(\mb 1_{N/2},-\mb 1_{N/2})$.}. 
All in all:
\begin{align*}
 E_N=   -\frac{1}{2}\sum\limits_{k=1}^{2N}\log\left(\text{min}\left(1,|\sigma\left(i\boldsymbol{\Sigma}^{-1}\mb \Gamma^{T_B}\right)|\right)\right) = 
 \log \prod\limits_{k} \frac{1}{\bar{\nu}_k}, 
\end{align*} 
where $\bar{\nu}_k$ are the symplectic eigenvalues of the partially transposed CM that are less than 1 \cite{Audenart2002,Vidal2002,AdessoContVar2004}. 

The behavior of EE in the GS of a finite-range, harmonic QBH on a general (not-necessarily-translationally invariant) lattice has been found to be \emph{highly non-generic}: it is non-extensive, and scales only as the boundary of the distinguished region, obeying an \textit{area law} -- rather than a volume law, which is the case for a typical quantum state \cite{Cramer2006, AreaLawBosonic,EisertEntropyEntanglement}. For gapped Hamiltonians of the form $\mb H_{xx}\oplus \mb 1$, an upper bound on the EE scaling was found analytically \cite{Cramer2006} by upper-bounding $E_N$ and leveraging the exponential decay of quadrature correlations:
\begin{align}
\label{eq: EEUpperBound}
    S_B\leq 16c^2\norm{\mb H_{xx}}\text{Li}_{1-2D}(e^{-1/\xi})\frac{s(B)}{(\Delta^\text{gs})^2},
\end{align}
with  $\xi$ the correlation length, $c>0$, $s(B)$ the area of region $B$, $\text{Li}_{1-2D}$ a polylogarithm of degree $1-2D$ and $\Delta^\text{gs}$ the many-body gap of Eq.\,\eqref{eq:HarmonicCorrelationsPIdentity}. This expression makes explicit the validity of the area law, as well as an inverse dependence of the EE bounds on the energy gap. 

For 1$D$ systems, the area law implies no scaling of EE with system-size. For the special case of a translationally invariant, symmetrically bisected harmonic chain with $\mb H_{pp}=\mb 1_N$ and nearest-neighbor position couplings, a simple analytic form of $E_N$ to bound EE has been found and shown to be $N$-independent: \cite{Audenart2002,EisertReview}:
\begin{align}
\label{eq: LogNegBound}
 E_N=\frac{1}{2}\log\bigg(\frac{\norm{\mb H_{xx}}^{1/2}}{\lambda_{\text{min}}(\mb H_{xx})^{1/2}}\bigg).
\end{align}
Again, the divergence of the entanglement bound with the energy gap-closing stands out. The area law has been shown \cite{StatDependenceCramer} to hold even in critical bosonic lattices for $D>1$. However, findings from conformal theory have shown that 1$D$ critical bosonic systems host a logarithmic divergence in EE \cite{EisertReview}. 

\subsubsection{Quantum fidelity and metric tensor.} 
Alongside entanglement measures, information-geometric tools have also become attractive methods for diagnosing quantum phase transitions. At the most basic level, a fundamental measure of the overlap between two arbitrary quantum states is provided by the {\em (Bures) quantum fidelity}, $\mathscr{F}(\rho_1,\rho_2) \equiv \text{Tr}\big( \sqrt{\sqrt{\rho_1} \,\rho_2\sqrt{\rho_1}}\big)$. Both the fidelity and the associated susceptibility, expressing the fidelity response to a change in the parameter(s) driving the system across phase boundaries have been extensively applied to study GS criticality \cite{GuFid}. Notably, the quantum fidelity between two arbitrary Gaussian states can be analytically computed in terms of the difference between their mean vectors $\vec{m}_2-\vec{m}_1$ and their CMs $\mb \Gamma_1,\mb\Gamma_2$ \cite{GaussFidelity}. For zero-mean pure states, the resulting expression can be simplified to a particularly compact form, which will be relevant to our subsequent analysis:
\begin{align}
 \mathscr{F}( \rho_1,\rho_2)= \frac{1}{\sqrt[4]{\det\left((\mb \Gamma_1+\mb \Gamma_2)/2\right)}}  .
 \label{eq:Fid}
\end{align}

A related but more general approach has been developed by considering the differential-geometric structure of a parametrized Hamiltonian family, say, $H(\vec{\lambda})$, with associated GS $|{\psi}(\vec{\lambda})\rangle$, 
where $\vec{\lambda}\equiv \lambda_\mu$ belongs to the parameter manifold $\mathcal{M}$. If $\partial_{\mu}\equiv\frac{\partial}{\partial \lambda^{\mu}}$ denotes differentiation with respect to a Hamiltonian parameter, a measure of distance between GSs corresponding to different parameter values may then be constructed via the {\em quantum metric tensor} (QMT) $g_{\mu\nu}$ \cite{QuamtumMetric}, defined in terms of the real (symmetric) part of the quantum geometric tensor (QGT) $\chi_{\mu\nu}$:
\begin{align}
 \chi_{\mu\nu} &\equiv 
 \langle \partial_{\mu}\psi|\partial_{\nu}\psi\rangle-\langle \partial_{\mu}\psi|\psi\rangle\langle \psi|\partial_{\nu}\psi\rangle.
\end{align}
In \cite{ZanardiInfoGeom}, singular behavior of the QMT was proposed as a faithful indicator of GS quantum phase transitions. This was motivated by the observation that, in the vicinity of a quantum critical point, the distinguishability between GSs corresponding to infinitesimally different Hamiltonian parameters is strongly enhanced, manifesting in a divergence of the QMT at criticality. 

This idea has recently been generalized to pseudo-Hermitian and non-Hermitian Hamiltonians by appropriately modifying the QGT and the QMT \cite{NHQGT,BosonicQGT}. The QGT for pseudo-Hermitian Hamiltonians was found to be:
\begin{align}
\label{eq: QGT}
    \chi^{a b }_{\mu\nu} = (\partial_{\mu}\vec{\beta}^{a})^{\dag}(\partial_{\nu}\vec{\beta}^{b})-(\partial_{\mu}\vec{\beta}^{a})^{\dag}(\vec{\beta}^{b})(\vec{\beta}^{a})^{\dag}(\partial_{\nu}\vec{\beta}^{b}).
\end{align}
Here, $a, b \in\{L,R\}$ label the left- and right- eigenvectors $\vec{\beta}$ of the Bloch dynamical matrix $\mb g(\bm k)$, normalized as $(\vec{\beta}^{a})^{\dag}(\vec{\beta}^{b})=1$, with $a\neq b$, while $\{\mu,\nu\}$ are, as before, related to the parameters of the Hamiltonian. Since, for QBHs, the left- and right- eigenvectors obey $(\vec{\beta}^{L})^{\dag}\equiv(\vec{\beta}^{R})^{\dag}\boldsymbol{\tau}_3$, the QMT is then given by:
\begin{align}
\label{eq: QMT}
g^{ a b }_{\mu\nu}= \frac{  \chi^{a b }_{\mu\nu}+  \chi^{a b }_{\nu\mu}}{2}= \frac{  \chi^{a b }_{\mu\nu}+  \chi^{b a }_{\nu\mu}}{2}.
\end{align}

\section{A new take on criticality}
\label{sec: newtake}

\subsection{The Krein gap and the degeneracy of the QPV manifold}

\begin{figure}[t]
    \centering\hspace*{20mm}
    \includegraphics[width=0.8\columnwidth]{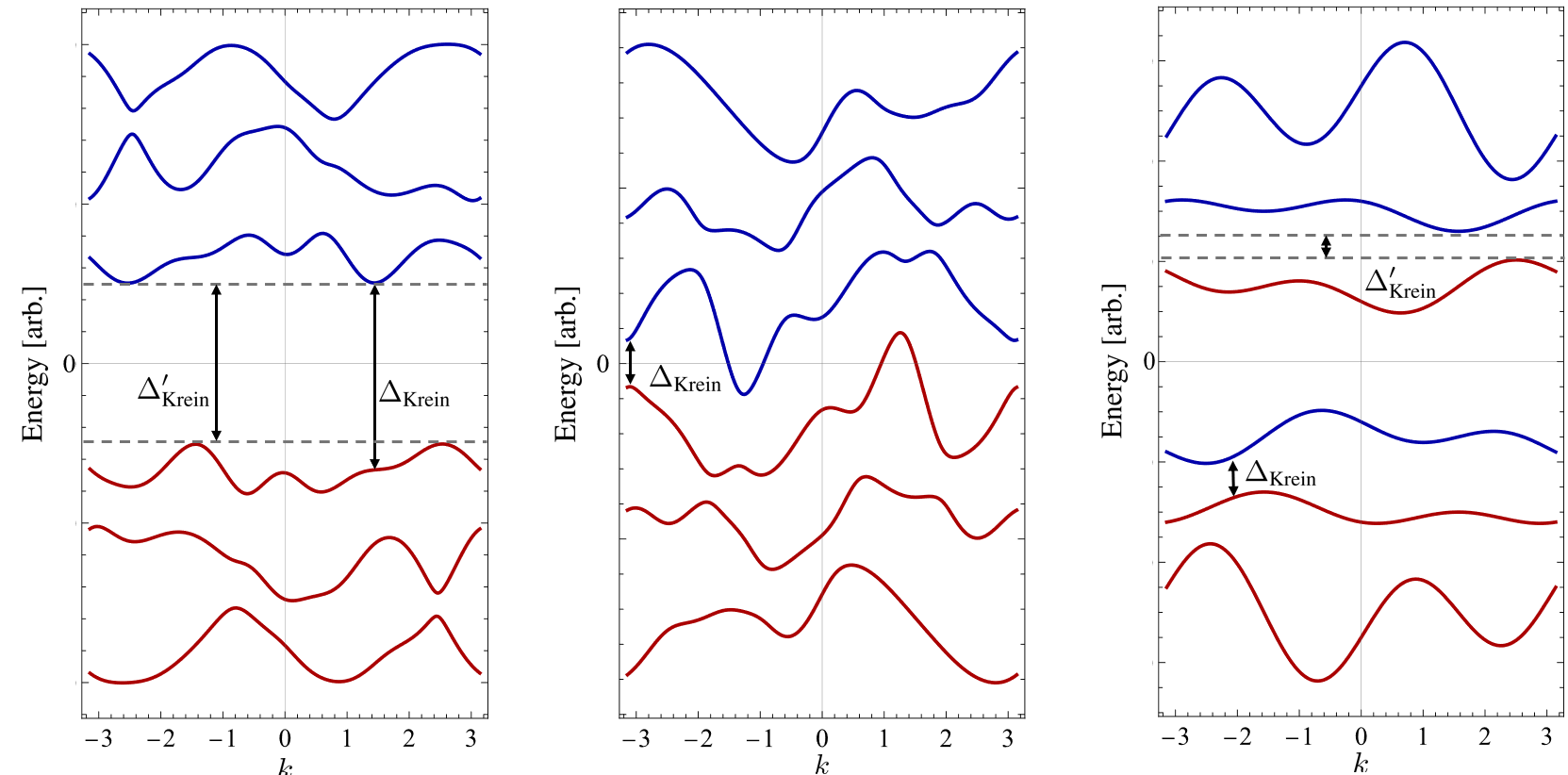}
    \caption{Representative band structures with the direct and indirect Krein gaps labeled. In all cases, the blue curves represent the three quasiparticle energy bands $\{\omega_1(k) \leq \omega_2(k)\leq \omega_3(k)\}$, while the red curves depict the (fictitious) hole energy bands $\{-\omega_1(-k) \geq -\omega_2(-k) \geq -\omega_3(-k)\}$. \textbf{Left:} A thermodynamically stable configuration with both gaps about zero energy. \textbf{Middle:} A thermodynamically unstable configuration with a vanishing indirect Krein gap but nonzero direct Krein gap about zero energy. \textbf{Right:} A thermodynamically unstable configuration with both the indirect and the direct Krein gaps about a nonzero energy. Note that, due to the symmetry of particle and hole bands, i.e., $\omega_n^h(\bm{k}) = - \omega_n^p(-\bm{k})$, the minimal gap between the second-highest particle band and the highest hole band is equal to that of the lowest particle band and the second-highest hole band. 
While all three of these example systems possess a unique translationally invariant QPV, only the left and right systems possess a unique QPV altogether, as $\Delta_{\text{Krein}}'=0$ for the middle system. }
\label{fig: KreinGapCartoon}
\end{figure}

The central idea of this work is to extend the notion of GS criticality to a larger class of QBHs, by situating the QPV into the same role  the many-body GS occupies in the conventional setting. Thus, {\em if a QPV exists -- and is unique -- for a given, possibly thermodynamically unstable QBH, we will say that the QBH is critical if the correlations in this state cannot be exponentially bounded}. This identification is in many ways natural, since in the thermodynamically stable regime the QPV reduces to the many-body GS, recovering the standard notion of criticality. For thermodynamically stable QBHs the existence of a positive many-body gap also implies uniqueness of the GS \cite{Wolf2006}; in our broader setting, we need to determine conditions under which (i) a given QBH has a QPV; and (ii) the QPV is unique. The necessary and sufficient condition for (i), as described in Sec.\,\ref{sec: qvacback}, is dynamical stability. To make headway in addressing (ii), appropriate generalizations of the notion of ``gap closing'' are needed:

\medskip
 
\noindent 
{\bf Definition.} {\em Let $H$ be a dynamically stable, translationally invariant QBH, with associated quasiparticle energies $\{ \pm \omega_n (\pm {\bm{k}}) \}$ [Eq.\,\eqref{eq:gkdiag}]. The \emph{(direct) Krein gap} of $H$ is the quantity} 
\begin{align}
\label{eq: KreinGap}
             \Delta_{\text{Krein}} \equiv \min_{\bm{k},m,n}|\omega_{n}(\bm{k})+\omega_{m}(-\bm{k})|, \quad {\bm{k}} \in {\mathrm{BZ}} . 
\end{align}
{\em Similarly, the \emph{indirect Krein gap} of $H$ is defined by} 
\begin{align}
\label{eq: IndKreinGap}
         \Delta_{\text{Krein}}' \equiv \min_{\bm{k},\bm{k'},m,n}|\omega_{n}(\bm{k})+\omega_{m}(\bm{k'})| \leq \Delta_\text{Krein},
          \quad {\bm{k}}, {\bm{k'}} \in \mathrm{BZ} .  
\end{align}

\medskip

                  {Physically, the quantity $\omega_n(\bm{k})+\omega_m(-\bm{k})$ is the energy, measured from the QPV, of a zero-momentum pair of quasiparticles. These pairs are the only excitations that a translationally invariant quadratic perturbation can create out of the QPV. The direct Krein gap is therefore the minimal energy cost of such a pair.}
                  {Additional} physical intuition into the meaning of the Krein gaps may be gained by introducing the notion of particle and (fictitious) hole bands. The {\em particle bands} are the quasiparticle energy bands, $\omega_n^p(\bm{k}) \equiv \omega_n(\bm{k})$, while the {\em hole bands} are \footnote{We remark that, unlike in the fermionic case, there is no physically meaningful notion of a bosonic hole. If we let $a' = a^\dag$ then $[a',a'{}^\dag] = -1_F$, with $1_F$ being the identity on Fock space. Thus, the CCRs are violated~\cite{Squaring,10.1063/5.0035358}.} $\omega^h_n(\bm{k}) \equiv -\omega_n(-\bm{k})$.
 The Krein gaps can then be conveniently rewritten as
\begin{align*}
    \Delta_\text{Krein} = \min_{\bm{k},m,n}|\omega_{n}^p(\bm{k})-\omega^h_{m}(\bm{k})|,\quad\Delta_\text{Krein}' = \min_{\bm{k},\bm{k'},m,n}|\omega_{n}^p(\bm{k})-\omega^h_{m}(\bm{k'})|.
\end{align*}
Thus, the (direct) Krein gap may be thought as the smallest energy separation between a particle and a hole of the same momentum, while the indirect Krein gap is the smallest energy separation between any particle and hole normal mode. Closing of a Krein gap requires a degeneracy between particle and hole bands, with the degeneracy occurring {\em locally} in $\bm{k}$-space in the direct case.                    {Experimentally, the Krein gap can be resolved whenever the normal-mode dispersion is accessible.} In Fig.\,\ref{fig: KreinGapCartoon}, we depict examples of these quantities for representative band structures in 1$D$. By definition, it follows that the Krein gap can {\em only} become zero at an EP or a (``direct'', fixed-$\bm k$)  KC. According to Sec. \ref{sec: DSKS}, this implies that $\Delta_{\text{Krein}}=0$ {\em only} at the dynamical stability phase boundaries of a QBH. We are now in a position to give our first main result:

\begin{theorem}
\label{thm: UniqueVac}
Let $H$ be a dynamically stable, translationally invariant QBH. Then: 
\begin{itemize}
\item[(i)] If the indirect Krein gap $ \Delta_\text{Krein}'  >0$, there exists a unique QPV. In addition, the QPV is translationally invariant.
\item[(ii)] If the direct Krein gap $\Delta_\text{Krein} >0$, there exists a unique translationally invariant QPV.
\end{itemize}
\end{theorem}

Accordingly, the fact that a QBH $H$ is ``Krein-gapped''  in the thermodynamic limit is {\em sufficient} for a unique translationally invariant QPV to exist. While we defer a proof of the above theorem and a discussion of the converse implication to Appendix \ref{app: UniqueVac}, we note that the statement in part $(ii)$ is more specialized and, in line with the definition of the direct Krein gap, is tied to translational symmetry. Since, in fact, the concept of the indirect Krein gap is not inherently linked to any specific symmetry, a statement akin to part $(i)$ holds more generally: namely, the indirect Krein gap being open implies uniqueness of the QPV for {\em any} dynamically stable QBH. Since in this work we focus on translationally invariant QPVs, we restrict our attention to the {\em direct} Krein gap in what follows. Consequently, the main question is whether the manifold of translationally invariant QPVs exhibits degeneracy, which part $(ii)$ addresses.

\subsection{The Krein gap and the QPV correlations}
\label{Sec: CorrelationsQPV}

The uniqueness of the QPV of a translationally invariant QBH away from dynamical stability phase boundaries grants us the freedom to bypass the constraint of thermodynamic stability, by lifting the calculation of correlation functions from the GS to the QPV. Our goal is then to determine necessary conditions for a QBH to possess a QPV that supports long-range correlations. We do so by identifying sufficient conditions for arbitrary QPV correlation functions to be exponentially bounded.

To establish boundedness of the correlations, we leverage known analyticity results from Fourier analysis. Specifically, it is known that a real multivariate function $f(\bm k)$, periodic in each variable $k_\mu$, is analytic if and only if the corresponding Fourier coefficients are exponentially bounded: $|f_{\bm r}|\leq A \,e^{-\eta |\bm r|}$, $A\in {\mathbb R}$ \cite{DynSysChaos}. It follows that real-space correlation functions are exponentially bounded if and only if the corresponding momentum-space correlations are analytic in $\bm k$.

Let $H$ be a dynamically stable Krein-gapped QBH and $\mathbf{C}(\bm{k}) = \mathbf{L}(\bm{k})\mathbf{L}^\dag(\bm{k})$ be the QPV momentum-space CM defined in Eq.\,\eqref{eq: Ck} (for brevity we will omit the superscript ``qpv" henceforth). Recall from Eq.\,\eqref{eq:gkdiag} that $\mathbf{L}(\bm{k})$ is a pseudo-unitary matrix that diagonalizes $\mathbf{g}(\bm{k})$. This, in terms of the eigenvectors $\vec{\beta}_{n,\pm}(\bm{k})$ in Eq.\,\eqref{eq: Geig} (which form the columns of $\mathbf{L}(\bm{k})$), gives
\begin{align}
\label{eq: Ckbeta}
    \mathbf{C}(\bm{k}) = \sum_{n=1}^d\left[\vec{\beta}_{n,+}(\bm{k})\vec{\beta}_{n,+}^{\,\dag}(\bm{k})+\vec{\beta}_{n,-}(\bm{k})\vec{\beta}_{n,-}^{\,\dag}(\bm{k})\right].
\end{align}
If the eigenbasis $\{\vec{\beta}_{n,\pm}(\bm{k})\}$ constituted an orthonormal basis, Eq.\,\eqref{eq: Ckbeta} would provide a resolution of the identity. This, however, is not the case whenever $\mathbf{g}(\bm{k})$ is non-normal. Instead, the conditions in Eq.\,\eqref{eq: Geig} provide an alternative resolution of the identity:
\begin{align}
\label{eq: resid}
    \bm{1}_{2d} = \sum_{n=1}^d\left[\vec{\beta}_{n,+}(\bm{k})\vec{\beta}_{n,+}^{\,\dag}(\bm{k})-\vec{\beta}_{n,-}(\bm{k})\vec{\beta}_{n,-}^{\,\dag}(\bm{k})\right]\bm{\tau}_3.
\end{align}
With Eqs.\,\eqref{eq: Ckbeta} and \eqref{eq: resid} in hand, we define a \textit{QPV Krein projector} by letting 
\begin{align}
\label{eq: Pk}
    \mathbf{P}(\bm{k}) \equiv \frac{1}{2}\big(\bm{1}_{2d} + \mathbf{C}(\bm{k})\bm{\tau}_3\big) = \sum_{n=1}^d\vec{\beta}_{n,+}(\bm{k})\vec{\beta}_{n,+}^{\,\dag}(\bm{k})\bm{\tau}_3.
\end{align}
This operator is a \textit{pseudo-orthogonal projector} (or $\bm{\tau}_3$-orthogonal projector \cite{GohbergIndefiniteLA}), i.e., it satisfies $\mathbf{P}(\bm{k})^2 = \mathbf{P}(\bm{k})$ and $\mathbf{P}^\dag(\bm{k}) = \bm{\tau}_3\mathbf{P}(\bm{k})\bm{\tau}_3$. More specifically, it projects onto the (direct sum of) particle eigenspaces: $\mathbf{P}(\bm{k})\vec{\beta}_{n,s} = \delta_{s,+}\vec{\beta}_{n,s}$. Importantly, Eq.\,\eqref{eq: Pk} completely characterizes the QPV correlation matrix and its dependency on $\bm{k}$ in terms of a projector which has an equivalent representation as a Riesz projector, namely:
\begin{align}
\label{eq: Riesz}
    \mathbf{P}(\bm{k}) = \frac{1}{2\pi i}\oint_\gamma \big(z\bm{1}_{2d}-\mathbf{g}(\bm{k})\big)^{-1}\,dz,
\end{align}
where $\gamma(\bm{k})$ is any rectifiable contour containing only the particle excitation energies $\{\omega_n(\bm{k})\}$ for some fixed $\bm{k}$. This key technical result facilitates establishing our main theorem: 

\begin{theorem}
\label{th:MainResult} 
Let $H$ be a dynamically stable, translationally-invariant QBH. If the Krein gap is open, $\Delta_{\text{Krein}}>0$ then the QPV is unique and all correlations are exponentially bounded.
\end{theorem}

\begin{proof}
If the Krein gap is open, then (i) the QPV is unique by Theorem \ref{thm: UniqueVac} and (ii) for any fixed $\bm{k}$, the set of particle eigenvalues $\{\omega_n(\bm{k})\}$ of $\mathbf{g}(\bm{k})$ is isolated. Thus, it is possible to smoothly choose a contour $\gamma(\bm{k})$ as described under Eq.\,\eqref{eq: Riesz}.  Now, $\mathbf{g}(\bm{k})$ is always analytic in $\bm{k}$ for systems with finite-range couplings. It follows from Eq.\,\eqref{eq: Riesz} that $\mathbf{P}(\bm{k})$, and therefore $\mathbf{C}(\bm{k})$, is analytic in $\bm{k}$. Thus, $\mathbf{C}_{\bm r}$ is exponentially bounded in view of the above-mentioned Fourier-analysis result. 
\end{proof}

We remark that the converse statement, namely, ``if the QPV is unique and all correlations are exponentially bounded then $\Delta_\text{Krein}>0$" is true if either of the following assumptions are added: (i) there are no KCs or (ii) the system is finite (see also Appendix \ref{app: UniqueVac}). The above theorem motivates a broader notion of criticality that leaves behind the GS in favor of the QPV, as anticipated. Specifically, it indicates that the Krein gap serves as an appropriate analogue of the many-body spectral gap within the class of dynamically stable, translationally invariant QBHs, and that long-range correlations in the QPV can emerge {\em only} at dynamical stability phase boundaries, where $\Delta_{\text{Krein}} = 0$. Since this may happen at both an EP and a KC, it is worth further examining what differences these two scenarios may engender.

\subsection{QPV correlations specialized to single-band models}

In order to isolate the essential features, we proceed with the case of a single DOF per site, i.e., $d=1$. Specializing to $d=1$ sharpens the picture considerably: the QPV CM  admits a closed-form expression in terms of the four real functions $d_0(\bm k), \ldots, d_3(\bm k)$, EPs and KCs are cleanly separated by a simple geometric criterion, and the parameter $d_0(\bm k)$ that decouples thermodynamic from dynamical stability is explicitly identified.

As a first step, we express the Bloch dynamical matrix in the form 
\begin{align*}
       \mb g(\bm k)\equiv \begin{pmatrix}
       d_0(\bm k)+d_3(\bm k)&i d_1(\bm k)+d_2(\bm k) \\
    i d_1(\bm k)-d_2(\bm k)   &d_0(\bm k)-d_3(\bm k)
    \end{pmatrix}, \qquad d_j(\bm k)=\sum\limits_{|\bm r|<R}d_{j\bm r}\, e^{i\bm k\cdot \bm r}.
\end{align*}
Physically, the real-space coefficients $d_{3\bm r}$ and $d_{0\bm r}$ encode real and imaginary hopping terms, while $d_{2\bm r}, d_{1\bm r}$ encode real and imaginary pairing terms, respectively. As we show in Appendix \ref{App: BlochMat}, the dynamical stability condition amounts to
\begin{align}
     \mathcal{E}(\bm k)^2
   = d_3(\bm k)^2-d_2(\bm k)^2-d_1(\bm k)^2\geq 0, \quad \forall {\bm k}, 
   \label{dynStab}
\end{align}
with $\mb g(\bm k)$ diagonalizable and $\mathcal{E}(\bm k)\equiv\tfrac{1}{2}( \omega(\bm k)+\omega(-\bm k))$ denoting the separation between the particle and the hole bands. The boundary of the dynamically stable regime is the set of parameters such that $\mathcal{E}(\bm k')=0$ for one or more momentum values $\bm k'$. On this boundary, the system remains dynamically stable unless $\mb g(\bm k')$ fails to be diagonalizable. 

On the other hand, thermodynamic stability is found to hold when
\begin{align}
\label{thStab}
\hspace{0.5em}d_3(\bm k)\geq 0 \text{ and }   \mathcal{E} (\bm k)^2-d_0(\bm k)^2\geq 0, \quad \forall \bm k, \quad \text{or}\quad 
 \hspace{0.4em}d_3(\bm k)\leq  0 \text{ and }   \mathcal{E} (\bm k)^2-d_0(\bm k)^2\geq 0, \quad \forall \bm k, 
\end{align}
where the two possibilities correspond to a system with energies bounded from below or from above, respectively. It follows that a dynamically stable QBH may be taken from a thermodynamically stable to a thermodynamically unstable phase by tuning $d_0(\bm k)$ \footnote{We note, in passing, that a harmonic lattice can only have non-zero $d_3(\bm k)$ and $d_2(\bm k)$ terms; hence, if it is dynamically stable, it is also thermodynamically stable if $d_3(\bm k)\geq 0$, or if $d_3(\bm k)\leq 0$, for all $\bm k$.}. Going a step further, the CM of the QPV in momentum space can be seen to attain the following form (see Appendix \ref{App: CMfromModal1} for detail): 
\begin{align}
\label{GammaQPV}
   \mb \Gamma^\text{qpv}(\bm k)&= \frac{\text{sgn}(d_3(\bm k))}{\mathcal{E}(\bm k)}\begin{pmatrix}
        d_3(\bm k)-d_2(\bm k)&-d_1(\bm k)\\
        -d_1(\bm k)& d_3(\bm k)+d_2(\bm k)
    \end{pmatrix}.
\end{align}

\noindent 
Based on the above findings, we are now in a position to draw a few key conclusions:

 \begin{enumerate}
   \item[1)] The absence of the parameter that tunes thermodynamic stability, $d_0(\bm k)$, implies that the QPV correlation functions are \textit{insensitive to the loss of thermodynamic stability}. Therefore, the CM of the QPV must be identical to the one of the GS in a regime where the QBH is thermodynamically stable. This is exactly what we find: first, even when thermodynamic stability is lost, $\mb \Gamma^\text{qpv}$ remains a {\em valid} CM describing a pure quantum state (as we formally prove in Appendix \ref{app:ValidCovMat}); second, the QPV CM in Eq.\,\eqref{GammaQPV} is structurally identical to the one describing the GS in Eq.\,\eqref{eq:WolfCovMat}.

  \item[2)] The gap of the symmetrized Hamiltonian in Eq.\,\eqref{eq:WolfCovMat}, introduced in \cite{Wolf2006} to track the onset of GS criticality, acquires a concrete physical significance, namely, it quantifies the distance to a nearest dynamical stability phase boundary. Mathematically, it relates directly to the Krein gap, 
$$ \Delta^\text{gs}_{\text{sym}}=  
\lambda_{\text{min}}\left(\boldsymbol{ \mathcal{E}}\right)=\underset{\bm k}{\text{min }} \mathcal{E}(\bm k)= \tfrac{1}{2} \, \Delta_{\text{Krein}}. $$
\end{enumerate}

\noindent 
It is useful to note that the entries of the momentum-space CM  take the form
\begin{align}
\label{entries}
f(\bm{k}) \in  \left\{ \frac{d_3(\bm{k}) \pm d_2(\bm{k})}{\mathcal{E}(\bm{k})},
\, \frac{d_1(\bm{k})}{\mathcal{E}(\bm{k})}
\right\},
 \quad \mathcal{E}(\bm{k}) = \sqrt{d_3(\bm{k})^2 - d_2(\bm{k})^2 - d_1(\bm{k})^2},\quad \forall \bm{k}\in \mathrm{BZ}.
\end{align}
It is illuminating to visualize the structure of $\mathcal{E}(\bm k)^2$ in the three-dimensional parameter space spanned by $(d_1(\bm k), d_2(\bm k), d_3(\bm k))$ at a fixed $\bm k$. The equation $\mathcal{E}(\bm k)^2 = d_3(\bm k)^2 - d_2(\bm k)^2 - d_1(\bm k)^2$ defines a cone: its interior, $\mathcal{E}(\bm k)^2 > 0$, is the dynamically stable region, while its surface, $\mathcal{E}(\bm k)^2 = 0$, marks the boundary where dynamical stability is lost. On the conical surface away from the apex, $\mathcal{E}(\bm k) = 0$ with at least one $d_i(\bm k) \neq 0$ -- these points are EPs. The apex itself, where 
$d_1(\bm k) = d_2(\bm k) = d_3(\bm k) = 0$ simultaneously, is a KC. Physically, this translates into the following conclusions:

\begin{itemize}
\item[3)] 
At an EP $\bm{k} = \bm{k}_c$, the Krein gap closes, $\Delta_{\text{Krein}} = 2\mathcal{E}(\bm{k}_c)= 0$, but not all $d_i(\bm{k}_c)$ vanish. Consequently, for at least one CM component $f(\bm{k})$, the numerator is finite, whereas the denominator vanishes, implying that $f(\bm{k})$ is non-analytic at $\bm{k}_c$. As noted, such a non-analyticity in $\bm{k}$-space is a sufficient condition for the absence of exponentially decaying correlations in real space.
    
\item[4)] 
At a KC, the Krein gap closes, $\Delta_{\text{Krein}} = 2\mathcal{E}(\bm{k}_c)= 0$, but now all $d_i(\bm{k}_c)$, $i=1,2,3$ 
vanish simultaneously, causing {\em both} the numerator and denominator of $f(\bm{k})$ in Eq.\,\eqref{entries} to vanish. Consequently, $\lim\limits_{\bm{k}\to \bm{k}_c} f(\bm{k})$ may not be well-defined, and the behavior of real-space correlation functions may depend sensitively on the direction of approach in parameter space, $\vec{\alpha}\to \vec{\alpha}_c$, with $\vec{\alpha}$ denoting the tunable Hamiltonian parameters, and $\Delta_{\text{Krein}} (H(\vec{\alpha}_c))=0$.
\end{itemize}

With that, the emerging picture for our generalized notion of criticality can then be summarized as follows:

\smallskip

\noindent {\bf Claim: Generalized criticality from Krein-gap closing.} {\em In a QBH with a finite number of Krein-gap closings, EPs mark critical boundaries where QPV correlations cease to decay exponentially and become long-range. KCs exhibit a more complex critical behavior, for which the scaling and divergence of correlation functions can vary depending on the direction of approach.}

\smallskip

The above is consistent with the behavior seen in equilibrium critical phenomena; in particular, multicritical points arise from the intersection of two or more critical boundaries and, corresponding to the fact that multiple relevant parameters must be tuned simultaneously, different diverging correlation lengths may be associated to different paths in parameter space, or divergence is only seen if the point is approached along the critical line \cite{Sachdev,Deng,Patra}. In the following section, we will explore and reinforce these ideas with illustrative examples.                   {We remark that the complementary case, in which $\Delta_{\rm Krein}$ vanishes on a continuum of momenta, is more subtle. As shown in \cite{MariamNext}, it gives rise to \textit{conditional criticality}, whereby correlations may be short- or long-range, depending on the choice of QPV.}

\section{Illustrative models}
\label{sec:models}

\subsection{Warm-up: A simple
harmonic-chain model}
\label{sub:PC}

Let us first apply the framework we have developed to a quintessential example of a model displaying GS criticality, a 1D harmonic lattice which describes a phonon chain, and can be equivalently thought of as a discretized Klein-Gordon field. The relevant QBH reads \cite{Dualities}:
\begin{align}
\label{eq:GHC}
        H&=\frac{1}{2}\sum_{j\in \mathbb{Z}}\Big(\Omega(p_j^2+x_j^2)-2 J x_jx_{j+1}\Big)\\
    \nonumber&=\frac{1}{2}\sum_{j\in \mathbb{Z}}\left(\Omega(a_j^{\dag}a_j+a_ja_j^{\dag})-J(a_{j+1}^{\dag}a_j^{\dag}+a_{j+1}^{\dag}a_j+\text{H.c.})\right), 
    \quad \Omega\geq 2J\geq 0,
\end{align}
whose many-body energy gap $\Delta$ is finite for all $\Omega>2J$ and becomes zero when $\Omega=2J$. From the dynamical matrix (which we explicitly give in Appendix \ref{App: GHCCovMat}), we determine that for $\Omega>2J$, the system is dynamically stable, while at $\Omega=2J$ it develops an EP. Thermodynamic stability holds throughout $\Omega\geq 2J$. Therefore, the procedure we have described for computing QPV correlations amounts to computing correlation functions in the GS for this model. Plugging in the appropriate parameters $\big(d_3(k)=\Omega-J \cos(k), d_2(k)=-J\cos(k)\big)$ into Eq.\,\eqref{GammaQPV} yields, after some simplification, the following expressions: 
\begin{align}
    \langle x_{j}x_{j+r}\rangle
    \sim \int^{\pi}_0 \frac{dk }{2\pi}\frac{1}{\sqrt{1-\alpha\cos (k)}} \, \cos(k r),     \quad 
    \langle p_{j}p_{j+r}\rangle
    \sim \int^{\pi}_0 \frac{dk }{2\pi}\sqrt{1-\alpha \cos (k)}\cos(kr), \quad \alpha\equiv \frac{2J}{\Omega}, 
\end{align}
where we have identified $\mathcal{E}(k)=                  {\Omega}\sqrt{1-\alpha \cos (k)}$ and the off-diagonal blocks of the CM are identically zero, since the model is harmonic. The Krein gap, which coincides with the ordinary spectral gap (up to a factor of 2), $\Delta_\text{Krein} = 2 \text{min}_k \,\mathcal{E}(k) = 2                  {\Omega}\sqrt{1-\alpha}$, is open for $0\leq \alpha<1$, while it closes for $\alpha=1$.
Since this is a thermodynamically stable system, the transition from an open Krein gap to a closed one is, in fact, just a transition from a phase with a finite many-body energy gap to one with a zero many-body energy gap. Thus, critical behavior is expected to emerge at $\alpha=1$. Following this thread, let us examine the above correlations more closely at different points in the parameter space. 

For $0\leq \alpha<1$, the integrands are real and analytic for all $k$, hence, as remarked, bounded by decaying exponentials as in Eq.\,\eqref{eq:HarmonicDecay}. The associated correlation length may be determined by letting $\xi \equiv 1/\eta$, where $\eta$ denotes the largest value for which the integrands admit an analytic continuation into the complex plane under $k\mapsto k+i\eta$. Setting $\Delta_\text{Krein}=0$, the edge of real-analyticity is marked by $\sqrt{1-\alpha \cos (k)}=0$, which yields $k =\pm\arccos(1/\alpha)$. Without loss of generality, let $k^{*}\equiv \arccos(1/\alpha)$, and $\tilde{\alpha}\equiv 1/\alpha$. We then have 
\begin{align*}
 \eta = \text{Im}(k^*)= \text{Im}(\arccos(\tilde{\alpha}))=\text{ Im}(-i\log(\sqrt{\tilde{\alpha}^2-1}+\tilde{\alpha}))=-\log(\sqrt{\tilde{\alpha}^2-1}+\tilde{\alpha}),
\end{align*}
where we used $\sqrt{\tilde{\alpha}^2-1}+\tilde{\alpha}>0$. Finally, this yields $\xi=1/\eta=\frac{1}{2(1-\alpha)^{1/2}}+\mathcal{O}(1-\alpha)^{1/2}$, meaning that, as the Krein-gap closing point is approached, the correlation length diverges with a critical exponent of $1/2$. Following the definition of the dynamic critical exponent $z$ in standard critical phenomena \cite{Sachdev}, it is natural to let  $\Delta_\text{Krein}\sim\xi^{-z}$ near a gap-closing point. Here, $\Delta_{\text{Krein}}=2\sqrt{1-\alpha}$, together with $\xi\sim (1-\alpha)^{-1/2}$. Accordingly, $z=1$.

When $\alpha=1$, the position correlation functions $\langle x_{j}x_{j+r}\rangle$ are unbounded for all $r$, whereas 
\begin{align*}
    \langle p_{j}p_{j+r}\rangle
    &\sim \int^{\pi}_0 \frac{dk }{2\pi}\sqrt{1-\cos (k)}\cos(kr)\sim \frac{1}{1-4r^2}.
\end{align*}
As expected, something dramatic happens at the gap-closing point -- with one set of correlations becoming unbounded, and the other set displaying algebraic decay with distance $r$. This is exactly the behavior of correlations reported for this model (up to slight changes in parametrization) in \cite{Wolf2006}. 

{\em Remark.---} We stress that the above results on correlations are {\em not} tied to the thermodynamic stability of the simple harmonic chain we considered. To make this point concrete, consider modifying the model in Eq.\,\eqref{eq:GHC} by adding an imaginary-hopping coupling term (see also Ref.~\cite{DynamicalMetastability}): 
\begin{align}
\label{eq:GHCwImHopHam}
    H\mapsto H+\frac{i\gamma}{2}\sum\limits_{j\in \mathbb{Z}}\,\big(a_{j+1}^{\dag}a_j-a_{j}^{\dag}a_{j+1}\big)\quad  \text{or}\quad  H\mapsto H+\frac{\gamma}{2}\sum\limits_{j\in \mathbb{Z}}\,\big(p_{j+1}x_j-x_{j+1}p_{j}\big).
\end{align}
The Bloch dynamical matrix is modified according to $ \mb g(k)\mapsto \mb g(k)+\gamma \sin(k) \mb 1_2.$ From conditions derived in Sec.\,\ref{Sec: CorrelationsQPV}, we know that for a large enough $\gamma=\gamma_{c}$ (see also Appendix \ref{App: GHCCovMat}), this additional term causes the system to transition from being  thermodynamically stable to unstable, but has no effect on dynamical stability. Crucially, however, we also know that this modification \textit{leaves all the correlation functions unchanged}. Thus, for $\Omega>2J$, all correlations are exponentially bounded, while at $\Omega=2J$, this system hosts long-range, algebraically decaying momentum correlations, just like the original model, despite lacking thermodynamic stability beyond $\gamma>\gamma_c$.

\subsection{An interpolation model}
\label{sub:im}

\begin{figure}[t]
    \centering
    \includegraphics[width=\textwidth]{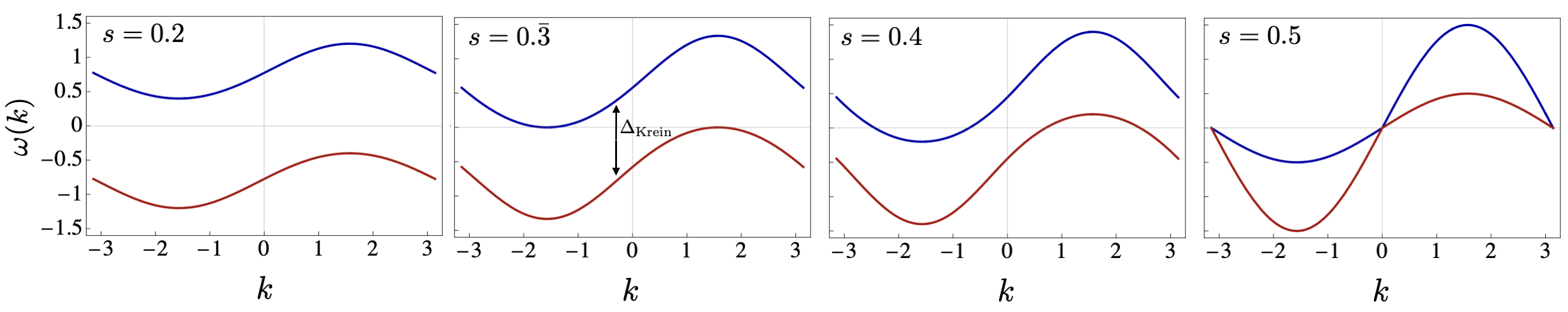}
    \vspace*{-4mm}
    \caption{
Eigenvalues of the interpolation model, Eq.\,\eqref{eq: InterpHam}, with positive (blue) and negative (red) Krein signature.  At $s=s_1 = 1/3$, thermodynamic stability is lost -- the indirect Krein gap closes at zero energy. However, the direct Krein gap remains open until $s=s_2 = 1/2$, where it finally closes for $k=0,\pm \pi$. The values $s_{1}, s_2$ are defined in Eq.\,\eqref{special} and the parameters are $J=2,\Delta=\Omega=1$. }
    \label{fig:BandStruc}
\end{figure}

Thus far, we forayed into a thermodynamically unstable regime by modifying a workhorse of standard bosonic criticality, only to realize a scenario where a thermodynamic stability phase transition takes place, yet the correlation functions are trivially affected. We now wish to take on a more challenging task and start in the opposite limit. In particular, let us consider a nearest-neighbor chain that is, loosely speaking, the opposite of harmonic, i.e., it only has couplings between $x$ and $p$ quadratures (as defined in Eq.\,\eqref{eq: DimLessQuads}):
\begin{align*}
   H= \frac{1}{2}\sum\limits_{j\in \mathbb{Z}}\big((J+\Delta)x_jp_{j+1}-(J-\Delta)p_jx_{j+1}\big), \quad J > \Delta.
\end{align*}
Unfortunately, this QBH model is not only thermodynamically, but also dynamically unstable. In fact, the conditions in Eqs.\,\eqref{dynStab}-\eqref{thStab} can be used to show that this is true for any model that lacks harmonic terms. To remedy this, we introduce, again, a simple harmonic contribution:
\begin{align}
\label{eq: InterpHam}
   H= \frac{1}{2}\sum\limits_{j\in \mathbb{Z}}\Big[ (1-s) \, \Omega
    (x_j^2+p_j^2)+s\,\big((J+\Delta)x_jp_{j+1}-(J-\Delta)p_jx_{j+1}\big)\Big], \quad s\in[0,1], \quad \Omega>0. 
\end{align}
In this way, we can interpolate between a model that is fully stable, both thermodynamically and dynamically, for $s=0$, and one -- the bosonic Kitaev chain \cite{ClerkBKC, Decon} -- which is known to be both thermodynamically and dynamically unstable under bi-infinite BCs, for $s=1$. Up to reparameterization, the above QBH coincides with the ``Kitaev oscillator chain'' introduced in \cite{DynamicalMetastability} and studied in connection with finite-size scaling.
 
The ``interpolation model'' defined in Eq.\,\eqref{eq: InterpHam} displays a rich stability phase diagram, which may be characterized by computing the relevant dynamical matrix and using the conditions in Eqs.\,\eqref{dynStab}-\eqref{thStab}, as further discussed in Appendix \ref{App: InterpCovMat}. The Krein gap                   {2$\sqrt{\Omega^2(1-s)^2-s^2\Delta^2}$} (see also Fig.\,\ref{fig:BandStruc} for a depiction of the quasiparticle energy bands), and two special points play an important role in the discussion:
\begin{align}
s_1\equiv \frac{1}{1+J/\Omega}, \qquad s_2\equiv \frac{1}{1+\Delta/\Omega}. 
\label{special}
\end{align}
As it turns out, assuming that $J> \Delta$, the model is thermodynamically stable only in the region $s\in [0,s_1]$, but retains dynamical stability all the way up to $s\in [0,s_2)$, with $s_2$ being an EP. This stability phase diagram thus allows us to analyze the behavior of the correlation functions both in the thermodynamically stable and unstable phases, to assess whether there is a meaningful transition between the two. By using the explicit expressions given in Appendix \ref{App: InterpCovMat}, along with Eq.\,\eqref{GammaQPV}, and                   {letting $\alpha^2\equiv \frac{\Delta^2 s^2}{\Omega^2 (1-s)^2}$} we obtain:
\begin{align}
\label{eq:interpcorrel}
\begin{pmatrix}
        (\mb \Gamma_{xx})_{j,j+r}&  (\mb \Gamma_{xp})_{j,j+r}\\
         (\mb \Gamma_{px})_{j,j+r}     &  (\mb \Gamma_{pp})_{j,j+r}
    \end{pmatrix}&={\int^{\pi}_{-\pi}\frac{dk}{2\pi} e^{ikr}  \underbrace{\frac{1}{\sqrt{1-\alpha ^2 \cos ^2(k)}}\left(
\begin{array}{cc}
1 & -\alpha \cos (k) \\
-\alpha \cos (k )& 1\\
\end{array} \right)}_{\mb \Gamma^\text{qpv}(k)}.}
\end{align}

For $s<s_2$, or $\alpha<1$, the Krein gap is open throughout, whereas for $\alpha=1$, the Krein gap vanishes at $k=0, \pm \pi$. Thus, for $0\leq \alpha <1$, all integrands in Eq.\,\eqref{eq:interpcorrel} are real-analytic and the correlation functions are bounded by decaying exponentials. Note, in particular, that the point of loss of thermodynamic stability, $s_1$, as well as the entire many-body gapless region $s\in (s_1,s_2)$, is non-critical and there is {\em no} qualitative change in the correlation functions between $s<s_1$ and $s_2>s>s_1$. 

As $\alpha\rightarrow 1$, the correlation length $\xi$ in the exponential bounds becomes larger and larger. Since the denominator for all the above momentum-space correlation functions is $\sqrt{1-\alpha^2\cos^2(k)}$, the edge of real-analiticity is determined by $\sqrt{1-\alpha^2\cos^2(k)}=0$. Solving for $k$ yields $k=\pm\arccos(\pm 1/\alpha)$ and, as in Sec.\,\ref{sub:PC}, we may let $k^{*}\equiv \arccos(1/\alpha)$ without loss of generality. Similar to the harmonic-chain model, the correlation length is then found to diverge as $\xi=\frac{1}{2 (1-\alpha)^{1/2}}+\mathcal{O}\left((1-\alpha)^{1/2}\right)$. 
By further writing $\Delta_{\text{Krein}}=\sqrt{1-\alpha^2} \rightarrow\sqrt{2}(1-\alpha)^{1/2}+\mathcal{O}((1-\alpha)^{3/2})$, we find, as before, $z=1$.

\begin{figure}[t]
\centering \hspace*{15mm}
    \includegraphics[width=0.9\linewidth]{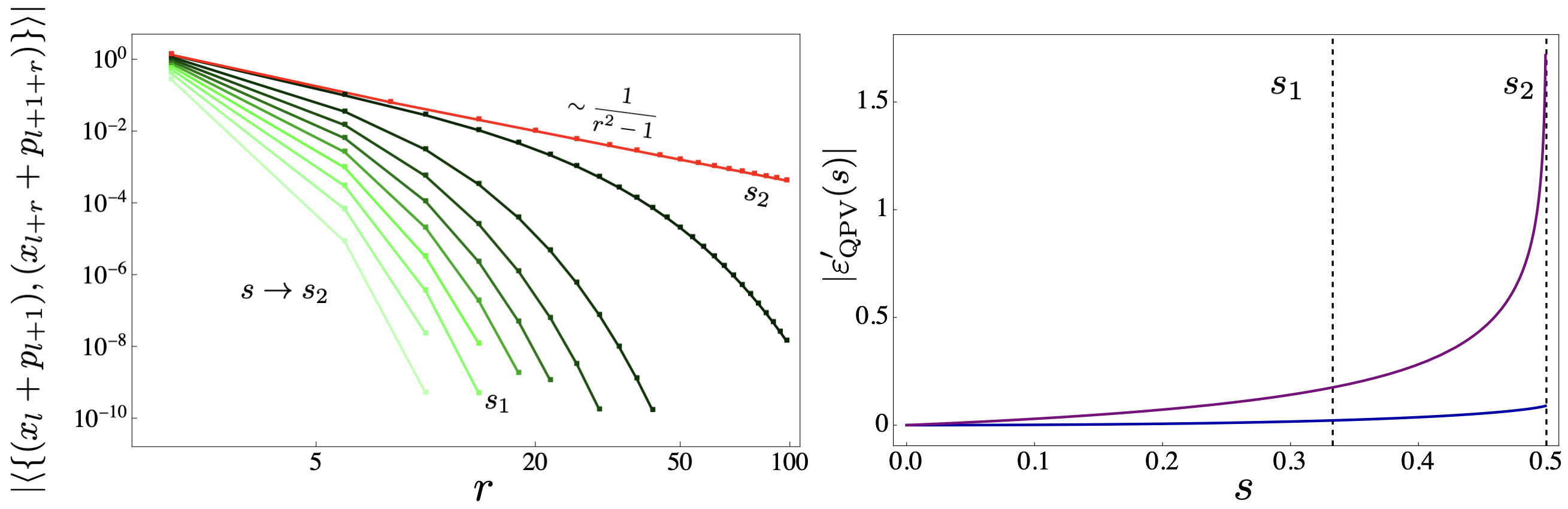}
    \vspace*{-4mm}
\caption{\textbf{Left:} Correlation functions for the interpolation model from Eq.\,\eqref{eq:algebraiccorrelinterp}, for different values of $s$. For $s<s_2$, the darker the color, the larger the value of $s$. The curve corresponding to $s= s_2$ is shaded in red. The exponential decay persists throughout (before and after the loss of thermodynamic stability at $s_1$), with the correlation length diverging as $s\rightarrow s_2$. Correlations are algebraic at $s_2$. \textbf{Right:} Absolute values of the QPV energy density $\varepsilon_\text{qpv}$ (blue, from Eq.\,\eqref{eq:QPVenergydensity}) and its first derivative (purple). The derivative diverges as $\Delta_{\text{Krein}}\rightarrow 0$ with $s\rightarrow s_2$. Here, $J=2,\Delta=\Omega = 1$.}
\label{fig: interpcorrelat}
\end{figure}

Precisely at $\alpha=1$, the Krein gap vanishes both for $k=0$ and $\pm \pi$, causing all integrals in Eq.\,\eqref{eq:interpcorrel} to diverge for all $r$. This is in contrast with the simple harmonic chain, for which $\Delta_\text{Krein}$ closes at a single $k$-point and, therefore, while the $xx$ correlations diverge, the $pp$ correlations remain finite. In the latter case, it is immediate to relate this behavior to the linear symmetries (the ``zero modes'' \cite{PostBosoranas}) of the system at criticality and the expectation that the GS should be invariant. At its critical point, the simple harmonic chain Hamiltonian commutes with the total momentum $\sum_{j\in \mathbb{Z}} p_j$, and the $xx$ correlations in an invariant GS can only accommodate this symmetry by formally taking the value infinity. Similarly, the zero modes of the interpolation model at $\alpha=1$,  which take the form $\sum_j\frac{x_j+p_j}{\sqrt{2}}$ and $\sum_j(-1)^j\frac{x_j-p_j}{\sqrt{2}}$, force the $xx,xp$, and $pp$ correlators to diverge. Nonetheless, it is possible to find a linear combination of these simple correlators that does remain invariant under the linear symmetries of the model and also cancels out the divergence in the integral at the critical point. Indeed, we identify an algebraically decaying correlator: 
\begin{align}
\label{eq:algebraiccorrelinterp}
  \langle\{ (x_{j}+p_{j+1}), (x_{j+r}+p_{j+1+r}) \}\rangle 
&=
\int^{\pi}_0\frac{dk}{\pi}\frac{2\cos(kr)-\alpha \cos(k)[\cos(k(r+1))+\cos(k(r-1))]}{\sqrt{1-\alpha^2\cos^2 (k)}}\bigg\vert_{\alpha=1}\\
\nonumber&=2\int^{\pi}_0\frac{dk}{\pi}\sin (k) \cos(kr)=-\frac{2( \cos (\pi  r)+1)}{\pi(r^2-1)}.
\end{align}
The behavior of these correlations is illustrated in Fig.\,\ref{fig: interpcorrelat} (left panel), which reveals that the correlation functions remain exponentially bounded across both thermodynamically stable and unstable regimes, with no qualitative distinction between them. Long-range correlations appear \emph{solely at the dynamical stability phase boundary}. Moreover, as seen in the right panel of Fig.\,\ref{fig: interpcorrelat} a divergence in the derivative of the QPV energy density is seen to emerge at the dynamical (rather than the thermodynamical) stability phase boundary -- reinforcing that the former constitutes the true critical phase boundary in QBHs.

\subsection{A double harmonic-chain model}
\label{sub:dhc}

\begin{figure}[t]
    \centering \hspace*{15mm}
    \includegraphics[width=0.85\textwidth]{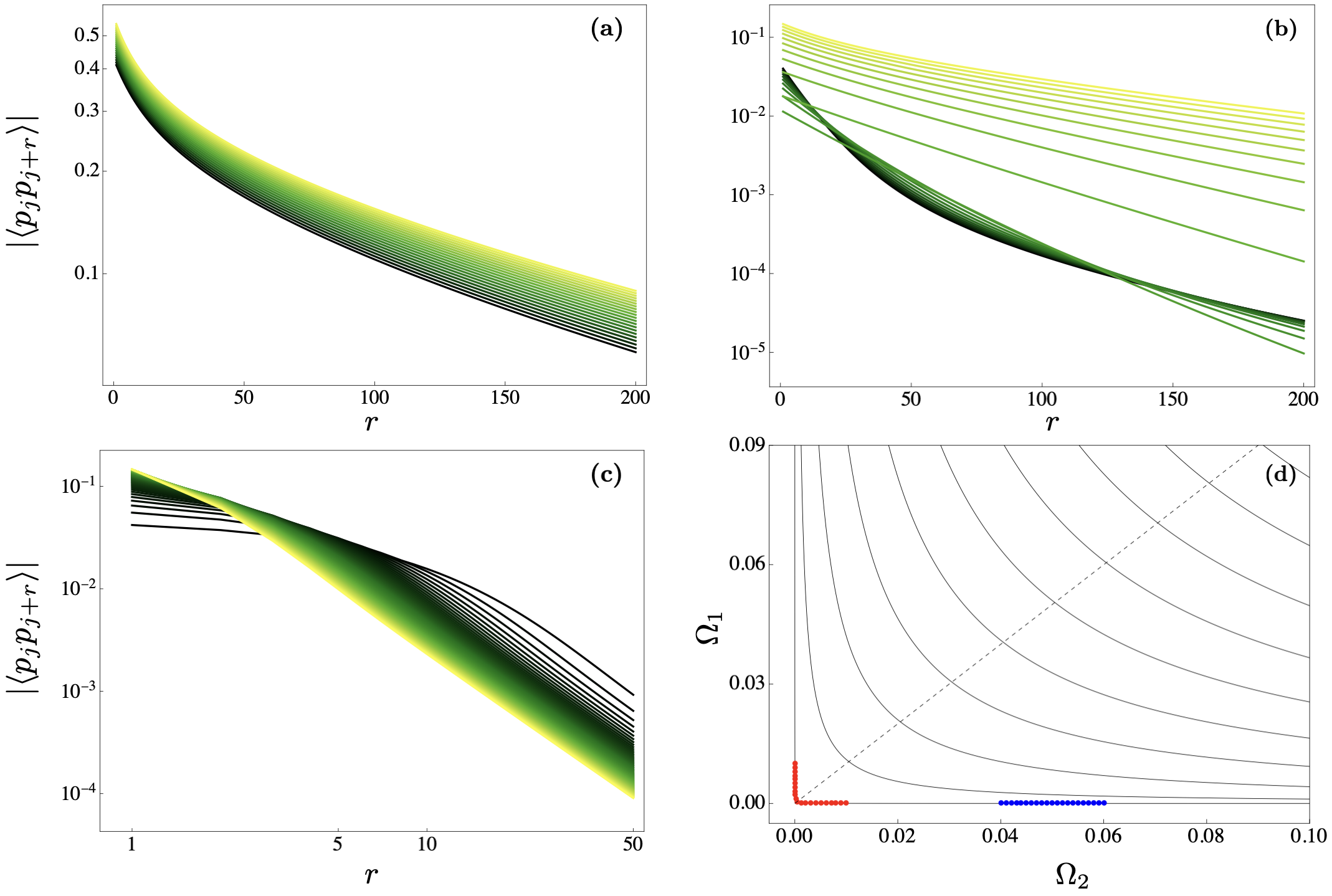}
    \caption{Correlation functions and Krein gap for the double harmonic chain. 
\textbf{(a)} Log plot of $|\langle p_jp_{j+r}\rangle|$ with a fixed $\Delta_{\text{Krein}}=0.001$ and varying choice of parameters $(\Omega_1,\Omega_2)$ (see also panel \textbf{d}): from darkest to lightest colors, the values of $\Omega_2$ ($\Omega_1$) are varied from $10^{-3}$ to $1$, with increments of $10^{-2}$. The exponential decay of correlations is shown for 20 points, that are sampled $\Delta_{\text{Krein}}$ away from the EP at $\Omega_1=0, \Omega_2=0.05$. \textbf{(b)} Same as \textbf{(a)}, but for the decay of correlations $\Delta_{\text{Krein}}$ away from the KC. For the EP, the slope of the decaying exponentials, and hence the correlation length $\xi$, does not vary for small variations in the parameters. For the KC, there is a large variation in $\xi$, depending on the parameters chosen, i.e., the {\em path of approach} to the KC. 
\textbf{(c)} Log-log plot of algebraically decaying $|\langle p_jp_{j+r}\rangle|$ $\big(|\langle x_jx_{j+r}\rangle|\big)$ correlations along the $\Omega_2=0$ ($\Omega_1=0$) EP line. 
\textbf{(d)} Solid black curves denote the level sets of the Krein gap: $\Delta_\text{Krein}$ is constant along each black curve. In blue we denote points near an EP at $(\Omega_1=0, \Omega_2=0.05)$, which we have sampled for correlations in panel \textbf{(a)}. In red we denote points around the KC ($\Omega_1=\Omega_2=0$) we have sampled for correlations in panel \textbf{(b)}. The dashed black line denotes a parameter regime for which the Fock vacuum is a valid QPV of the Hamiltonian. In all panels, $K_1=K_2=1$.} 
 \label{fig:KCapproach}
\end{figure}

So far, we have studied models with EP-type Krein-gap closings. While the next model we present is thermodynamically stable, it displays a rich interplay between EP- and KC- type Krein gap-closings and, as such, extends the existing picture of GS criticality in QBHs. Let us consider the following Hamiltonian:
\begin{align}
\label{eq:multicritHam}
    H=\frac{1}{2}\sum\limits_{j\in\mathbb{Z}}~\Big(2\Omega_1 p_j^2+2\Omega_2 x_j^2+K_1(p_{j+1}-p_j)^2 + K_2(x_{j+1}-x_j)^2\Big), \quad K_1,K_2,\Omega_1,\Omega_2\geq 0.
\end{align}
We note that the above QBH is both harmonic and self-dual under the local canonical transformation $x_i\mapsto p_i$ and $p_i\mapsto -x_i$. It follows that all the spectral properties of the models are invariant under the simultaneous exchange of parameters $\Omega_1\leftrightarrow\Omega_2$ and $K_1\leftrightarrow K_2$. The self-dual surface in parameter space is characterized by the conditions $\Omega_1=\Omega_2$ and $K_1=K_2$. This anticipates the mirroring of $xx$ and $pp$ correlations we will find below. 

From the form of the dynamical matrix (given in Appendix \ref{App: Double}), and Eq.\,\eqref{GammaQPV}, we can infer that, for $K_1,K_2,\Omega_1,\Omega_2>0$, this model is both thermodynamically and dynamically stable. In this parameter regime, it thus lies within the realm of \cite{Cramer2006,Wolf2006}. When $(\Omega_2=0, \Omega_1\neq 0)$ or $(\Omega_1=0, \Omega_2\neq 0)$, however, an EP develops at $k=0$. When the two resulting EP lines meet at $\Omega_1=\Omega_2=0$,  a KC occurs at $k=0$ instead. 
As we will now show, peculiar features arise due to these spectral singularities -- and particularly the KC -- in spite of the underlying thermodynamical stability.
The QPV correlation functions are calculated to be
\begin{align}
\label{eq:MultiCritCorrel}
 \langle p_jp_{j+r}\rangle &\nonumber=\int\limits_{0}^\pi \frac{dk}{\pi} \frac{\Omega_2+K_2(1-\cos (k))}{\sqrt{4(\Omega_1+K_1(1-\cos (k)))(\Omega_2+K_2(1-\cos (k)))}}\cos(kr),\\
    \langle x_jx_{j+r}\rangle &=\int\limits_{0}^\pi  \frac{dk}{\pi} \frac{\Omega_1+K_1(1-\cos (k))}{\sqrt{4(\Omega_1+K_1(1-\cos (k)))(\Omega_2+K_2(1-\cos (k)))}}\cos(kr),
\end{align}
where we identify $\mathcal{E}(k)=\sqrt{4(\Omega_1+K_1(1-\cos (k)))(\Omega_2+K_2(1-\cos (k)))}$ and $\Delta_\text{Krein} = 4\sqrt{\Omega_1\Omega_2}$. Their behavior is illustrated in Fig.\,\ref{fig:KCapproach}. The two EP lines are critical, with long-range correlations that are dual to each other: for $\Omega_1=0, \Omega_2>0$, the QBH has algebraically decaying $xx$ correlations and unbounded $pp$ correlations, whereas for $\Omega_2=0,\Omega_1>0$, the opposite is true. In fact, when $K_1=K_2$, the $xx$ and $pp$ correlations exactly mirror each other along the two EP lines (see Fig.\,\ref{fig:KCapproach}(c)). A transition between these two regimes happens at the KC, where the two EP lines meet and the scaling of the $xx$ and $pp$ correlations is interchanged. 

By examining the behavior in Fig.\,\ref{fig:KCapproach}(c), one may observe that, as $\Omega_1\rightarrow 0$, the correlations are decaying more slowly as a function of $r$. However, their overall amplitude is also decreasing, eventually becoming identically zero at the KC, for a given $r\geq 1$. Let us see precisely how correlations behave at the KC. Although the QBH is harmonic and gapless at this point, directly plugging $\Omega_1=\Omega_2=0$ into Eq.\,\eqref{eq:MultiCritCorrel} shows correlations are manifestly not long-range: 
\begin{align}
\label{eq:multicritKCcorrelat}
     \langle p_jp_{j+r}\rangle =
    \begin{cases}
\frac{1}{2}\sqrt{\frac{K_2}{K_1}}& r=0 \\
0& r\neq 0
\end{cases}, \quad \quad
    \langle x_jx_{j+r}\rangle =
    \begin{cases}
\frac{1}{2}\sqrt{\frac{K_1}{K_2}}& r=0 \\
0& r\neq 0
\end{cases}.
\end{align}

In fact, the above correlations are perfectly local, in apparent contradiction with known results on harmonic lattices. The resolution is that the theorem of \cite{Cramer2006}, which guarantees long-range correlations whenever the energy gap vanishes, applies only when the system is non-trivially coupled \emph{through position alone}, a condition that fails for the KC. Indeed, KC-type Krein gap-closings are eliminated entirely when the QBH takes the form $\mb H = \mb H_{xx} \oplus \mb 1$ (see Appendix \ref{app:direct KC condition}). Accordingly, extensions of the criticality results of \cite{Cramer2006} to QBHs with non-trivial \emph{position and momentum} couplings should be taken with care.

The importance of this distinction becomes clear when we examine $|\langle p_j p_{j+r} \rangle|$ in the vicinity of the KC rather than at the point itself. Reaching a KC requires tuning \emph{two} parameters simultaneously, in contrast to the single-parameter tuning of the earlier examples. Figures \ref{fig:KCapproach}(a) and \ref{fig:KCapproach}(b) compare the decay of $pp$ correlations at fixed Krein gap near an EP and a KC, respectively, as $\Omega_1$ and $\Omega_2$ are varied. Near the EP, the decay is insensitive to the direction of approach; near the KC, it depends sharply on the path taken. This path sensitivity, together with the intersection of two critical lines at the KC and the crossover of the gap closing from linear to quadratic in $k$, identifies 
the KC as \emph{multicritical}.

Such a multicritical nature can be made most evident by studying the \emph{path-dependent critical behavior} upon approach, as in standard settings \cite{Deng,Patra}.  Unlike an ordinary critical point, where a single $(\nu, z)$ pair governs scaling, at the KC the relationship between the Krein gap and the correlation length is path-dependent -- a direct consequence of the KC being a codimension-2 multicritical point at the intersection of two EP lines.
Consider a general one-parameter approach to the KC,  $(\Omega_1, \Omega_2) \equiv (a\, t^{\alpha_1}, b\, t^{\alpha_2})$ with $a, b > 0$, $\alpha_1, \alpha_2 > 0$, and $t \to 0^+$. A direct calculation gives
\begin{equation}
\Delta_{\text{Krein}} \sim 2\sqrt{ab}\: t^{(\alpha_1+\alpha_2)/2}, 
\qquad 
\xi \sim \bigg[2\min\!\left(\frac{a\,t^{\alpha_1}}{K_1},\, 
\frac{b\,t^{\alpha_2}}{K_2}\right)\bigg]^{-1/2}.
\end{equation}
For small $t$, i.e., very close to the Krein gap-closing, the minimum is controlled by the term with the larger exponent, so that $\xi \sim t^{-\max(\alpha_1, \alpha_2)/2}$. -The amplitudes $a, b, K_1, K_2$ enter the prefactor, but not the exponent. Two features emerge: (i) The correlation length diverges with $\nu = 1/2$ as a function of the distance to the nearest EP, with an amplitude that depends on the direction of approach.\footnote{In the fine-tuned case $K_1 = K_2$, along $\Omega_1 = \Omega_2$, the $k$-space correlator becomes constant, correlations are purely on-site, and $\xi = 0$ along the entire diagonal.} (ii) The scaling ratio between the Krein gap and correlation length {\em explicitly depends on the path exponents}:
\begin{equation}
\Delta_{\text{Krein}} \sim \xi^{-z}, \qquad 
z = \frac{\alpha_1 + \alpha_2}{ \max(\alpha_1, \alpha_2)}  \in [1, 2].
\end{equation}
Thus, a behavior that continuously interpolates between two limits emerges: Diagonal approaches ($\alpha_1 = \alpha_2$) yield $z = 2$, whereas paths tangent to an EP line ($\min(\alpha_1, \alpha_2) \to 0$) produce $z \to 1$. In contrast, at the EPs, $\xi$ is fixed by the {\em distance to the critical point alone} ($\xi \sim \Omega_2^{-1/2}$ or $\xi \sim \Omega_1^{-1/2}$), and 
$\Delta_{\rm Krein} \sim \xi^{-1}$ holds, independently of the path, with a constant exponent $z=1$.

The above behavior is illustrated in Fig.\,\ref{fig:KCbump}, where we analyze the path-dependence at the KC by considering a specific family of approaches: $\Omega_1 \equiv \Omega_2^n \rightarrow 0$, with $n \in \mathbb{N}$. For this family, $\alpha_1 = n$ and $\alpha_2 = 1$, giving $z = (n+1)/n$: thus, the dynamic exponent ranges from $z = 2$ at $n = 1$ (diagonal approach) to $z \to 1$ as $n \to \infty$ (asymptotically tangent to the $\Omega_2$-axis EP line in Fig.\,\ref{fig:KCbump}(a)). Figure \ref{fig:KCbump}(c) further shows how the correlation length diverges along such paths when $n > 1$, as expected at a bona fide critical point. These paths also reveal a striking non-monotonic feature: although the magnitude of correlations for $r > 1$ eventually vanishes as $\Omega_2 \to 0$ (see Fig.\,\ref{fig:KCbump}(b)), there is an intermediate regime where they are \emph{amplified}, with the extent of amplification growing with $n$.

The origin of this amplification is the proximity of the path to nearby EPs: for $n > 1$, the $\Omega_1 = \Omega_2^n$  trajectory inches asymptotically toward the $\Omega_1 = 0$ EP line (Fig.\,\ref{fig:KCbump}(a)), and EPs, as dynamical stability phase boundaries, drive the growth of correlation functions. This intuition is made precise by the \emph{Krein phase rigidity} (KPR) \cite{Decon}, an indicator that vanishes at a dynamical stability phase boundary. For our model, $\text{KPR}(k;\Omega_2) = \norm{\vec{\beta}(k;\Omega_2)}^{-2}$, with the $\boldsymbol{\tau}_3$-normalized eigenvector $\vec{\beta}(k;\Omega_2)$ given in Appendix \ref{App: Double}. We find that the KPR vanishes as $\Omega_2 \to 0$ at $k = k_c = 0$ along the paths $\Omega_1 = \Omega_2^n$ with $n > 1$, while remaining finite for $n = 1$, consistent with the picture above.

\begin{figure}[t]
\centering
    \includegraphics[width=\textwidth]{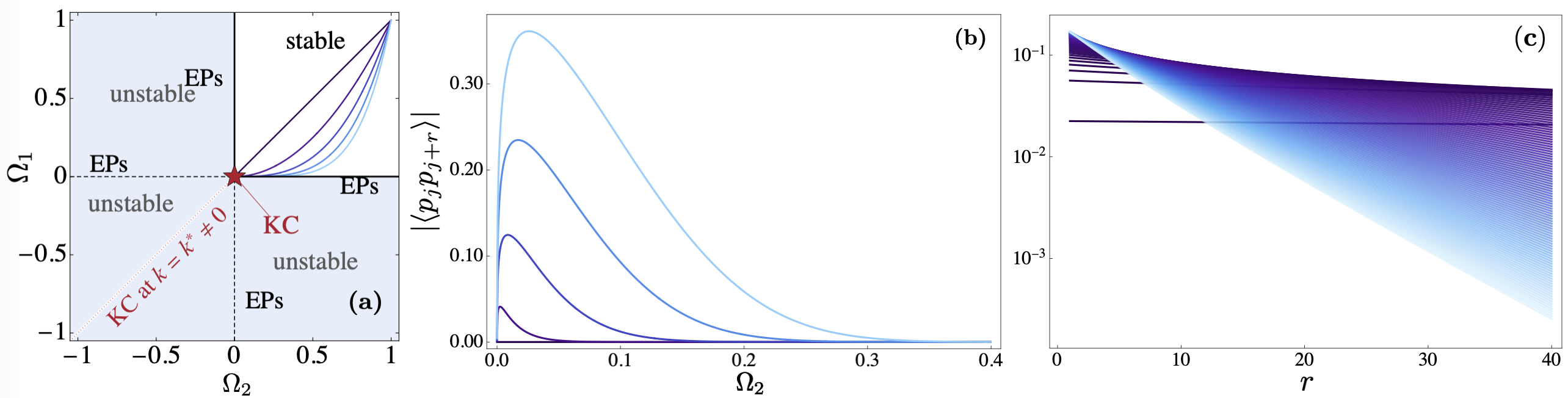}
\vspace*{-4mm} 
\caption{\textbf{(a)} The dynamical stability phase diagram of the double harmonic chain, with blue curves representing paths $\Omega_1=\Omega_2^n$ and $n=1,2,3,4,5$ from darkest to lightest. 
\textbf{(b)} Momentum correlation functions of the double harmonic chain for fixed $r=50$, plotted as a function of $\Omega_2$
along various paths approaching the KC:
$\Omega_1=\Omega_2^n$, with $n=1,2,3,4,5$, from darkest to lightest. In all cases (except for $n=1$, when the QPV is the Fock vacuum), the magnitude of the correlations undergoes an amplification as $\Omega_2\rightarrow0$, until eventually going back down to zero. 
\textbf{(c)} $|\langle p_jp_{j+r}\rangle|$ on a logarithmic scale as a function of $r$, with varying $(\Omega_1,\Omega_2)$. From light and transparent to dark and opaque, $\Omega_1=\Omega_2^2\rightarrow0$. As $\Delta_\text{Krein}$ closes, nearing the KC, the correlation length increases. The overall amplitude of the correlations increases at first, reaching a peak, until eventually decreasing 
to zero. In all cases, $K_1=K_2=1$. }
\label{fig:KCbump}
\end{figure}

\section{Information-theoretic indicators of QPV criticality}
\label{sec:otherindicators}

In this section, we further bolster our expanded notion of criticality by looking at information-theoretic indicators which, as discussed in Sec.\,\ref{sub: info-th}, have been extensively used to characterize quantum critical phenomena. 

\subsection{Entanglement entropy scaling}

For non-critical QBHs, exponentially bounded correlation functions in the GS have been associated with sub-extensive scaling of the EE. Having established that exponentially bounded correlations are more generally supported by the QPV, as long as $\Delta_{\text{Krein}}>0$, it becomes natural to examine the behavior of EE in the QPV. If our extended notion of QPV criticality is to be a legitimate one, a reasonable question to ask is whether the scaling of EE remains unaltered; specifically: (i) Does the EE satisfy an area law for $\Delta_{\text{Krein}}>0$? (ii) Does the EE scale inversely with the Krein gap, analogously to its dependence on the many-body energy gap in the harmonic case? We argue that the answer to both these questions is positive, as summarized in the following:

\smallskip
 
\noindent {\bf Claim: Area law for Krein-gapped QBHs.} {\em If $\Delta_{\text{Krein}}>0$, the EE of a translationally invariant QBH with finite-range couplings obeys an area-law in the QPV and scales inversely with $\Delta_{\text{Krein}}$.}

\smallskip

This statement can be easily proved for translationally invariant QBHs                   {with $d=1$} which, in the quadrature form of Eq.\,\eqref{eq:quadHam}, additionally obey the condition 
\begin{align}
\label{eq:HarmonicWithHopping}
\mb H_{xp}+\mb H_{xp}^T=0.
\end{align}
QBHs with this property can be obtained by adding imaginary hopping terms to a harmonic-lattice model, as in Eq.\,\eqref{eq:GHCwImHopHam}. For translationally invariant harmonic lattices, the upper-bound on the EE (recall Eq.\,\eqref{eq: EEUpperBound}) is proportional to the area of the distinguished sub-region and inversely proportional to the GS energy gap. Strictly speaking, Eq.\,\eqref{eq: EEUpperBound} was obtained for $\mb H_{pp}=\mb 1$. Note, however, that the CM of a translationally invariant harmonic Hamiltonian is $ \mb \Gamma= \left(\mb H_{xx}^{-1}\mb H_{pp}\right)^{1/2}\oplus \left(\mb H_{xx}^{-1}\mb H_{pp}\right)^{-1/2}$ \cite{Wolf2006}. Thus, $\mb H= \mb H_{xx}\oplus\mb H_{pp}$ has the same CM as $\tilde{\mb H}= \tilde{\mb H}_{xx}\oplus \mb 1$, with $ \tilde{\mb H}_{xx}\equiv \mb H_{xx}\mb H_{pp}^{-1}$, when $\mb H_{xx}$ and $\mb H_{pp}$ are {\em both} invertible -- as it is the case when $\Delta_{\text{Krein}}>0$). Under the same assumptions, the elements of $\tilde{\mb H}_{xx}$ and $\tilde{\mb H}_{xx}^{-1/2}$ are exponentially decaying \cite{Wolf2006}. Therefore,  Eq.\,\eqref{eq: EEUpperBound} can still be used, under an appropriate change of variables. Going further, by the results in Sec.\,\ref{Sec: CorrelationsQPV}, the strength of the imaginary hopping can be tuned to induce loss of thermodynamic stability, without changing the relevant CM: the underlying state is then reinterpreted as the QPV rather than the GS. Since entanglement measures in the GS, as well as the QPV, are computed through the CM, such a transition leaves the EE unaffected as well. In particular, for QBHs obeying Eq.\,\eqref{eq:HarmonicWithHopping}, the bound in Eq.\,\eqref{eq: EEUpperBound} still holds, with $\Delta^\text{gs}=\lambda_{\text{min}}(\mb H_{xx})^{1/2}$ being replaced by $\Delta_\text{Krein}/2$ \footnote{ While $ \lambda_{\text{min}}(\mb H_{xx})^{1/2}=\lambda_{\text{min}}(\boldsymbol{ \mathcal{E}})$ no longer has the meaning of the many-body energy gap when the GS is lost, it is still non-negative and corresponds to $\Delta_\text{Krein}/2$ in the thermodynamically unstable phase, as long as dynamical stability is retained.}.

\begin{figure}[t]
    \centering
    \includegraphics[width=\textwidth]{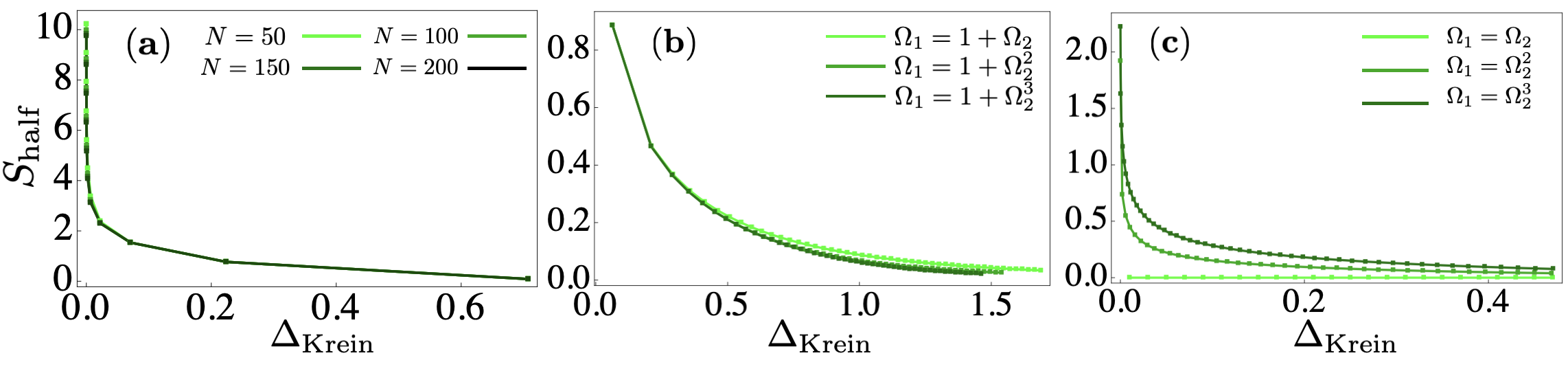}
    \vspace*{-3mm}
    \caption{EE of a symmetrically bisected chain as a function of the Krein gap. (\textbf{a}) Interpolation model for different system size $N$: the EE reaches a limiting asymptotic value inversely proportional to $\Delta_{\text{Krein}}$. No qualitative changes
are seen between regions where the model is thermodynamically stable (here, $\Delta_{\text{Krein}}  \gtrsim 0.57$), or unstable ($\Delta_{\text{Krein}}  \lesssim  0.57$). Parameters $J=2,\Delta=\Omega=1$.
(\textbf{b}) and (\textbf{c}): Double harmonic chain for fixed system size, $N=30$. In (\textbf{b}), an EP at $(\Omega_2=0,\Omega_1=1)$ is approached using different paths in parameter space. As the Krein gap-closing is approached, different paths converge to the same EE value, which scales inversely with $\Delta_{\text{Krein}}\rightarrow 0$. In (\textbf{c}), the KC $(\Omega_1=0,\Omega_2=0)$ 
is approached using different paths instead.
In contrast with the EP case, the EE does not converge to the same value:
 while there are some paths for which the EE diverges as $\Delta_{\text{Krein}}\rightarrow 0$, there is also a path for which it vanishes.
In both panels, $K_1=K_2=1$.}
    \label{fig:MultiCritEE}
\end{figure}

                  {At first glance, the interpolation model of Sec.\,\ref{sub:im}} does not have the simple structure associated with Eq.\,\eqref{eq:HarmonicWithHopping}, so we numerically examine the validity of our above area-law claim by assessing the scaling of EE as a function of both system-size and the magnitude of the Krein gap. As shown in Fig.\ref{fig:MultiCritEE}(a), we find that area law is satisfied and the magnitude of the EE diverges as the Krein gap closes. Notably, there is {\em no} qualitative difference between thermodynamically stable and unstable phases.                   {In fact, it is possible to prove this analytically by performing a change of basis of the form $a_j\mapsto e^{i\pi/4}a_j, a_j^{\dag}\mapsto e^{-i\pi/4}a_j^{\dag}$. In the bosonic basis, this unitary transformation leaves the hopping and onsite terms unaffected, while swapping the imaginary pairing for real pairing. Crucially, it leaves the spectrum as well as the EE of the model unaffected. Thus, the interpolation model takes the form of a harmonic lattice with imaginary hopping which, by the argument developed above, provably satisfies the area law.}

It is also worth commenting on the distinct signatures that EPs and KCs engender on the EE behavior. For the double harmonic chain, we showed in Sec.\,\ref{sub:dhc} that correlation functions behave in a qualitatively different way as EPs or KCs are approached. Figs.\,\ref{fig:MultiCritEE}(b) and \ref{fig:MultiCritEE}(c) show that this is also reflected by the EE behavior: while it is found to diverge in the same fashion regardless of the path of approach to an EP, this is not the case for a KC. Each path asymptotically approaches a different value of the EE, as $\Delta_\text{Krein}$ closes. Furthermore, just as in case of correlation functions, the EE diverges along two of the chosen paths to the KC, $\Omega_1=\Omega_2^2$ and $\Omega_1=\Omega_2^3$, whereas it remains identically zero for the path with $\Omega_1=\Omega_2$. Once again, the idea of path-dependence near a KC emerges.

\subsection{Quantum metric tensor and fidelity behavior} 

As a final observation in support of our generalized notion of bosonic criticality, we now show how the pseudo-Hermitian QMT given in Eq.\,\eqref{eq: QMT} is sensitive to the closing of the Krein gap, further corroborating its role as the relevant many-body gap for QBHs. 

Let us first consider the interpolation model in Eq.\,\eqref{eq: InterpHam}. Using the explicit expressions of the eigenvectors given in Appendix \ref{App: InterpCovMat}, and focusing on the Hamiltonian parameter $\alpha= s\Delta/[\Omega(1-s)]$, we obtain the following expression for the corresponding momentum-space QMT:
\begin{align*}
    g^{LR}_{\alpha\alpha}(k)=-\frac{\cos ^2(k)}{4
   \left(\alpha ^2 \cos
   ^2(k)-1\right)^2}.
\end{align*}
This quantity diverges only when $\alpha=1$ and $k=k_c=0,\pm\pi$, exactly coinciding with the critical points we identified from the analysis of the correlation functions in Eq.\,\eqref{eq:interpcorrel}. 

A similar analysis may be repeated for the double-harmonic chain model in Eq.\,\eqref{eq:multicritHam}, using the eigenvectors from Appendix \ref{App: Double}. At the gap-closing momentum $k = k_c = 0$, we can evaluate three independent components of the momentum-space QMT: 
\begin{equation}
g^{LR}_{\Omega_1\Omega_1}(0) =- \frac{1}{16\Omega_1^2}, \qquad 
g^{LR}_{\Omega_2\Omega_2}(0) = -\frac{1}{16\Omega_2^2}, \qquad 
g^{LR}_{\Omega_1\Omega_2}(0) = \frac{1}{16\,\Omega_1\Omega_2}.
\end{equation}
Each diagonal component diverges on exactly one EP line: 
$g^{LR}_{\Omega_1\Omega_1}$ at $\Omega_1 = 0$, detecting the $\Omega_2$-axis EP line, and $g^{LR}_{\Omega_2\Omega_2}$ at $\Omega_2 = 0$, detecting the $\Omega_1$-axis EP line. The off-diagonal component diverges on both EP lines and at the KC (see Fig.\,\ref{fig:QMT}(left)). Thus, at the KC, the QMT becomes singular, despite the real-space correlations being trivially short-ranged. The multicritical nature of this point is reflected in correlations only indirectly, through their path-dependent decay. In contrast, the QMT captures the singular character of this point directly at the KC itself, without requiring a study of how it is approached. In this sense, the QMT is a sharper diagnostic than correlation functions, its divergence providing further evidence that the KC is a genuine multicritical point.

\begin{figure}[t]
    \centering \hspace*{20mm}
    \includegraphics[width=0.85\textwidth]{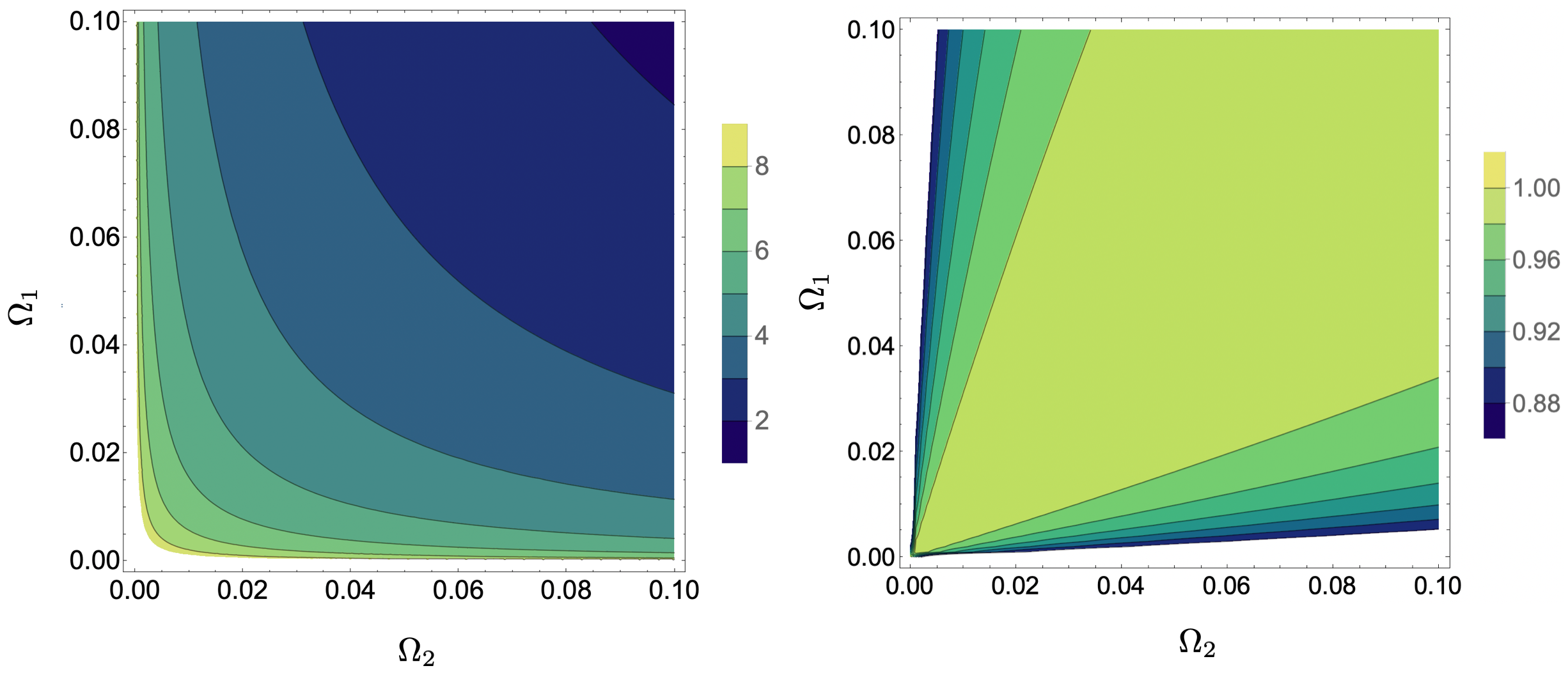}
    \vspace*{-3mm}
    \caption{\textbf{Left:} The $k$-space QMT for the double harmonic chain, $\log\left(g^{LR}_{\Omega_1\Omega_2}(k=0)\right)$, displays a divergent behavior near the critical values of the parameters $\Omega_1,\Omega_2$ at $k=k_c=0$. Around $\Omega_1=\Omega_2=0$, the QMT reflects the characteristic multicritical nature of the KC. \textbf{Right:} A density plot of the distance between the QPV and the Fock vacuum, as quantified by the quantum fidelity for zero-mean pure Gaussian states. In both panels, $K_1=K_2=1$. } 
    \label{fig:QMT}
\end{figure}

A complementary, global picture of multicriticality at the KC emerges from the quantum fidelity given in Eq.\,\eqref{eq:Fid}, which in this case further simplifies to $$\mathscr{F}( \rho^\text{qpv},|0\rangle\langle 0|) = \left(\det\left(\frac{\mb \Gamma + \mb I}{2}\right)\right)^{-1/4},$$ 
and whose behavior is illustrated in Fig.\,\ref{fig:QMT}(right). Comparing QPVs in the vicinity of the KC against the Fock vacuum, we find that states infinitesimally close to the KC can nonetheless be drastically distinct from one another. This provides a structural origin for the path-dependent decay of correlations discussed in Sec.\,\ref{sub:dhc}. The KC is thus revealed as a limit point of sharply differing QPVs. This reconciles the apparent tension between the triviality of on-site correlations at the KC and the richness of path-dependent correlations upon approach.

\section{Conclusions and outlook}
\label{sec:conclusion}


We have introduced a new framework for investigating the critical behavior of                   {many-body bosonic systems effectively described by quadratic (Gaussian) Hamiltonians}
based on dynamical rather than thermodynamical stability considerations. We recovered the known critical behavior of thermodynamically stable, translationally invariant QBHs as corollaries to our results, and found that the fundamental dichotomy for QBHs is the {\em difference between dynamical stability and instability} -- irrespective of thermodynamic stability, hence the existence of a many-body GS. 

From a technical standpoint, two ingredients are key to our generalization: i) the notion of a Krein gap 
-- loosely speaking, the minimal energy separation between bosonic creation and annihilation operators; and ii) the identification of the QPV as the appropriate state hosting critical behavior, even in the absence of a GS. Notably, the Krein gap coincides with the usual spectral gap (up to a factor of 2) for the simplest case of harmonic QBHs, whereas for general thermodynamically stable QBHs, it recovers -- and provides a physical interpretation to -- the symmetrized gap introduced in Ref.\,\cite{Wolf2006} on the basis of mathematical convenience. By leveraging Gaussianity and translational invariance, we have computed the two-point correlation functions for dynamically stable QBHs in the QPV, for arbitrary spatial dimensions and number of internal DOFs. The Krein gap must close for long-range correlations to arise: for a finite Krein gap, correlations are always bounded by a decaying exponential. A corollary to this is that the QPV may become quantum-critical 
{\em only} at dynamical stability phase boundaries, as the Krein gap also quantifies the proximity of a QBH to a nearest dynamical stability phase boundary. More precisely, we found that when the Krein gap closes at a finite number of points in momentum space, and the gap-closings are of EP-type, long-range correlations generically emerge. In contrast, hallmarks of multicritical behavior are observed for KC-type Krein-gap closings, with the scaling of the correlations being highly dependent upon the path of approach to the KC.

We have bolstered our proposed generalized criticality notion by examining the extent to which it is also reflected in the behavior of information-theoretic indicators that are often used in parallel with correlation functions for standard  critical phenomena. In particular, we showed that, for the general class of dynamically stable QBHs, the bipartite EE in the QPV obeys the same scaling with system size and the relevant (Krein) gap as it does for harmonic lattices in the GS. Similarly, we have found that the QMT diverges precisely at the critical points identified through correlation functions, and in fact provides a sharper diagnostic tool for the multicritical nature of a KC point. 

From a conceptual standpoint, the fundamental categorization of QBHs into dynamically stable and unstable ones is both natural and physically compelling; in a way, it seems fair to say that it brings the equilibrium statistical mechanics of dynamically stable QBHs on par with that of quadratic fermionic Hamiltonians. At the most basic level, these two classes of Hamiltonians are as similar as they can be, given the difference in quantum statistics, for the simple but profound reason that both of them are fully solvable in terms of canonical quasiparticles. The crucial difference is that, for fermions, there is always a QPV GS, whereas for dynamically stable bosons, that is not the case. Interestingly, the problem of simulating an arbitrary quadratic fermionic Hamiltonian on a quantum computer is known to always belong to the BQP-complete complexity class \cite{KrausBQP}. Taken together, our results instead provide strong support to the idea that dynamical stability may underlie the exponential advantage in quantum simulation recently reported for a broad class of QBHs \cite{SchuchBQP}, with the onset of dynamical instability leading to a computational complexity transition into PostBQP. Likewise, lifting stability requirements has been advocated as a necessary step for a topological, symmetry-based classification of QBHs to be achievable \cite{EmilioQBH,MariamNext}. 

Our work suggests several venues for further investigation. First, consistent with the fact that the standard notion of criticality is defined for systems in the thermodynamic limit, throughout our analysis we have assumed to work in the idealized setting of a boundary-less system, subject to bi-infinite BCs. In practice, finite-size scaling methods are used to infer critical properties (notably, critical exponents) from data pertaining to realistic systems, which are necessarily finite and modeled through open BCs. Care is needed for bosonic systems, however, due to the possibility that dynamical stability is lost upon imposing boundaries, for so-called Type II dynamically metastable QBHs \cite{DynamicalMetastability} (which e.g., the interpolation model exemplifies). Examining the behavior of correlation functions under open BCs may thus be a well-worth next step. Since, even in this class, dynamical stability is always realized for sufficiently large system size (say, $N >N_c$), we expect our current conclusions to hold in such a regime. Distinctive behavior may nonetheless emerge in a ``transient system-size regime'' $N<N_c$, akin to a similar phenomenon reported for quadratic bosonic Lindbladians \cite{DM2} and calling for additional investigation.

Second, for pure states of a standard (Hermitian) system, it is known that the QMT can be directly related to the quantum Fisher information matrix, a core tool in the theory of quantum estimation theory \cite{MultiParam}. It would be interesting to determine whether such a relationship admits a pseudo-Hermitian generalization, applicable to QBHs. If so, the singular behavior of the pseudo-Hermitian QMT at the generalized critical phase boundaries would imply a corresponding enhancement of metrological sensitivity, pointing to connections between the fields of EP-enhanced and criticality-enhanced quantum sensing \cite{Wsensing,Montenegro}, including uncovering potential sensing implications of KCs.

Finally, extending the analysis beyond static QBHs as considered here remains a broad important problem. For instance, analyzing a time-dependent quench to an EP or KC may reveal whether universal dynamics emerges and is governed by the critical exponents that the static analysis predicts, as in the Kibble-Zurek scenario for standard critical points \cite{Deng} and effective non-Hermitian Hamiltonians \cite{Dora}. Likewise, elucidating the role and influence (if any) of EPs in steady-state criticality properties of quadratic bosonic Lindbladians beyond one dimension \cite{BarthelCriticality} is especially compelling in the light of modified scaling behavior recently uncovered for an extended Dicke model \cite{LeeEnhanced}. It is our hope to report on some of these issues in future work. 


\section*{Acknowledgments}

The authors are indebted to Yikang Zhang for many insightful discussions on quantum criticality in bosonic systems and for a critical reading of the manuscript. Work at Dartmouth was supported by the National Science Foundation through Grants No.\,PHY-2013974 and PHY-2412555 (to LV).

\appendix

\section*{Appendix}

\renewcommand{\theequation}{A.\arabic{equation}}
\setcounter{equation}{0}
\renewcommand{\thesection}{\Alph{section}}
\setcounter{section}{0}

\section{Detail on dynamical and thermodynamical stability properties for single-band QBHs}

\subsection{Stability properties from the Bloch dynamical matrix}
\label{App: BlochMat}

Here, we establish basic properties of the momentum-space dynamical matrix with $d=1$.  The structure of $\mb g(\bm k)$ in Eq.\,\eqref{eq: Hphicomm} can be leveraged to show that:
\begin{align}
\label{dag}
   \mb  g^{\dag}(\bm k)&=
   \boldsymbol \tau_3\mb g(\bm k)\boldsymbol \tau_3,\\
   \label{star}
   \mb  g^*(\bm k)&=
   -\boldsymbol \tau_1\mb g(-\bm k)\boldsymbol \tau_1.
\end{align}
For a system with finite-range couplings as we consider, we can write 
\begin{align*}
    \mb g(\bm k)&=d_0(\bm k)\boldsymbol \sigma_0+id_1(\bm k)\boldsymbol \sigma_1+id_2(\bm k)\boldsymbol \sigma_2+d_3(\bm k)\boldsymbol \sigma_3\equiv \mathcal{G}_{\bm k}(d_0,d_1,d_2,d_2).
\end{align*}
From Eq.\,(\ref{dag}) and the properties of Pauli matrices we deduce:
\begin{align*}
    \mathcal{G}_{\bm k}(d_0^*,-d_1^*,-d_2^*,d_3^*)= \mathcal{G}_{\bm k}(d_0,-d_1,-d_2,d_3),
\end{align*}
i.e. all $d_i(\bm k)$ are real. From Eq.\,(\ref{star}), we obtain:
\begin{align*}
    \mathcal{G}_{\bm k}(d_0,-d_1,d_2,d_3)= \mathcal{G}_{-\bm k}(-d_0,-d_1,d_2,d_3), 
\end{align*}
which implies that $d_0$ is an odd function of $\bm k$, while the rest are even. All in all, Eq.\,\eqref{eq: Hphicomm} allows the Bloch dynamical matrix to be cast in the form
\begin{align}
\label{eq: Gkgeneral}
    \mb g(\bm k)=\begin{pmatrix}
       d_0(\bm k)+d_3(\bm k)&i d_1(\bm k)+d_2(\bm k)\\
    i d_1(\bm k)-d_2(\bm k)   &d_0(\bm k)-d_3(\bm k)
    \end{pmatrix}, \qquad d_j(\bm k)=\sum\limits_{|\bm r |\leq R} d_{j\bm r}\, e^{i\bm k\cdot \bm r}, \; j=0,1,2,3,
    \end{align}
with $d_j(\bm k)$ odd under $\bm k$ for $j=0$ and even otherwise. Physically, the real-space coefficients $d_{3\bm r}$ and $d_{0\bm r}$ encode real and imaginary hopping terms, while $d_{2\bm r}, d_{1\bm r}$ encode real and imaginary pairing terms, respectively. For illustration, we show how the form of Eq.\,\eqref{eq: Gkgeneral} may be obtained from real-space coupling matrices in the simplest case of $D=1$. In this case, the Hamiltonian in Eq.\,\eqref{eq: QBH} and the corresponding dynamical matrix may be rewritten as 
\begin{align}
\label{eq: HamWithGrs}
  H=\frac{1}{2} \sum_{j\in {\mathbb Z}} \sum_{r =-R}^R \phi_j^{\dag} {\mb h}_r \phi_{j+r} ,\qquad  
\mb G =\mathds{1} \otimes \mb g_0+\sum_{r=1}^{R}\left(\mb V^r\otimes \mb g_r+\mb V^{\dag r}\otimes \mb g_{-r}\right),
\end{align}
where $\mb g_r \equiv {\boldsymbol \tau}_3 \mb h_r$ and the operator $\mb V\equiv \sum\limits_{r=-\inf}^{\inf}\vec{e}_j \vec{e}_{j+1}^{\,\dag}$, defined in terms of canonical basis vectors in ${\mathbb Z}$, appropriately accounts for the bi-infinite BCs we employ \cite{DynamicalMetastability}.  More specifically, the matrices $\mb K,\mb \Delta$ appearing in Eq.\,\eqref{eq: QBH} can be parametrized as 
\begin{align}
 \text{Re}(\mb K)&=d_{30}\mb 1+\sum\limits_{r=1}^R\frac{d_{3r}}{2}(\mb V^r+\mb V^{\dag r}) ,\qquad 
    \text{Im}(\mb K)=\sum\limits_{r=1}^R\frac{d_{0r}}{2}(\mb V^{\dag r}-\mb V^{ r}), 
    \label{eq: KDeltaCirculant}\\
     \text{Re}(\mb \Delta)&=d_{20}\mb 1+\sum\limits_{r=1}^R\frac{d_{2r}}{2}(\mb V^r+\mb V^{\dag r}),\qquad 
       \text{Im}(\mb \Delta)=d_{10}\mb 1+\sum\limits_{r=1}^R\frac{d_{1r}}{2}(\mb V^r+\mb V^{\dag r}),
       \label{eq: KDeltaCirculant2}
\end{align}
which, combined with Eq.\,\eqref{eq: HamWithGrs}, yields
\begin{align*}
  \mb g_0 &= \begin{pmatrix}
    d_{30}&d_{20}+id_{10}\\
    -d_{20}+id_{10}&-d_{30}
    \end{pmatrix}, \;
    \mb g_{r}=\frac{1}{2}
     \begin{pmatrix}d_{3r}-i d_{0r}&d_{2r}+id_{1r}\\
     -d_{2r}+id_{1r}&-id_{0r}-d_{3r}\end{pmatrix}, \;
     \mb g_{-r}=\begin{pmatrix}d_{3r}+i d_{0r}&d_{2r}+id_{1r}\\
     -d_{2r}+id_{1r}&+id_{0r}-d_{3r}\end{pmatrix}.
\end{align*}
In other words, the Bloch dynamical matrix is consistently recovered as the Fourier transform of the real-space coupling matrices: $\mb g(\bm k)=\sum\limits _{r=-R}^R \mb g_{r} e^{ik r }$.

Back to the general $D$-dimensional setting, the eigenvalue spectrum of $\mb g(\bm k)$ is given by
$$\sigma(\mb g(\bm k))=d_0(\bm k)\pm\sqrt{d_3(\bm k)^2-d_2(\bm k)^2-d_1(\bm k)^2}.$$
Denoting the two eigenvalues as $\omega(\bm k)$ and $-\omega(-\bm k)$, we may then identify 
\begin{align*}
\frac{1}{2}\big(\omega(\bm k)+\omega(-\bm k)\big) \equiv 
    \mathcal{E}(\bm k)
    = \sqrt{d_3(\bm k)^2-d_2(\bm k)^2-d_1(\bm k)^2}, 
\end{align*}
whereby the dynamical stability condition given in Eq.\,\eqref{dynStab} follows. The thermodynamic stability condition can be obtained from the eigenvalues of $\tau_3\mb g(\bm k)$, namely, 
\begin{align*}
 \sigma\left(\boldsymbol \tau_3\mb g(\bm k)\right)=d_3(\bm k)\pm\sqrt{d_0(\bm k)^2+d_1(\bm k)^2+d_2(\bm k)^2}   , 
\end{align*}
which leads to the conditions stated in the main text, Eq.\,\eqref{thStab}.

\subsection{Stability boundaries and conditions for a Krein collision}
\label{app:direct KC condition}

Here, we prove that, for a translationally invariant QBH with $d=1$, a KC at some point $\bm k'$ requires that $d_3(\bm k')=d_2(\bm k')=d_1(\bm k')=0$. From the form of the eigenvectors of $\mb g(\bm k)$ in Eq.\,\eqref{eq: eigenvecs}, we can say that a gap closing is a KC and not an EP if:
\begin{align*}
    d_3(\bm k')^2-d_2(\bm k')^2-d_1(\bm k')^2=0,\quad 
    d_3(\bm k')=p \, (-d_2(\bm k')+i d_1(\bm k')), \quad -d_2(\bm k')-id_1(\bm k')=\Bar{p}\, d_3(\bm k'),
\end{align*}
such that $p\neq \Bar{p}$, so that the eigenvectors are not proportional, and hence there is no EP. But then, assuming that $\Bar{p}\ne 0$,
the following equations also hold:
\begin{align*}
  d_3(\bm k')^2-d_2(\bm k')^2-d_1(\bm k')^2=0,\quad 
    d_3(\bm k')=p \,(-d_2(\bm k')+i d_1(\bm k')), \quad \frac{1}{\Bar{p}}(-d_2(\bm k')-id_1(\bm k'))=d_3(\bm k').
\end{align*}
Taken together, these equations imply that  
\begin{align*}
     d_3(\bm k')^2-d_2(\bm k')^2-d_1(\bm k')^2=0,\quad d_3(\bm k')^2=\frac{p}{\Bar{p}} \, (d_2(\bm k')^2+d_1(\bm k')^2) .
\end{align*}
Thus, 
we have a contradiction, since we assumed $\frac{p}{\Bar{p}}\neq 1$, unless the three terms $d_j(\bm k')$ separately vanish. If 
$\Bar{p}=0$, this in turn implies $d_2(\bm k')-id_1(\bm k')=0$ and hence $d_2(\bm k')=d_1(\bm k')=0$, since both quantities are real. Furthermore, $d_3(\bm k')=0$, as the gap is zero.

We can further show that, if $H$ has the direct-sum structure of the form $\mb H=\mb H_{xx}\oplus\mb 1$, the KC-type gap closing is impossible. Assume that $\text{Re}(\mb K-\mb \Delta)=\mb H_{pp}\sim \mb 1$, and that there is a direct KC at some $\bm k=\bm k'$. From the first condition, $d_3(\bm k)-d_2(\bm k)=c$, with $c$ some non-zero, $\bm k$-independent constant, and $d_3(\bm k)+d_2(\bm k)=b(\bm k)$. From here, we obtain $d_2(\bm k)=\tfrac{1}{2} (b(\bm k)-c)$ and $d_3(\bm k)=\tfrac{1}{2}(b(\bm k)+c)$. The condition for a KC at $\bm k'$ imposes $d_3(\bm k')=d_2(\bm k')=d_1(\bm k')=0$. This implies $b(\bm k')=\pm c$ and therefore $c=0$. This is a contradiction, which implies there cannot be a direct KC, as stated.  

                  {\subsection{Tuning thermodynamic stability for $d>1$}

The Bloch dynamical matrix takes the following general form in the translationally invariant case with finite range and $D,d>1$, in terms of the hopping and pairing matrices:
\begin{align}
\label{eq: Hphicomm}
    [H,\phi(\bm{k})] \equiv -\mathbf{g}(\bm{k})\phi(\bm{k}),\quad \mathbf{g}(\bm{k}) = \begin{pmatrix}
        \mathbf{K}(\bm{k}) & \bm{\Delta}(\bm{k})
        \\
        -\bm{\Delta}^*(-\bm{k}) & - \mathbf{K}^*(-\bm{k}) 
    \end{pmatrix},
\end{align}
with $\mb K^{\dag}(\bm k)=\mb K^{\dag}(\bm k)$, $\mb K (\bm k)=\sum_{\bm r}e^{i\bm k\cdot \bm r}\mb K_{\bm r}$, and $\mb K_{\bm r}^{\dag}=\mb K_{-\bm r}$. We can split the real-space hopping matrices into real and imaginary parts,
$\mathbf{K}_{\bm{r}} = \mb X_{\bm{r}} + i\mb Y_{\bm{r}}$ where $\mb X_{\bm{r}},
\mb Y_{\bm{r}}$ are real. Hermiticity of $H$, i.e., the condition $\mathbf{K}_{\bm{r}}^\dagger =
\mathbf{K}_{-\bm{r}}$, is then equivalent to the conditions
\begin{equation}
\mb X_{-\bm{r}} = \mb X_{\bm{r}}^{T}, \qquad
\mb Y_{-\bm{r}} = -\mb Y_{\bm{r}}^{T}.
\label{eq:XYherm}
\end{equation}
Introducing the Fourier transforms
\begin{equation}
\mathbf{X}(\bm{k}) \equiv \sum_{\bm{r}} e^{i\bm{k}\cdot\bm{r}} \mb X_{\bm{r}},
\qquad
\mathbf{A}(\bm{k}) \equiv i\sum_{\bm{r}} e^{i\bm{k}\cdot\bm{r}} \mb Y_{\bm{r}},
\end{equation}
so that $\mathbf{K}(\bm{k}) = \mathbf{X}(\bm{k}) + \mathbf{A}(\bm{k})$, a
direct computation using Eq.~\eqref{eq:XYherm} shows that $\mathbf{X}^\dagger(\bm k) = \mathbf{X}(\bm k)$ and $\mathbf{A}^\dagger(\bm k) =
\mathbf{A}(\bm k)$, while 
\begin{equation}
\mathbf{X}^{*}(-\bm{k}) = +\,\mathbf{X}(\bm{k}), \qquad
\mathbf{A}^{*}(-\bm{k}) = -\,\mathbf{A}(\bm{k}).
\label{eq:parity}
\end{equation}
This allows us to write
\begin{equation}
\mathbf{g}(\bm{k}) =
\underbrace{\begin{pmatrix}
\mathbf{X}(\bm{k}) & \bm{\Delta}(\bm{k}) \\
-\bm{\Delta}^{*}(-\bm{k}) & -\mathbf{X}(\bm{k})
\end{pmatrix}}_{\textstyle \equiv\, \mathbf{g}_X(\bm{k})}
+
\mb 1_2 \otimes \mathbf{A}(\bm{k}).
\label{eq:gsplit}
\end{equation}

Finally, we can decompose $\mathbf{A}$ into its trace and traceless parts,
\begin{equation}
\mathbf{A}(\bm{k}) = \mu(\bm{k})\,\mb 1_d + \tilde{\mathbf{A}}(\bm{k}),
\qquad
\mu(\bm{k}) \equiv \tfrac{1}{d}\,\mathrm{tr}\,\mathbf{A}(\bm{k}),
\quad \mathrm{tr}\,\tilde{\mathbf{A}} = 0.
\end{equation}
Hermiticity imposes $\mu\in \mathbb{R}$, while Eq.~\eqref{eq:parity} imposes,
$\mu(-\bm{k}) = -\mu(\bm{k})$. For $d=1$ the traceless part is absent and
$\mu(\bm{k}) = d_0(\bm{k})$, recovering the parametrization of
Eq.\,\eqref{eq: Gkgeneral}.

\begin{proposition}
Let $H$ be a dynamically and thermodynamically stable, translationally invariant QBH with finite-range coupling $R$, arbitrary internal dimension $d$ and spatial dimension $D$. Let $\mu_0(\bm{k})$ be a real, odd, trigonometric polynomial of
degree at most $R$. Consider the one-parameter family $H_\gamma$ obtained by
the deformation
\begin{equation}
\mathbf{g}(\bm{k}) \mapsto
\mathbf{g}_\gamma(\bm{k}) = \mathbf{g}(\bm{k})
+ \gamma\,\mu_0(\bm{k})\,\mb 1_{2d},
\qquad \gamma\in\mathbb{R}.
\label{eq:scalardef}
\end{equation}
Then, for all $\gamma$, the following holds: {\rm (i)} $H_\gamma$ is a legitimate QBH; {\rm (ii)} $\Delta_{\rm Krein}(H_\gamma) =
\Delta_{\rm Krein}(H)$ and
$C^{\rm qpv}_\gamma(\bm{k}) = C^{\rm qpv}(\bm{k})$;  {\rm (iii)} $H_\gamma$ is thermodynamically unstable for all $\gamma >
\gamma_c$, where
\begin{equation}
\gamma_c \;\equiv \; \min_{\substack{n,\,\bm{k} \\ \mu_0(\bm{k})<0}}
\frac{\omega_n(\bm{k})}{|\mu_0(\bm{k})|}.
\label{eq:gammac}
\end{equation}
\end{proposition}

\begin{proof}
(i) Pseudo-Hermiticity is immediate, $\bm{\tau}_3(\mu_0\mb 1)\bm{\tau}_3
= \mu_0\mb 1$, since $\mu_0$ is real. Charge conjugation requires
$\mu_0^{*}(\bm{k})\mb 1 = -\bm{\tau}_1\mu_0(-\bm{k})\mb 1
\bm{\tau}_1 = -\mu_0(-\bm{k})\mb 1$, which holds since
$\mu_0$ is real and odd. By Eq.~\eqref{eq:gsplit}, the deformation corresponds
to $\mathbf{K}_{\bm{r}} \mapsto \mathbf{K}_{\bm{r}} + i\gamma\nu_{\bm{r}}
\mb 1_d$, with $\nu_{\bm{r}}$ real and $\nu_{-\bm{r}} = -\nu_{\bm{r}}$, that is, to adding a uniform imaginary hopping acting identically on every internal DOF. Explicitly,
$$
\delta H = \frac{i\gamma}{2}\sum_{\bm{j},\ell}\sum_{\bm{r}>0}
\nu_{\bm{r}}\Big(a^{\dagger}_{\bm{j},\ell} a_{\bm{j}+\bm{r},\ell}
- a^{\dagger}_{\bm{j}+\bm{r},\ell} a_{\bm{j},\ell}\Big) + {\rm H.c.}. $$
The finite-range condition is preserved by the assumptions on $\mu_0$.

(ii) Because $\gamma\mu_0(\bm{k})\mb 1_{2d}$ is a multiple of
the identity, $\mathbf{g}_\gamma(\bm{k})$ and $\mathbf{g}(\bm{k})$ have
identical eigenvectors at every $\bm{k}$. Only
the eigenvalues shift, $\omega(\bm k) \mapsto \omega(\bm k) + \gamma\mu_0(\bm{k})$. Since the eigenvalues shift uniformly at each $\bm k$, the Krein gap is unmodified. Furthermore, since the $\bm k$-space CM is constructed from the modal matrix of $\mb g(\bm k)$ and the eigenvectors of the dynamical matrix are unchanged, so is the CM. 

(iii) Finally, a dynamically stable QBH is bounded from below (above) if and only if 
$\omega_n(\bm{k}) \geq 0$ ($\leq 0$) for all $n,\bm{k}$. Without loss of generality, let us assume boundedness from below. Thus, $\omega_n(\bm{k}) \geq 0$. Under the deformation, we have new particle bands $\tilde{\omega}_n(\bm{k}) \equiv \omega_n(\bm{k}) +
\gamma\mu_0(\bm{k})$.  Since
$\mu_0$ is real, odd and not identically zero, it takes both signs on the BZ,
so for $\gamma > \gamma_c$, there exist $n_c,\bm{k}_c$ such that
$\omega_n(\bm{k}_c)+\gamma\mu_0(\bm k_c) < 0$. However, by assumption, $\omega_{n_c}(-\bm k_c)\geq 0$, so $\tilde{\omega}_n(-\bm{k}_c) =\omega_{n_c}(-\bm{k}_c) + \gamma\mu_0(-\bm{k}_c)=\omega_{n_c}(-\bm{k}_c) - \gamma\mu_0(\bm{k}_c)=
\omega_{n_c}(-\bm{k}_c) + \gamma|\mu_0(\bm{k_c})| > 0$. Thus, $\tilde{\omega}_n(\bm k)$ is no longer bounded in either direction and
$H_\gamma$ is thermodynamically unstable.
\end{proof}
}

\renewcommand{\theequation}{B.\arabic{equation}}
\setcounter{equation}{0}

\section{Technical detail on the QPV covariance matrix}

\subsection{General dynamically stable, translationally invariant QBHs}
\label{App: CMfromModal}

Here, we derive Eq.\,\eqref{eq: Ck} from the main text. Since all first moments are zero for QBHs in the QPV, the CM can be expressed as
$(\mathbf{C}^{\text{qpv}}_{\bm{j},\bm{j'}})_{\ell \ell'} = \braket{\{\phi_{\bm{j},\ell} , \phi_{\bm{j'},\ell'}^\dag \}}_{\text{qpv}}$ or, equivalently, $\mathbf{C}^{\text{qpv}}_{\bm{j},\bm{j'}}=2 \braket{\mathcal{B}(\phi_{\bm j}\phi_{\bm j'}^{\dag})}_{\text{qpv}}$, 
with $\mathcal{B}(\bullet)\equiv \frac{1}{2}\left(\bullet+\boldsymbol{\tau}_1(\bullet)^T\boldsymbol{\tau}_1\right)$ defining a bosonic projection operator \cite{Squaring}.  Then, the full CM reads 
\begin{align*}
 \mb C^{\text{qpv}}
 \nonumber& = 
 2\sum_{\bm{j}\bm{j'}} \vec{e}_{\bm j}\vec{e}_{\bm j'}^{\,\dag}\otimes\mathcal{B} \left(\langle \phi_{\bm j}\phi_{\bm j'}^{\dag}\rangle\right)  \\
   \nonumber &=2\sum_{\bm{j}\bm{j'}}\vec{e}_{\bm j}\vec{e}_{\bm j'}^{\dag}\otimes\int^{\pi}_{-\pi}\frac{d^D\bm k d^D\bm k'}{(2\pi)^{2D}} e^{i(\bm k\cdot \bm j-\bm k'\cdot \bm j')}\mathcal{B}\left(\langle \phi(\bm k)\phi(\bm k')^{\dag}\rangle \right)\\
    \nonumber&=2\sum_{\bm{j}\bm{j'}}\vec{e}_{\bm j}\vec{e}_{\bm j'}^{\dag}\otimes\int^{\pi}_{-\pi}\frac{d^D\bm kd^D\bm k'}{(2\pi)^{2D}} e^{i(\bm k\cdot \bm j-\bm k'\cdot \bm j')}\mathcal{B}\left(\mb L(\bm k)\langle \psi(\bm k)\psi(\bm k')^{\dag}\rangle \mb L(\bm k')^{\dag}\right)\\
    &=\sum_{\bm{j}\bm{j'}}\vec{e}_{\bm j}\vec{e}_{\bm j'}^{\dag}\otimes\int^{\pi}_{-\pi}\frac{d^D\bm k}{(2\pi)^D} e^{i\bm k\cdot (\bm j-\bm j')}\mb L(\bm k)\mb L(\bm k)^{\dag}.
\end{align*}
Accordingly, $ \mathbf{C}^\text{qpv}(\bm k) = \mathbf{L}(\bm k)\mathbf{L}^\dag(\bm k)$. Alternatively, the elements of the CM in the quadrature representation, with $R=[x_{\bm 1,1},\ldots,p_{\bm 1,1},\ldots]^T$, are given by
\begin{align*}
\bm \Gamma^\text{qpv}= \begin{pmatrix}
   (  \mb \Gamma_{xx})_{\bm j,\bm j+\bm r}  &  (\mb \Gamma_{xp})_{\bm j,\bm j+\bm r}  \\
         (\mb \Gamma_{px})_{\bm j,\bm j+\bm r}     &  (\mb \Gamma_{pp})_{\bm j,\bm j+\bm r}
    \end{pmatrix}
    =\int^{\pi}_{-\pi}\frac{d^D\bm k}{(2\pi)^D}\, e^{i\bm k\cdot\bm r} \mb  U^{\dag}\mb L(\bm k)\mb L(\bm k)^{\dag}\mb U,
\end{align*}
with $\mb U \equiv \frac{1}{\sqrt{2}}\begin{pmatrix}
    1&i\\
    1&-i
\end{pmatrix}\otimes \mathbf{1}_d$ effecting a basis change between quadrature and bosonic degrees of freedom.

\subsection{Specialization to single-band QBHs}
\label{App: CMfromModal1}

For a system with  $d=1$, as in Appendix \ref{App: BlochMat}, we can make further headway using Eqs.\,\eqref{eq: Geig} and \eqref{eq: Gkgeneral}. Let's first assume $d_3(\bm k)>0$ \footnote{$d_3(\bm k)$ changes sign when there is a direct KC (if $d_3(\bm k)$ changes sign, it must have zero as an intermediate value. Imposing dynamical stability fixes $d_2(\bm k)$ and  $d_1(\bm k)$ to individually be zero as well. This is precisely the condition for a direct KC, see Appendix \ref{app:direct KC condition}). Thus, away from a KC, we can assume that $d_3(\bm k)$ has a definite sign.}. The eigenvectors of ${\mb g}(\bm k)$ may then be found as
\begin{align}
\label{eq: eigenvecs}
    \vec{\beta}_{+}(\bm k)\nonumber&=\frac{1}{\mathcal{N}(\bm k)}\left[d_3(\bm k )+  \mathcal{E}(\bm k)
    ,\, -d_2(\bm k)+i d_1(\bm k)\right]^T , \\
     \vec{\beta}_{-}(\bm k)&=\frac{1}{\mathcal{N}(\bm k)}\left[-d_2(\bm k)-i d_1(\bm k),\, d_3(\bm k)+ \mathcal{E}(\bm k)
     \right]^T,
\end{align}
where the normalization factor 
\begin{align*}
 \mathcal{N}(\bm k)^2 = \pm \vec{\beta}_{\pm}^{\,\dag}(\bm k)\boldsymbol \tau_3\vec{\beta}_{\pm}(\bm k)=
 2     \mathcal{E}(\bm k) (  \mathcal{E}(\bm k) + d_3(\bm k) ). 
\end{align*}
A similar derivation holds for the case $d_3(\bm k)<0$, except that $\vec{\beta}_{\pm}(\bm k)$ get exchanged, and the new normalization is given by $\mathcal{N'}(\bm k)^{2}=-\mathcal{N}(\bm k)^2$. 

Putting this all together, the real-space correlation functions in the QPV take the form: 
\begin{align}
\label{eq: FourierCoeffs}
  2\langle p_{\bm j} p_{\bm j+\bm r}\rangle
  \nonumber 
  &= \int^{\pi}_{-\pi}\frac{d^D\bm k e^{i\bm k\cdot \bm r}}{(2\pi)^D}\frac{\text{sgn}(d_3(\bm k))(d_{3}(\bm k)+d_{2}(\bm k))}{  \mathcal{E}(\bm k) }, \\
     2\langle x_{\bm j}x_{\bm j+\bm r}\rangle \nonumber&=  \int^{\pi}_{-\pi}\frac{d^D\bm k e^{i\bm k\cdot \bm r}}{(2\pi)^D}\frac{\text{sgn}(d_3(\bm k))(d_{3}(\bm k)-d_{2}(\bm k))}{  \mathcal{E}(\bm k)} ,\\
     \langle p_{\bm j}x_{\bm j+\bm r}+x_{\bm j+\bm r}p_{\bm j}\rangle& = \int^{\pi}_{-\pi}\frac{d^D \bm k e^{i\bm k\cdot \bm r}}{(2\pi)^D}\frac{\text{sgn}(d_3(\bm k))(-d_{1}(\bm k))}{  \mathcal{E}(\bm k)} .
\end{align}
Thus, the elements of $\mb\Gamma$ in real space are obtained in terms of its Fourier components as follows:
\begin{align}
\label{eq:FTCovMat}
\begin{pmatrix}
       (\mb \Gamma_{xx})_{\bm j,\bm j+\bm r}&  (\mb \Gamma_{xp})_{\bm j,\bm j+\bm r}\\
         (\mb \Gamma_{px})_{\bm j,\bm j+\bm r}     &  (\mb \Gamma_{pp})_{\bm j,\bm j+\bm r}
    \end{pmatrix}&=\int_\mathrm{BZ}
    {d^D\bm k}\, e^{i\bm k\cdot \bm r}\, \frac{\text{sgn}(d_3(\bm k))}{\mathcal{E}(\bm k)}\begin{pmatrix}
        d_3(\bm k)-d_2(\bm k)&-d_1(\bm k)\\
        -d_1(\bm k)& d_3(\bm k)+d_2(\bm k)
    \end{pmatrix},
\end{align} 
consistent with Eq.\,\eqref{GammaQPV}. 

The real-space CM may be obtained from the momentum-space expression by taking a Fourier transform. With $\boldsymbol{ \mathcal{E}}$ defined as in Eq.\,\eqref{eq:WolfCovMat} of the main text, the result is
\begin{align}
\label{eq: CovMatCompact}
    \mb \Gamma^\text{qpv}
    &=\begin{pmatrix}
       \text{Re}(\mb K-\mb \Delta)&-\text{Im}(\boldsymbol{\Delta})\\
       -  \text{Im}(\boldsymbol{\Delta})&\text{Re}(\mb K+\mb \Delta)  
    \end{pmatrix}\boldsymbol{ \mathcal{E}}^{-1}
    = \begin{pmatrix}
       \mb H_{pp}&- \frac{\mb H_{xp}+\mb H_{xp}^T}{2}\\
       - \frac{\mb H_{xp}+\mb H_{xp}^T}{2}&\mb H_{xx}
    \end{pmatrix}\boldsymbol{ \mathcal{E}}^{-1}. 
\end{align}
To see how Eq.\,\eqref{eq: CovMatCompact} is obtained, once the Fourier transform in Eq.\,\eqref{eq:FTCovMat} is explicitly carried out, consider first $D=1$ and define: 
\begin{align}
   \widehat{ \text{Re}(\mb K)}\equiv \sum_{r=-R}^R\text{Re}\left(\mb K_{r }\right)e^{ikr}\equiv d_{3}(k)=d_{30}+\sum\limits_{r=1}^Rd_{3r}\cos(kr)\equiv K(k).
\end{align}
Here $\text{Re}\left(\mb K_{r}\right)$ denote elements of $\text{Re}\left(\mb K\right)$ on the $r$'th diagonal as $\mb K$ and $\mb \Delta$ take the form of a Laurent operator:
\begin{align}
\label{eq:KandD}
 \mb K=\begin{pmatrix}
    \ddots &\ddots &\ddots&\\
    \ddots &\mb K_0 &\mb K_1&\ddots\\
    \ddots&\mb K_{-1}&\,\mb K_0&\ddots\\
    &\ddots&\ddots&\ddots
\end{pmatrix} ,  
\quad \mb \Delta=\begin{pmatrix}
    \ddots &\ddots &\ddots&\\
    \ddots &\mb \Delta_0 &\mb \Delta_1&\ddots\\
    \ddots&\mb \Delta_{-1}&\,\mb \Delta_0&\ddots\\
    &\ddots&\ddots&\ddots
\end{pmatrix}. 
\end{align} 
We analogously define the Fourier counterparts for other Laurent operators in Eqs.\,\eqref{eq: KDeltaCirculant}-\eqref{eq: KDeltaCirculant2}. Let $\mb C_1$ and $\mb C_2$ denote Laurent operators. Then $\widehat{\mb C_1\mb C_2}=\widehat{\mb C_1}\widehat{\mb C_2}$ and $\widehat{\mb C_1+\mb C_2}=\widehat{\mb C_1}+\widehat{\mb C_2}$. Therefore, e.g., $\widehat{\text{Re}(\mb K+\mb \Delta)\boldsymbol{\mathcal{E}}^{-1}}=\frac{d_3(k)+d_2(k)}{\mathcal{E}(k) }.$ Following this through with the other blocks of the CM, Eq.\,\eqref{eq: CovMatCompact} can be obtained. 

For $D>1$, while $\text{Re}(\mb K)$ is no longer a Laurent operator, it has the following structure:
\begin{align}
\label{eq:BlockStructure}
    \text{Re}(\mb K)&=\sum_{|\bm r|\equiv |(r_1,\ldots, r_D)^T|<R}d_{3\bm r}\mb V^{ r_1}\otimes\cdots\otimes \mb V^{ r_D},
\end{align}
with an analogous structure for other blocks. After performing the Fourier transformation, the above yields 
\begin{align*}
    \widehat{\text{Re}(\mb K)}&=\sum_{|\bm r|<R}d_{3\bm r}e^{i\bm k\cdot \bm r}\equiv d_3(\bm k).
\end{align*}
Furthermore, due to the Kronecker-product structure, two operators of the form $\mb K_1=\sum_{\bm r}d_{1\bm r}\mb V^{r_1}\otimes\cdots\otimes \mb V^{ r_D}$ and $\mb K_2=\sum_{\bm r'}d_{2\bm r'}\mb V^{r_1'}\otimes\cdots\otimes \mb V^{r_D'}$ commute. Therefore, the above results also hold for $D>1$.

\subsection{Validity of the covariance matrix for a Krein-gapped QBH}
\label{app:ValidCovMat}

In order for $\mb \Gamma$ to be a valid CM of a pure Gaussian state, it needs to be (i) positive-definite; and (ii) satisfy the purity condition in Eq.\,\eqref{CovMatPurity}. Here, we show that as long as the Krein gap is open, $\mb \Gamma^\text{qpv}$ given by Eq.\,\eqref{eq: CovMatCompact} is a valid CM. 

First, we note that since the blocks of the CM have the structure analogous to one given in Eq.\,\eqref{eq:BlockStructure}, they all commute and, for a given block $\mb C$, its spectrum is given by $\sigma(\mb C)=\underset{\bm k\in \text{BZ}}{\cup}C(\bm k)$, with $C(\bm k)\equiv\widehat{\mb C}$. 
An open Krein gap imposes $\boldsymbol{\mathcal{E}}> 0$. Without loss of generality, we can further assume that $d_3(\bm k)>0$. Therefore, $\text{Re}(\mb K+\mb \Delta)>0$ and $\boldsymbol{ \mathcal{E}}^{-1}\text{Re}(\mb K+\mb \Delta)>0$, since $d_3(\bm k)>0, d_3^2(\bm k)>d_2^2(\bm k)+d_1^2(\bm k)$, and thus $d_3(\bm k)+d_2(\bm k)>0, \forall \bm k$. Now, a $2\times 2$ block matrix $\begin{pmatrix}
    \mb A&\mb B\\
    \mb B^T&\mb D
\end{pmatrix}$ 
is positive-definite if $\mb D>0$ and $\mb A-\mb B\mb D^{-1}\mb B^T>0$. Since all of the blocks commute, this condition translates into $\boldsymbol{ \mathcal{E}}^{-1}\text{Re}(\mb K+\mb \Delta)>0$ and $\boldsymbol{ \mathcal{E}}^{-1}\text{Re}(\mb K+\mb \Delta)\text{Re}(\mb K-\mb \Delta)-\boldsymbol{\mathcal{E}}^{-1}\text{Im}(\mb \Delta)^2>0$, which is automatically satisfied, as we are assuming $\boldsymbol{\mathcal{E}}>0$. We further check that the state is a pure Gaussian state by confirming that
$$(\boldsymbol{\Sigma}\mb \Gamma^\text{qpv})^2= \boldsymbol{\mathcal{E}}^{-2} \begin{pmatrix}
\text{Im}(\mb \Delta
)^2-\text{Re}(\mb K+\mb \Delta)\text{Re}(\mb K-\mb \Delta)& 0 \\
0&\hspace*{-4mm}
\text{Im}(\mb \Delta
)^2-\text{Re}(\mb K+\mb \Delta)\text{Re}(\mb K-\mb \Delta)
\end{pmatrix}=-\mb 1.$$

\renewcommand{\theequation}{C.\arabic{equation}}
\setcounter{equation}{0}

\section{Proof of main theorems}
\label{app: UniqueVac}

\subsection{Two preliminary lemmas}

\begin{lemma}
\label{lem: vacnum}
Let $\ket{\tilde{0}}$ be a QPV and $V$ a (homogeneous) Gaussian unitary transformation. Let $\{\beta_n\}$ be the complete set of bosonic annihilation operators satisfying $\beta_n\ket{\tilde{0}} = 0$ . Then $V\ket{\tilde{0}} = e^{i\theta}\ket{\tilde{0}}$ if and only if $V$ commutes with $N_\text{qp}\equiv \sum_n \beta_n^\dag \beta_n$. Equivalently, $V\ket{\tilde{0}}$ is a distinct quantum state from $\ket{\tilde{0}}$ if and only if $V$ does not preserve the total particle number. 
\end{lemma}

\begin{proof}
Suppose first that $V$ commutes with the operator $N_\text{qp}$.
It follows that
\begin{align*}
       \beta_n' = V\beta_n V^\dag = \sum_m \mathbf{v}_{nm}\beta_m ,
 \end{align*}
where $\mathbf{v}_{nm}$ is a unitary matrix. These new quasiparticle $\{\beta_n'\}$ can be seen to annihilate $\ket{\tilde{0}}$. On the other hand, they annihilate $V\ket{\tilde{0}}$ as well.  Since a set of operators cannot annihilate two distinct QPVs, $V\ket{\tilde{0}} = e^{i\theta}\ket{\tilde{0}}$. 

Conversely, suppose that $V\ket{\tilde{0}} = e^{i\theta}\ket{\tilde{0}}$. Then $V\ket{\tilde{0}}$ is annihilated by all of the $\beta_n$ operators. Since $V$ is a Gaussian unitary, we have the identity
 $$   \beta_n' = V\beta_n V^\dag = \sum_m \mathbf{v}^{(1)}_{nm}\beta_m +\mathbf{v}^{(2)}_{nm}\beta_m^\dag.$$
These operators must then annihilate $\ket{\tilde{0}}$. By a direct check,
   $    \beta_n'\ket{\tilde{0}}=\sum_m \mathbf{v}^{(2)}_{nm}\ket{\tilde{1}_m}, $
where $\ket{\tilde{1}_m} = \beta_m^\dag \ket{\tilde{0}}$ are the single quasiparticle states. Since these states are linearly independent, this expression can only vanish if $\mathbf{v}^{(2)}_{nm}=0$ for all $n,m$. Thus, $V$ must commute with $N_\text{qp}$.
\end{proof}

\begin{lemma}
\label{lem: TIvac} 
If a QBH $H$ can be written in the form 
\begin{align}
\label{eq: diagTIQBH}
    H = \frac{1}{V}\int_\text{BZ} d^D\bm{k} \sum_{n}\omega_n(\bm{k}) \beta_n^\dag(\bm{k}) \beta_n(\bm{k}) + 
    \Lambda, 
\end{align}
with $\{\beta_n(\bm{k})\}$ a complete set of bosonic quasiparticles and $\Lambda$ a (potentially infinite) constant, then $H$ has a translationally invariant QPV. 
\end{lemma}

\begin{proof}
Consider the state satisfying  $\beta_n(\bm{k}) \ket{\bar{0}} = 0$ for all $n$ and $\bm{k}$. This uniquely specifies $\ket{\bar{0}}$ up to a phase. Let $S_{\bm{r}}$ represent the unitary translation operator by the lattice vector $\bm{r}$, i.e., the unique (Gaussian) unitary satisfying $S_{\bm{r}}^\dag a_{\bm{j},\ell} S_{\bm{r}} = a_{\bm{j}+\bm{r},\ell}$. The quasiparticles transform under translation according to $S_{\bm{r}}^\dag \beta_n(\bm{k}) S_{\bm{r}} = e^{i\bm{k}\cdot \bm{r}} \beta_n(\bm{k})$. Therefore, we have 
\begin{align*}
    \beta_n(\bm{k}) S_{\bm{r}}\ket{\bar{0}} =  S_{\bm{r}} e^{i\bm{k}\cdot\bm{r}}\beta_n(\bm{k}) \ket{\bar{0}} = 0.
\end{align*}
Thus, for all $\bm{r}$, the state $S_{\bm{r}}\ket{\bar{0}}$  is annihilated by all of the $\beta_n(\bm{k})$ and so it can only differ from $\ket{\bar{0}}$ by a phase. That is, $\ket{\bar{0}}$ is translationally-invariant.
\end{proof}

\subsection{Statement and proof of Theorem \ref{thm: UniqueVac} in the zero-dimensional case}
\label{app: zerodimcase}

Before directly proving Theorem \ref{thm: UniqueVac}, we formally consider $D=0$ and $d<\infty$ arbitrarily (but quadratically) coupled bosonic modes $\{a_1,\ldots,a_d\}$. In this scenario, the direct and indirect Krein gaps are identical, since $\bm{k}$ and $\bm{k'}$ are effectively removed from the definitions in Eqs.\,\eqref{eq: IndKreinGap}-\eqref{eq: KreinGap}. Under the assumptions of the theorem, such a QBH can be diagonalized as
\begin{align*}
                     { H = \sum_{n=1}^d \omega_n\bigg( \beta_n^\dag \beta_n +\frac
    {1}{2}\bigg)},
\end{align*}
where $\omega_n$ are the quasiparticle energies and $\beta_n$ and $\beta_n^\dag$ are the bosonic quasiparticle annihilation and creation operators, respectively. Moreover, $H$ has at least one QPV $\ket{\tilde{0}}$ defined by $\beta_n\ket{\tilde{0}} = 0$ for all $n$. Equivalently, $\ket{\tilde{0}} = U \ket{0}$ where $U$ is the Gaussian unitary implementing $Ua_n U^\dag = \beta_n$.

We will now prove the zero-dimensional statement of the theorem: {\em If $H$ lacks KCs, then its QPV is unique.} 

\begin{proof}
We proceed via the contrapositive: if the QPV is not unique, then $H$ must have a KC (i.e., an $n_0$ and $m_0$ such that $\omega_{n_0}=-\omega_{m_0}$). Nonuniqueness of the QPV means there are two QPV eigenstates $\ket{\tilde{0}}$ and $\ket{\tilde{0}'} = V\ket{\tilde{0}}$, where $V$ is some Gaussian unitary. Since these states must be distinct, $V$ cannot commute with the quasiparticle number $N_\text{qp}$. It follows that, for at least one quasiparticle energy $\omega_n$, there is an operator $\beta'$ satisfying
\begin{align*}
   &\text{(i): } [H,\beta'] = - \omega_n \beta',
    \\
    &\text{(ii): } \beta' = \sum_{m=1}^d c_{m}\beta_m + d_{m}\beta_m^\dag,\text{ with not all }d_{m} = 0.
\end{align*}
To see why, note that $H$ must have a set of normal modes that annihilate $\ket{\tilde{0}'}$ (as it is a QPV eigenstate of $H$). Thus, for every $\omega_n$, there is an operator that both satisfies (i) and annihilates $\ket{\tilde{0}'}$. Since $\ket{\tilde{0}'}$ is obtained from $\ket{\tilde{0}}$ by a Gaussian unitary that does not conserve $N_\text{qp}$, at least one such normal mode must have weight on the creation operators $\{\beta_n^\dag\}$ (otherwise the entire set, and thus $\ket{\tilde{0}'}$, could be obtained via a number-preserving Gaussian unitary). Now, plugging (ii) into the left hand side of (i), we obtain the system of equations
\begin{align}
   \sum_{m=1}^d \left( -\omega_m c_{m}\beta_m + \omega_m d_{m}\beta_m^\dag\right) = - \omega_n\sum_{m=1}^d (c_m\beta_m+ d_m  \beta_m^\dag)  .
\end{align}
This is equivalent to the linear equation
\begin{align*}
    \begin{pmatrix}
        \omega_1 &  &  & & &0
        \\
         & \ddots &  & & &
         \\
         & & \omega_d & & &
         \\
         & & & -\omega_1 &  &
         \\
         &  & & & \ddots &
         \\
        0& & & & & -\omega_d
    \end{pmatrix}\begin{pmatrix}
        c_1
        \\
        \vdots
        \\
        c_d
        \\
        d_1
        \\
        \vdots
        \\
        d_d
    \end{pmatrix} = \omega_n\begin{pmatrix}
        c_1
        \\
        \vdots
        \\
        c_d
        \\
        d_1
        \\
        \vdots
        \\
        d_d
    \end{pmatrix}.
\end{align*}
This is just an eigenvalue equation for the matrix on the left. However, we also know that at least one $d_m$ is not zero. Since the matrix is diagonal, this means at least one of $\{-\omega_1,\ldots,-\omega_d\}$ is equal to $\omega_n$, and so we are done. 
\end{proof}

It is worth noting that, in this zero-dimensional case, the converse can be proven as well: {\em If $H$ has a unique QPV, then it must lack KCs.} 

\noindent 
\begin{proof} 
Again, we proceed via the contrapositive: if $H$ has a KC, then it must have multiple distinct QPVs. Let $(n_0,m_0)$ be the quasiparticle modes such that $\omega_{n_0} +\omega_{m_0}=0$. Then $H$ can be rewritten as
\begin{align*}
    H = \frac{\omega_{n_0}}{2}\left( \beta_{n_0}^\dag \beta_{n_0}-\beta_{m_0}^\dag \beta_{m_0}+\beta_{n_0} \beta_{n_0}^\dag-\beta_{m_0} \beta_{m_0}^\dag \right) + \sum_{n\neq n_0,m_0}                   {\omega_n\bigg( \beta_n^\dag \beta_n +\frac
    {1}{2}\bigg)}.
\end{align*}
The first term in parentheses can be seen to be invariant under the family of Gaussian unitary transformations
\begin{subequations}
\label{eq: zerodimsqueeze}
\begin{align}
    V(r,\theta): \,\,&\beta_{n_0} \mapsto \beta_{n_0}' = \cosh(r)\beta_{n_0} + e^{i\theta}\sinh(r)\beta_{m_0}^\dag,
    \\
    V(r,\theta): \,\,&\beta_{m_0} \mapsto \beta_{m_0}' = \cosh(r)\beta_{m_0} + e^{i\theta}\sinh(r)\beta_{n_0}^\dag,
    \\
    V(r,\theta): \,\,&\beta_n \mapsto \beta_n,\quad n\neq n_0,m_0,
\end{align}
\end{subequations}
where $r,\theta\in\mathbb{R}$ are arbitrary. These (two-mode, if $n_0\neq m_0$, and single-mode if $n_0=m_0$) squeezing transformations manifestly do not conserve quasiparticle number. By Lemma \ref{lem: vacnum}, $V(r,\theta)\ket{\tilde{0}}$ is a distinct QPV from $\ket{\tilde{0}}$. Moreover, since $[H,V(r,\theta)]=0$, $V(r,\theta)\ket{\tilde{0}}$ is an eigenstate of $H$, so  we are done. \end{proof}

\subsection{Proof of Theorem \ref{thm: UniqueVac} in the general case}

\begin{proof}
To prove (i), note that, by the assumptions of the theorem, the QBH can be diagonalized as in Eq.\,\eqref{eq: diagTIQBH}. The conclusion then follows from the same contrapositive approach used in Appendix \ref{app: zerodimcase}, since translational invariance is not essential.

To prove (ii), we will again prove the contrapositive: if there are multiple distinct translationally invariant QPVs, then the direct Krein gap must be closed. Suppose there are two translationally invariant QPVs $\ket{\bar{0}}$ and $\ket{\bar{0}'}$. Since these are both QPVs and both translationally invariant, there must be a translation Gaussian unitary $U$ such that $\ket{\bar{0}'} = U\ket{\bar{0}}$. Translational invariance of $U$ means that $[U,S_{\bm{r}}]=0$ for all $\bm{r}$. The quasiparticles annihilating $\ket{\bar{0}'}$ are given by $\beta_n'(\bm{k}) = U\beta_n(\bm{k})U^\dag$. Thanks to translational  invariance of $U$, these new quasiparticles must transform identically under translations as the original set. Moreover, since $\ket{\bar{0}}$ and $\ket{\bar{0}'}$ are distinct, there must exist an $n=n_0$ and a $\bm{k}=\bm{k}_0$ such that
\begin{align}
\label{eq: betaks}
    \beta_{n_0}'(\bm{k}_0) = \sum_{m=1}^dc_{m}\beta_m(\bm{k}_0) + d_{m}\beta_m^\dag(-\bm{k}_0),\quad \text{ with not all }d_{m} = 0.
\end{align}
Now, we proceed exactly as in the zero-dimensional case to obtain the existence of an energy $-\omega_{m}(-\bm{k}_0)$ equal to $\omega_{n_0}(\bm{k}_0)$. Thus, the direct Krein gap is closed. We remark that translational invariance is essential here for restricting the right hand-side of Eq.\,\eqref{eq: betaks} to a combination of $\beta_m(\bm{k}_0)$ and $\beta_m^\dag(-\bm{k}_0)$, in order for both sides to transform identically under translation.
\end{proof}

\subsection{Remarks on the converse of Theorem \ref{thm: UniqueVac}} 
\label{app: converses}

It is natural to wonder whether, as in the $D=0$ case, necessity of the conditions (i) and (ii) also hold, namely, whether the converse statements are true:
\begin{itemize}
    \item[(i$'$)] If there exists a unique QPV (which then must be translationally invariant, by Lemma \ref{lem: TIvac}), then the indirect Krein gap must be open. 
    \item[(ii$'$)] If there exists a unique translationally invariant QPV, then the direct Krein gap must be open.
\end{itemize}
To establish a proof, it is tempting to simply import the squeezing transformations Eqs.\,\eqref{eq: zerodimsqueeze}. Below, we attempt that (in just the case of (ii$'$), for simplicity), and discuss potential issues.

Assume that the direct Krein gap closes at the point $\bm{k}_0$ between bands $n_0$ and $m_0$ and denote by $\ket{\bar{0}}$ the translationally invariant vacuum annihilated by all of the $\beta_n(\bm{k})$. Now, define the Gaussian unitary
\begin{subequations}
\label{eq: Dsqueeze}
    \begin{align}
    V(r,\theta): \,\,&\beta_{n_0}(\bm{k}_0) \mapsto \beta_{n_0}'(\bm{k}_0) = \cosh(r)\beta_{n_0}(\bm{k}_0)+ e^{i\theta}\sinh(r)\beta_{m_0}^\dag(-\bm{k}_0),
    \\
    V(r,\theta): \,\,&\beta_{m_0}(\bm{k}_0) \mapsto \beta_{m_0}'(\bm{k}_0) = \cosh(r)\beta_{m_0}(\bm{k}_0) + e^{i\theta}\sinh(r)\beta_{n_0}^\dag(-\bm{k}_0),
    \\
    V(r,\theta): \,\,&\beta_n(\bm{k}) \mapsto \beta_n(\bm{k}),\quad n\neq n_0,m_0, \,\, \bm{k}\neq \bm{k}_0.
\end{align}
\end{subequations}

As in the proof in Appendix \ref{app: zerodimcase}, the Hamiltonian can be shown to be invariant under this transformation thanks to the condition $\omega_{n_0}(\bm{k}_0) = -\omega_{m_0}(-\bm{k}_0)$. Moreover, $V(r,\theta)$ leaves the transformation properties of the quasiparticles under translations invariant. Once again, as in the proof in Appendix \ref{app: zerodimcase}, since $V(r,\theta)$ does not commute with quasiparticle number, $V(r,\theta)\ket{\bar{0}}$ is distinct from $\ket{\bar{0}}$. Thus, we have, in principle, constructed a translationally invariant QPV that is distinct from the original, completing this direction of the proof. 

The problem with this ``proof" is the following: the transformation $V(r,\theta)$ in Eq.\,\eqref{eq: Dsqueeze}, which may be decomposed as $\bigoplus_{\bm{k}}V_{\bm{k}}(r,\theta)$, is not continuous in $\bm{k}$. In fact, since this transformation only alters correlations at a single $\bm{k}$-point, it  leaves the real-space correlations unchanged in the infinite-size limit. This follows because altering the value of a function at one point (or, in general, over a set of measure zero) cannot change its Fourier transform. If the real-space correlations are unchanged, and the state remains Gaussian, it must be the same physical state. Moreover, one can show that, by considering a finite-torus geometry with $N$ sites, the real-space correlations of the original QPV and the squeezed QPV constructed using Eqs.\,\eqref{eq: Dsqueeze} differ by a quantity that vanishes as $N\to\infty$. It is possible that this issue is related to the breakdown of the Stone-von Neumann theorem (i.e., uniqueness of representations of the CCRs) for an infinite number of degrees of freedom. As it stands, it is unknown whether or not a direct KC}                   {for a set of measure zero momenta} is sufficient for the existence of non-unique vacua in $D>0$ systems, leaving open an interesting future direction for research.

                  {However, suppose that $\omega_{n_0}(\bm k) = -\omega_{m_0}(-\bm k)$ for \emph{all} $\bm k \in$ BZ. This is the case, for instance, whenever the QBH admits a symmetry that commutes with $\mathbf{g}(\bm k)$ while anticommuting with $\boldsymbol{\tau}_3$, since such a symmetry pairs each band with a partner of opposite Krein signature at the same energy. The transformations of Eq.~\eqref{eq: Dsqueeze} can then be applied for every momentum value. The momentum-space CM of the squeezed state will now differ from the original one on a set of non-zero measure, and
so will the real-space correlations. The translationally invariant QPV is therefore non-unique in the thermodynamic limit, establishing statement (ii$'$) in this case. A QBH realizing precisely this -- namely, a bosonic SSH-like chain with broken particle-number conservation, and whose bands are twofold degenerate with KCs at every $\bm k$ point -- is analyzed in Ref.~\cite{MariamNext}.}

\renewcommand{\theequation}{D.\arabic{equation} }
\setcounter{equation}{0}

\section{QPV correlation functions for illustrative models}

\subsection{Harmonic-chain and imaginary-hopping models}
\label{App: GHCCovMat}

The dynamical matrix for the Hamiltonian in Eq.\,\eqref{eq:GHC} takes the form 
\begin{align*}
    \mb g(k)&=\begin{pmatrix}
        \Omega-J\cos (k) &-J\cos (k)\\
        J\cos (k)&-\Omega+J\cos (k)
    \end{pmatrix},
\end{align*}
with eigenvalues equal to $\pm \sqrt{\Omega(\Omega-2J\cos(k))}.$ The $\boldsymbol \tau_3$-normalized eigenvectors are:
\begin{align*}
    \vec{\beta}_{+}(k)\nonumber&=\frac{1}{\mathcal{N}(k)}\left[\Omega-J\cos(k)+\sqrt{\Omega(\Omega-2J\cos(k))},J\cos(k)\right]^T,\\
    \vec{\beta}_{-}(k)&=\frac{1}{\mathcal{N}(k)}\left[J\cos(k), \Omega-J\cos(k)+\sqrt{\Omega(\Omega-2J\cos(k))}\right]^T, \\
\mathcal{N}(k) &= \sqrt{\Big(\Omega-J\cos(k)+\sqrt{\Omega(\Omega-2J\cos(k))}\Big)^2-J^2\cos^2(k)}.    
\end{align*}
As $\Omega\rightarrow2J$, $\mathcal{N}(k)\rightarrow0$, and both the eigenvalues and eigenvectors become degenerate at $k=0$: that is, the system approaches an EP.

Similarly, the dynamical matrix for the modified chain with imaginary hopping in Eq.\,\eqref{eq:GHCwImHopHam} is given by
\begin{align*}
    \mb g(k)&=\begin{pmatrix}
        \Omega-J\cos(k)+\gamma \sin(k)&-J\cos(k)\\
        J\cos(k)&-\Omega+J\cos(k)+\gamma \sin(k)
    \end{pmatrix},
\end{align*}
with eigenvalues equal to $\gamma \sin(k)\pm \sqrt{\Omega(\Omega-2J\cos(k))}$. Furthermore,
$$\sigma(\boldsymbol{\tau}_3\mb g(k))=\Omega-J\cos(k) \pm \sqrt{J^2\cos ^2k+\gamma^2\sin^2k}.$$ 
Thus, dynamical stability is retained as long as $\Omega>2J>0$, while thermodynamic stability holds up to 
$$\gamma^2=\gamma_c^2 \equiv  \tfrac{1}{2} \Omega^2 \big(1                  {+}\sqrt{1-(2J/\Omega )^2} \big).$$

\subsection{Interpolation model}
\label{App: InterpCovMat}

Expressed in terms of bosonic operators, the Hamiltonian given by Eq.\,\eqref{eq: InterpHam} takes the form 
\begin{align*}
        H=\frac{1}{2}\sum\limits_{j\in \mathbb{Z}}\left[ (1-s) \, \Omega
        (a_j^{\dag}a_j+a_ja_j^{\dag})+s\left(\Delta(ia_{j+1}^{\dag}a_j^{\dag}+\text{H.c.})+J (ia_{j+1}^{\dag}a_j+\text{H.c.})\right)\right] .
        \end{align*}
The corresponding dynamical matrix reads
\begin{align*}
       \mb g(k)&=\begin{pmatrix}
      \Omega (1-s)+J s\sin (k)& i\Delta s \cos (k)\\
     i s\Delta \cos (k)  &- \Omega (1-s)+J s\sin (k)
    \end{pmatrix} ,
\end{align*}
with eigenvalues equal to 
$s J \sin (k) \pm \sqrt{ \Omega^2 (1-s)^2-s^2\Delta^2\cos^2(k)}$,
while the normalized eigenvectors read:
\begin{align*}
   \vec{\beta}_{+}(k)\nonumber&= \frac{1}{\mathcal{N}(k)}\left[\Omega(1-s)+\sqrt{\Omega^2(1-s)^2-s^2\Delta^2\cos^2(k)},\,  i\Delta s \cos(k)\right]^T,\\
    \vec{\beta}_{-}(k)&= \frac{1}{\mathcal{N}(k)}\left[-i\Delta s \cos(k),\,  \Omega(1-s)+\sqrt{\Omega^2(1-s)^2-s^2\Delta^2\cos^2(k)}\right]^T, \\
\mathcal{N}({k})& = \sqrt{2\left(\Omega^2(1-s)^2-s^2\Delta^2\cos^2(k)+\Omega(1-s)\sqrt{\Omega^2(1-s)^2-s^2\Delta^2\cos^2(k)}\right)}.    
\end{align*}

From the above expressions, and using Eqs.\,\eqref{dynStab} and \eqref{thStab}, it become possible to determine the stability phase diagram. In particular, the eigenvalues of $\mb g(k)$ are real up to the value $s\leq s_2=\frac{1}{1+\Delta/\Omega}$ given in Eq.\,\eqref{special}. At $s=s_2$, the eigenvalues and eigenvectors of $\mb g(k)$ become degenerate at $k=0,\pm \pi$, marking the loss of dynamical stability associated with the presence of an EP.

By examining instead the eigenvalues 
$$\sigma\left(\boldsymbol{\tau}_3\mb g(k)\right)=\Omega(1-s)\pm \sqrt{(J s \sin(k))^2+(\Delta s \cos(k))^2},$$ 
we determine that the model is thermodynamically stable for $s\leq s_1=\frac{1}{1+J/\Omega},$ as stated in the main text.

\subsection{Double harmonic-chain model}
\label{App: Double}

The dynamical matrix for the Hamiltonian in Eq.\,\eqref{eq:multicritHam} has the form 
\begin{align*}
    \mb g(k)=\begin{pmatrix}
       \Omega_1+\Omega_2 +(K_1+K_2)(1- \cos (k))&\Omega_2-\Omega_1 +(K_2-K_1)(1- \cos (k))\\
      -\Big( \Omega_2-\Omega_1 +(K_2-K_1)(1- \cos (k))\Big)&-\Big( \Omega_1+\Omega_2 +(K_1+K_2)(1- \cos (k))\Big)
    \end{pmatrix}.
\end{align*}
The eigenvalues read 
$$\pm \sqrt{4(\Omega_1+K_1(1-\cos(k)))(\Omega_2+K_2(1-\cos(k)))} =\pm \mathcal{E}(k),$$ 
whereas the normalized eigenvectors are:
\begin{align}
     \vec{\beta}_{+}( k)\nonumber&=\frac{1}{\mathcal{N}(k)}\left[\Omega_1+\Omega_2 +(K_1+K_2)(1- \cos (k))+  \mathcal{E}(k)
    ,\, -(\Omega_2-\Omega_1 +(K_2-K_1)(1- \cos (k)))\right]^T , \\
     \vec{\beta}_{-}( k)&\nonumber=\frac{1}{\mathcal{N}( k)}\left[-(\Omega_2-\Omega_1 +(K_2-K_1)(1- \cos (k))),\, \Omega_1+\Omega_2 +(K_1+K_2)(1- \cos (k))+ \mathcal{E}(k)
     \right]^T , \\
      \mathcal{N}(k)\nonumber&=\sqrt{2\mathcal{E}(k)(\mathcal{E}(k)+ \Omega_1+\Omega_2 +(K_1+K_2)(1- \cos (k)))}.
\end{align}

From here, it is straightforward to determine the dynamical stability boundaries as discussed in the main text, namely, EPs for both $(\Omega_1=0,\Omega_2\neq 0)$ and $(\Omega_2=0,\Omega_1\neq 0)$, and a KC at the point $\Omega_1=\Omega_2=0$.

\bibliographystyle{bibstyle_v4}
\bibliography{criticality}

\end{document}